\pdfoutput=1

\documentclass[11pt]{article}
\usepackage[letterpaper,margin=1in]{geometry} 
\usepackage{iftex}

\ifPDFTeX
\usepackage[utf8]{inputenc}
\usepackage[T1]{fontenc}
\fi

\usepackage{amsmath}
\usepackage{amssymb}
\usepackage{amsthm}
\usepackage{mathtools}
\usepackage{thmtools}
\usepackage{isomath}
\usepackage{algorithm}
\usepackage{thm-restate}
\usepackage{enumitem}
\usepackage{algpseudocode}
\usepackage{enumitem}
\usepackage{subcaption}

\usepackage{lineno}
\usepackage[disable]{todonotes}
\usepackage{changes}
\usepackage[most]{tcolorbox}

\ifPDFTeX
\usepackage{libertine}
\usepackage[libertine]{newtxmath} 
\usepackage[scaled=0.96]{zi4} 
\else
\usepackage{fontspec}
\usepackage{unicode-math}
\fi

\usepackage{microtype}

\usepackage{xcolor}
\usepackage{xspace}

\usepackage[
    style=alphabetic,
    maxalphanames=4,
    minalphanames=4,
    maxcitenames=4,
    minnames=1,
    maxbibnames=20,
    backref=true,
    doi=true,
    url=true,
    backend=biber
]{biblatex}

\ifpdf
\usepackage[ocgcolorlinks]{hyperref} 
\else
\usepackage[hidelinks]{hyperref}
\fi

\allowdisplaybreaks[1]

\makeatletter
\g@addto@macro\bfseries{\boldmath}
\makeatother

\makeatletter
\g@addto@macro\mdseries{\unboldmath}
\g@addto@macro\normalfont{\unboldmath}
\g@addto@macro\rmfamily{\unboldmath}
\g@addto@macro\upshape{\unboldmath}
\makeatother

\makeatletter
\g@addto@macro\bfseries{\boldmath}
\def\thmhead@plain#1#2#3{%
  \thmname{#1}\thmnumber{\@ifnotempty{#1}{ }\@upn{#2}}%
  \thmnote{ {\the\thm@notefont\unboldmath(#3)}}}
\let\thmhead\thmhead@plain
\makeatother

\DeclareCiteCommand{\citem}
    {}
    {\mkbibbrackets{\bibhyperref{\usebibmacro{postnote}}}}
    {\multicitedelim}
    {}

\renewcommand*{\multicitedelim}{\addcomma\space}

\AtEveryCitekey{\clearfield{note}}

\makeatletter
\AtBeginDocument{%
    \newlength{\temp@x}%
    \newlength{\temp@y}%
    \newlength{\temp@w}%
    \newlength{\temp@h}%
    \def\my@coords#1#2#3#4{%
      \setlength{\temp@x}{#1}%
      \setlength{\temp@y}{#2}%
      \setlength{\temp@w}{#3}%
      \setlength{\temp@h}{#4}%
      \adjustlengths{}%
      \my@pdfliteral{\strip@pt\temp@x\space\strip@pt\temp@y\space\strip@pt\temp@w\space\strip@pt\temp@h\space re}}%
    \ifpdf
      \typeout{In PDF mode}%
      \def\my@pdfliteral#1{\pdfliteral page{#1}}
      \def\adjustlengths{}%
    \fi
    \ifxetex
      \def\my@pdfliteral #1{}
      \def\adjustlengths{\setlength{\temp@h}{-\temp@h}\addtolength{\temp@y}{1in}\addtolength{\temp@x}{-1in}}%
    \fi%
    \def\Hy@colorlink#1{%
      \begingroup
        \ifHy@ocgcolorlinks
          \def\Hy@ocgcolor{#1}%
          \my@pdfliteral{q}%
          \my@pdfliteral{7 Tr}
        \else
          \HyColor@UseColor#1%
        \fi
    }%
    \def\Hy@endcolorlink{%
      \ifHy@ocgcolorlinks%
        \my@pdfliteral{/OC/OCPrint BDC}%
        \my@coords{0pt}{0pt}{\pdfpagewidth}{\pdfpageheight}%
        \my@pdfliteral{F}
        \my@pdfliteral{EMC/OC/OCView BDC}%
        \begingroup%
          \expandafter\HyColor@UseColor\Hy@ocgcolor%
          \my@coords{0pt}{0pt}{\pdfpagewidth}{\pdfpageheight}%
          \my@pdfliteral{F}
        \endgroup%
        \my@pdfliteral{EMC}%
        \my@pdfliteral{0 Tr}
        \my@pdfliteral{Q}%
      \fi
      \endgroup
    }%
}
\makeatother
\usepackage[capitalize]{cleveref}

\newcommand{\linecref}[1]{Line~\ref{#1}}
\newcommand{\linescref}[2]{Lines~\ref{#1} and~\ref{#2}}

\DeclareSourcemap{
	\maps[datatype=bibtex]{
		\map{
			\step[fieldset=abstract, null]
			\step[fieldset=pagetotal, null]
			\step[fieldset=timestamp, null]
			\step[fieldset=biburl, null]
			\step[fieldset=bibsource, null]
		}
		\map{
			\step[fieldsource=doi,final]
			\step[fieldset=url,null]
		}  
		\map[overwrite]{
			\pertype{inproceedings}
			\step[fieldset=editor, null]
			\step[fieldset=series, null]
			\step[fieldset=volume, null]
			\step[fieldset=publisher, null]
			\step[fieldset=isbn, null]
       }
	}
}

\newcommand{\antonis}[1]{\todo[linecolor=orange!50!black,backgroundcolor=orange!25,bordercolor=orange!50!black]{\scriptsize \textbf{AS:} #1}}

\colorlet{DarkRed}{red!50!black}
\colorlet{DarkGreen}{green!50!black}
\colorlet{DarkBlue}{blue!50!black}

\hypersetup{
    hypertexnames = false,
	linkcolor = DarkRed,
	citecolor = DarkGreen,
	urlcolor = DarkBlue,
	bookmarksnumbered = true,
	linktocpage = true
}

\declaretheorem[numberwithin=section]{theorem}
\declaretheorem[numberlike=theorem]{lemma}

\declaretheorem[numberlike=theorem]{corollary}
\declaretheorem[numberlike=theorem]{Definition}
\declaretheorem[numberlike=theorem]{claim}
\declaretheorem[numberlike=theorem]{observation}
\declaretheorem[numberlike=theorem]{definition}

\newcommand{\dist}{\operatorname{dist}}

\newcommand{\old}{\text{(old)}}

\newcommand{\ellexp}{{(\ell)}}
\newcommand{\tauexp}{{(\tau)}}
\newcommand{\ball}{\operatorname{Ball}}
\newcommand{\cost}{\operatorname{cost}}
\newcommand{\lazy}{\operatorname{Lazy}}

\newcommand{\myhs}{\hspace{0.05em}}

\DeclareMathOperator{\polylog}{polylog}
\DeclareMathOperator{\poly}{poly}

\DeclareMathOperator*{\argmin}{argmin}

\title{Adaptive Fully Dynamic $k$-Center Clustering with (Near-)Optimal Worst-Case Guarantees}

\author{
  Mara Grilnberger\thanks{Department of Computer Science, University of Salzburg, Austria. This project has received funding from the European Research Council (ERC) under the European Union's Horizon 2020 research and innovation programme (grant agreement No~947702). This work has been supported by the EXDIGIT (Excellence in Digital Sciences and Interdisciplinary Technologies) project, funded by Land Salzburg under grant number 20204-WISS/263/6-6022.} \and
  Antonis Skarlatos\thanks{Department of Computer Science, University of Warwick, Coventry, England. Antonis Skarlatos is funded by the European Union (ERC grant, DYNALP, 101170133). Views and opinions expressed are however those of the author(s) only and do not necessarily reflect those of the European Union or the European Research Council Executive Agency. Neither the European Union nor the granting authority can be held responsible for them.}
}

\date{}

\begin{document}
\maketitle
\begin{abstract}
    Given a sequence of adversarial \emph{point insertions} and \emph{point deletions}, is it possible to simultaneously optimize the \emph{approximation ratio}, \emph{update time}, and \emph{recourse} for a $k$-clustering problem? If so, can this be achieved with \emph{worst-case} guarantees against an \emph{adaptive adversary}? These questions have garnered significant attention in recent years. 
    
Prior works by Bhattacharya, Costa, Garg, Lattanzi, and Parotsidis [FOCS '24] and by Bhattacharya, Costa, and Farokhnejad [STOC '25] have taken significant steps toward this direction for the $k$-median clustering problem and its generalization, the $(k, z)$-clustering problem. In this paper, we study the $k$-center clustering problem, which is one of the most classical and well-studied $k$-clustering problems. Recently, 
Bhattacharya, Costa, Farokhnejad, Lattanzi, and Parotsidis [ICML '25] provided an affirmative answer to the first question for the $k$-center clustering problem. However, their work did not resolve the second question, as their result provides only expected amortized guarantees against an oblivious adversary.

In this work, we make significant progress and close the gap by 
answering both questions in the affirmative. Specifically, we show that the fully dynamic $k$-center clustering problem admits a constant-factor approximation, near-optimal worst-case update time, and constant worst-case recourse, even against an adaptive adversary. This is achieved by first developing a fully dynamic bicriteria approximation algorithm with (near-)optimal worst-case bounds, and then designing a suitable fully dynamic $k$-center algorithm with near-linear update time. For the fully dynamic bicriteria approximation algorithm, we establish the worst-case recourse and worst-case update time guarantees separately, and then merge them into a single algorithm through a simple yet elegant process.
\end{abstract}

\thispagestyle{empty}
\newpage
\tableofcontents
\thispagestyle{empty}
\newpage

\setcounter{page}{1}
\section{Introduction}
Clustering is a fundamental and extensively studied problem in unsupervised learning and, more broadly, in computer science. In general, clustering problems involve partitioning data points into groups (clusters) such that points within the same group are more similar to each other, based on a distance measure, than to points in other groups. Clustering has numerous applications across various domains, including machine learning, data analysis, community detection, and image segmentation~\cite{hansen1997cluster,fortunato2010community,shi2000normalized,jain2010data}. Clustering algorithms optimize a given objective function, and specifically for $k$-clustering problems
the goal is to output a set of $k$ representative points that minimize a $k$-clustering objective.

Among the most fundamental and widely studied $k$-clustering objectives is the \emph{$k$-center objective}.
Given a metric space $(\mathcal{X}, \dist)$, a set of points $P \subseteq \mathcal{X}$ with $|P| = n$, and a positive integer $k \leq n$, the goal of the $k$-center clustering problem is to output a subset $S \subseteq P$ of at most $k$ points, called \emph{centers}, such that the maximum distance from any point in $P$ to its closest center in $S$ is minimized. The $k$-center clustering problem admits  polynomial-time $2$-approximation algorithms~\cite{Gonzalez85, HochbaumS86}, and it is known to be NP-hard to approximate the $k$-center objective within a factor of $(2-\epsilon)$ for any constant $\epsilon > 0$~\cite{hsu1979easy}. 
Another line of work investigates the $k$-center clustering problem in the graph setting, where the metric is induced by the pairwise shortest-path distances in a weighted undirected graph~\cite{Thorup04, EppsteinHS20, AbboudCLM23}.

In recent years, the rapid growth of data has spurred notable interest in developing dynamic clustering algorithms~\cite{Cohen-AddadSS16, GoranciHL18, HenzingerK20, GoranciHLSS21, HenzingerLM20, ChanGS18, BateniEFHJMW23, PellizzoniPP25, incrementalkzclusteringgraphs}. In the dynamic setting, the point set $P$ undergoes adversarial point updates. 
The setting is called \emph{fully dynamic} when it involves adversarial point insertions and point deletions, \emph{incremental} when it involves only adversarial point insertions, and
\emph{decremental} when it involves only adversarial point deletions. 
An \emph{oblivious adversary} fixes the entire sequence of updates before the algorithm begins. Namely, the adversary cannot adapt the updates based on the choices of the algorithm during the execution. This contrasts with the \emph{adaptive adversary}, who has access
to the output of the algorithm and can adapt the updates accordingly. A \emph{stronger} adaptive adversary can also be considered, one that in addition has knowledge of the internal random decisions of the algorithm.

\paragraph{Adaptiveness with worst-case guarantees.} 
For the $k$-center clustering problem, the work in~\cite{BateniEFHJMW23} investigates the distinction between oblivious and adaptive adversaries.
Dynamic algorithms that are robust against adaptive adversaries can be used as black-box subroutines; the authors in~\cite{banihashem2025dynamicdiameterhighdimensionsadaptive} provide an extensive discussion of the advantages of adaptive adversaries over oblivious ones. Similar to~\cite{banihashem2025dynamicdiameterhighdimensionsadaptive}, our dynamic algorithms (including our randomized algorithms) are resilient to an adaptive adversary that not only has access to the output of the algorithm but also has full knowledge of the algorithm and access to its random choices. As explained in~\cite{banihashem2025dynamicdiameterhighdimensionsadaptive}, this notion closely reflects the \emph{white-box adversarial model} in robust machine learning and is significantly stronger than the commonly studied model.

Moreover, worst-case guarantees are crucial in real-time systems, and extensive research has been devoted to achieving them~\cite{Sankowski04,KapronKM13,NanongkaiSW17,AbrahamCK17,CharikarS18,BernsteinFH19,abs-2511-07354, abs-2511-08485,banihashem2025dynamicdiameterhighdimensionsadaptive,GrandoniSSU26}.\footnote{For example,~\cite{abs-2511-08485} (STOC '26) achieves worst-case recourse and update time for the fully dynamic set cover problem.} Therefore in the dynamic algorithms literature, fully dynamic algorithms that are robust against adaptive adversaries and provide worst-case guarantees are considered the \emph{gold standard}---a standard that has been achieved for only a few problems (see also the discussion in~\cite{BrandFN22}). 

\paragraph{The main three key aspects.}
The primary goal in our context is to maintain a set of centers by optimizing three key aspects: 
\begin{itemize}
    \item The \emph{approximation ratio}, which measures the accuracy of the maintained solution.
    \item The \emph{update time}, which measures the efficiency of the algorithm.
    \item The \emph{recourse}, defined as the number of changes made to the set of centers after an adversarial update, which quantifies the stability of the algorithm.
\end{itemize}

The line of research focusing on optimizing the recourse is commonly referred to as \emph{consistent clustering}, a concept introduced by Lattanzi and Vassilvitskii~\cite{LattanziV17}. Consistency is important from both
a theoretical and a practical point of view, and in recent
years it has attracted considerable attention~\cite{LattanziV17, fichtenberger2021consistent, Cohen-AddadHPSS19, GuoKLX21, lkacki2024fully, forster_skarlatos2025,GuoKLX20,abs-2508-10800,ChanJWZ25}.
From a theoretical perspective, minimizing the worst-case recourse is a fundamental combinatorial question, as it measures how smoothly the solutions adapt after each adversarial update. 
From a practical perspective, consider a scenario where the centers are used for a task that is costly to recompute,
and modifications to the solution incur waiting time. Amortized guarantees on the recourse can lead to occasional long waits, whereas worst-case recourse guarantees prevent significant delays.
Consistency is also closely related to the notion of low-recourse in the online algorithms literature, where the irrevocability of past choices is relaxed, permitting modifications to previous decisions~\cite{MegowSVW12, BernsteinHR19, BhattacharyaBLS23, SandersSS09, GuptaKS14}. 

Similarly to the recourse, achieving a worst-case update time is highly motivated from both theoretical and practical points of view.
The gold standard then is to obtain (near-)optimal worst-case bounds on both the update time and the recourse, with a constant-factor approximation against an adaptive adversary in the fully dynamic setting.

\subsection{Related Work}
Several works on the $k$-center clustering problem focus solely on optimizing the update time after an adversarial point update~\cite{ChanGS18, BateniEFHJMW23} and after an adversarial edge update in the graph setting~\cite{cruc_for_gor_yas_skar2024dynamic}.
Other works focus on optimizing the recourse; for instance, Lattanzi and Vassilvitskii~\cite{LattanziV17} developed an incremental constant-factor approximation algorithm for classical $k$-clustering problems (e.g., $k$-median, $k$-means, $k$-center) with $O(k^2 \log^4 n)$ total recourse. They also established an $\Omega(k \log n)$ lower bound on the total recourse. For the $k$-median clustering problem in the incremental setting, Fichtenberger, Lattanzi, Norouzi-Fard, and Svensson~\cite{fichtenberger2021consistent} improved the total recourse to $O(k \polylog n)$, which is optimal up to polylogarithmic factors. Afterwards, Chan, Jiang, Wu, and Zhao provided further improvements~\cite{ChanJWZ25}.
The works of Bhattacharya, Costa, Garg, Lattanzi, and Parotsidis~\cite{BhattacharyaCGL24} and of
Bhattacharya, Costa and Farokhnejad~\cite{BhattacharyaCF25}
developed new techniques that improved the guarantees for the fully dynamic $(k, z)$-clustering problem.

Regarding the consistent $k$-center clustering problem, the authors in~\cite{LattanziV17,forster_skarlatos2025} argue that in the incremental setting, the ``doubling algorithm'' by Charikar, Chekuri, Feder, and Motwani~\cite{CharikarCFM97} achieves the tight $O(k \log n)$ total recourse bound. The work of Łącki \textcircled{r} Haeupler \textcircled{r} Grunau \textcircled{r} Rozhoň \textcircled{r} Jayaram~\cite{lkacki2024fully} studied the fully dynamic setting by obtaining a deterministic fully dynamic algorithm with constant-factor approximation and constant worst-case recourse. Subsequently, Forster and Skarlatos~\cite{forster_skarlatos2025} 
developed deterministic constant-factor approximation algorithms with optimal worst-case recourse bounds for all three dynamic settings (i.e., fully dynamic, incremental, and decremental).

For the $k$-center clustering problem where all three key aspects are optimized, the work of Bhattacharya, Costa, Garg, Lattanzi, and Parotsidis~\cite{BhattacharyaCGL24} achieved an $O(\log n \myhs \log k)$ approximation ratio, $\tilde{O}(k)$ amortized update time, and $\tilde{O}(1)$ amortized recourse against an adaptive adversary.\footnote{The notation $\tilde{O}(\cdot)$ hides factors that are polylogarithmic in $n$ and in the aspect ratio of the metric space.} Thereafter, Bhattacharya, Costa, Farokhnejad, Lattanzi, and Parotsidis~\cite{bhattacharya2025alm_opt_kcenter} provided new bounds for the $k$-center objective, as we describe later in~\cref{thm:fully_dyn_almost_opt}.

\paragraph{Bicriteria approximation algorithm by Mettu and Plaxton.}
Our fully dynamic $k$-center clustering algorithms build upon the well-known bicriteria approximation algorithm of Mettu and Plaxton~\cite{mettu2004optimal} for the $(k,z)$-clustering problem. Mettu and Plaxton~\cite{mettu2004optimal} developed a $(O(1), O(\log n \myhs \log \frac{n}{k}))$-bicriteria approximation algorithm for the $k$-median clustering problem in the static setting. We refer to this static algorithm as the \emph{MP-bi algorithm}. As noted by Huang and Vishnoi~\cite{HuangV20}, the MP-bi algorithm can be easily generalized to the $(k, z)$-clustering problem.

A fully dynamic version of the MP-bi algorithm was developed by Bhattacharya, Costa, Lattanzi, and Parotsidis~\cite{bhattacharya2023fully} for the $(k, z)$-clustering problem, which we refer to as the \emph{dynamic MP-bi algorithm}. Recently, Bhattacharya, Costa, Farokhnejad, Lattanzi, and Parotsidis~\cite{bhattacharya2025alm_opt_kcenter} focused solely on the $k$-center clustering problem and showed that the dynamic MP-bi algorithm has constant amortized recourse, as stated in~\cref{lem:prior_bicr_alg} (Lemma 3.3 in~\cite{bhattacharya2025alm_opt_kcenter}).\footnote{The authors mention a slight modification to the dynamic MP-bi algorithm.}

\begin{lemma}[dynamic MP-bi algorithm~\cite{bhattacharya2023fully, bhattacharya2025alm_opt_kcenter}] \label{lem:prior_bicr_alg}
    There is a randomized fully dynamic algorithm against an adaptive adversary that, given a point set $P$ in a metric space subject to point updates and an integer~$k \geq 1$, maintains a subset of points $S \subseteq P$ such that:
    \begin{itemize}
        \item The set $S$ is with high probability a $(4, O(\log \frac{n}{k}))$-bicriteria approximate solution to the $k$-center clustering problem.

        \item The amortized update time is $O(k \log \frac{n}{k})$.

        \item The amortized recourse is $O(1)$ and the worst-case recourse is $\Theta(k \myhs \log \frac{n}{k})$.
    \end{itemize}
\end{lemma}

The worst-case recourse of the dynamic MP-bi algorithm is $\Theta(k  \myhs \log \frac{n}{k})$, because the execution of their procedure \texttt{ConstructFromLayer}$(0)$\footnote{The procedure \texttt{ConstructFromLayer}$(i)$ in~\cite{bhattacharya2023fully} is presented as Algorithm~4.} can produce an entirely new set $S$ of size $\Theta(k \myhs \log \frac{n}{k})$. In order to convert the bicriteria approximate solution to a $k$-center solution, the authors in~\cite{bhattacharya2025alm_opt_kcenter} modify the fully dynamic maximal independent set algorithm of Behnezhad, Derakhshan, Hajiaghayi, Stein, and Sudan~\cite{BehnezhadDHSS19} to obtain the following theorem.

\begin{theorem}[\cite{bhattacharya2025alm_opt_kcenter}]\label{thm:fully_dyn_almost_opt}
    There is a randomized fully dynamic algorithm against an oblivious adversary that, given a point set $P$ in a metric space subject to point updates and an integer $k \geq 1$, maintains an $O(1)$-approximate $k$-center solution with $\tilde{O}(k)$ expected amortized update time and $O(1)$ expected amortized recourse.
\end{theorem}

The high-level idea, utilized in~\cite{bhattacharya2025alm_opt_kcenter} and also employed in our work, is to maintain a near-linear time $k$-center algorithm on top of a sparsified bicriteria approximate solution. We note that the expected guarantees and the assumption of an oblivious adversary come from the algorithm in~\cite{BehnezhadDHSS19}. To the best of our knowledge, the work of~Bhattacharya, Costa, Farokhnejad, Lattanzi, and Parotsidis~\cite{bhattacharya2025alm_opt_kcenter} is the first to pursue the development
of a fully dynamic algorithm with (near-)optimal approximation ratio, update time, and recourse simultaneously for the $k$-center clustering problem. 

However, the fully dynamic algorithm in~\cref{thm:fully_dyn_almost_opt} assumes an oblivious adversary, and some of its guarantees hold only in expectation or in the amortized sense. For this reason,  the following important question remains open:
\begin{tcolorbox}[colback=blue!5!white, colframe=blue!40!black, coltitle=black, sharp corners=south, boxrule=0.7pt]
    \begin{center}
        Is there a fully dynamic algorithm against an adaptive adversary that maintains a constant-factor approximate solution
        for the $k$-center clustering problem, with $\tilde{O}(k)$ worst-case update time and $O(1)$ worst-case recourse? 
    \end{center}
\end{tcolorbox}

\subsection{Our Contributions}
In this work, we provide a definitive affirmative answer to this important question (in~\cref{thm:our_fully_dyn}). 
To achieve this, we develop two subroutines: a fully dynamic bicriteria approximation algorithm and a suitable fully dynamic $k$-center algorithm. The first subroutine is a fully dynamic $(8, O(\log n \myhs \log \frac{n}{k}))$-bicriteria approximation algorithm with \emph{near-optimal worst-case update time} and \emph{constant worst-case recourse}, as demonstrated in~\cref{thm:bicr_alg_merged}. 

\begin{restatable}{theorem}{bicralgmerged}
\label{thm:bicr_alg_merged}
    There is a randomized fully dynamic algorithm against an adaptive adversary that, given a point set $P$ in a metric space subject to point updates and an integer $k \geq 1$, maintains a subset of points $S \subseteq P$ such that:
    \begin{itemize}
        \item The set $S$ is with high probability a $(8, O(\log n \myhs \log \frac{n}{k}))$-bicriteria approximate solution to the $k$-center clustering problem.

        \item The worst-case update time is $O\big(k \log n \myhs (\log \frac{n}{k})^2\big)$.

        \item The worst-case recourse is $O(1)$.
    \end{itemize}
\end{restatable}
\noindent
Our fully dynamic bicriteria approximation algorithm of~\cref{thm:bicr_alg_merged} is developed in two stages.
First in~\cref{thm:bicr_alg_worst_case} we establish the worst-case recourse bounds, and then in~\cref{thm:bicr_alg_worst_case_update_time} we establish the worst-case update time bounds.
Finally, we merge the two algorithms into a single fully dynamic bicriteria approximation algorithm through a simple yet elegant process (in~\cref{sec:merged_bicr_alg}).

For the second subroutine, we introduce a crucial and non-trivial modification to the 
fully dynamic algorithm by Forster and Skarlatos~\cite{forster_skarlatos2025}, improving
the worst-case update time to $O(n \log n)$, as demonstrated in~\cref{thm:k_center_near_linear}.

\begin{restatable}{theorem}{kcenternearlinear}\label{thm:k_center_near_linear}
    There is a deterministic fully dynamic algorithm against an adaptive adversary that, given a point set $P$ in a metric space subject to point updates and an integer $k \geq 1$, maintains a subset of points $S \subseteq P$ such that:
    \begin{itemize}
        \item The set $S$ is an $O(1)$-approximate $k$-center solution.

        \item The worst-case update time is $O(n \log n)$.

        \item The worst-case recourse is $1$.
    \end{itemize}
\end{restatable}

Our fully dynamic consistent $k$-center algorithm of~\cref{thm:k_center_near_linear} achieves an optimal worst-case recourse bound and a much faster worst-case update time than prior work~\cite{lkacki2024fully,forster_skarlatos2025}.
The combination of~\cref{thm:bicr_alg_merged,thm:k_center_near_linear} yields our main result, which is~\cref{thm:our_fully_dyn}.

\begin{restatable}{theorem}{mainresult}\label{thm:our_fully_dyn}
    There is a randomized fully dynamic algorithm against an adaptive adversary that, given a point set $P$ in a metric space subject to point updates and an integer $k \geq 1$, maintains with high probability an $O(1)$-approximate $k$-center solution with $\tilde{O}(k)$ worst-case update time and $O(1)$ worst-case recourse.
\end{restatable}

To the best of our knowledge, this is the first result for any $k$-clustering objective that achieves 
near-optimal bounds on approximation ratio, update time, and recourse, with worst-case guarantees in the fully dynamic setting against an adaptive adversary. It would be very interesting to explore whether similar bounds can be obtained for other $k$-clustering objectives, such as $k$-median and $k$-means. In particular, obtaining worst-case guarantees on both the update time and recourse for these problems remains a compelling direction for future work.

\subsubsection{Further Remarks on Our Contributions} 
The work of Bateni, Esfandiari, Fichtenberger, Henzinger, Jayaram, Mirrokni, and Wiese~\cite{BateniEFHJMW23} optimizes both the approximation ratio and the update time for the $k$-center clustering problem. The observation by Bhattacharya, Costa, Farokhnejad, Lattanzi, and Parotsidis~\cite{bhattacharya2025alm_opt_kcenter} is that maintaining a bicriteria approximate solution enables achieving constant recourse, by relaxing the approximation ratio to be a constant. In order to convert efficiently the bicriteria approximate solution into a $k$-center solution, a fully dynamic algorithm with $\tilde{O}(n)$ update time and constant recourse is needed.
Forster and Skarlatos~\cite{forster_skarlatos2025} developed a deterministic algorithm that achieves a worst-case recourse of $1$, but has an update time of $O(nk)$. 

In summary, the core idea is to maintain a bicriteria approximation and then apply a fully dynamic algorithm with $\tilde{O}(n)$ update time. Our~\cref{thm:bicr_alg_merged} addresses the first part providing near-optimal worst-case guarantees, while our~\cref{thm:k_center_near_linear} handles the second part achieving $\tilde{O}(n)$ update time. Together, these results yield the desired~\cref{thm:our_fully_dyn}. Since~\cref{thm:bicr_alg_merged} and~\cref{thm:our_fully_dyn} are randomized algorithms (against an adaptive adversary), the following question naturally arises:

\begin{tcolorbox}[colback=blue!5!white, colframe=blue!40!black, coltitle=black, sharp corners=south, boxrule=0.7pt]
    \begin{center}
        Are there deterministic versions of~\cref{thm:bicr_alg_merged} and~\cref{thm:our_fully_dyn}?
    \end{center}
\end{tcolorbox}

In the next two paragraphs, we discuss this question using existing lower bounds to argue that the answer is negative. In turn, this implies that our randomized algorithms are optimal up to polylogarithmic factors.

\paragraph{Optimality under known lower bounds.}
Based on Theorem 3 in~\cite{BateniEFHJMW23}, there is no deterministic fully dynamic constant-factor approximation algorithm with $\tilde{O}(\poly(k))$ update time for the $k$-center clustering problem, for some values of $k$.\footnote{Regarding Theorem~3 in~\cite{BateniEFHJMW23}, as the authors mention, deterministic algorithms are correct even against stronger metric-adaptive adversaries.} We highlight that a corollary is that there is no deterministic fully dynamic $(O(1), \tilde{O}(\poly(k)))$-bicriteria approximation algorithm with $\tilde{O}(\poly(k))$ update time for the $k$-center clustering problem. The reason is that any $(O(1), \tilde{O}(\poly(k)))$-bicriteria approximate solution could be combined with a deterministic static $k$-center algorithm~\cite{Gonzalez85, HochbaumS86} (or with~\cref{thm:k_center_near_linear}) to maintain a fully dynamic $O(1)$-approximate $k$-center solution in $\tilde{O}(\poly(k))$ update time, which would yield a contradiction. Consequently, both of our \emph{randomized} fully dynamic algorithms
in~\cref{thm:bicr_alg_merged} and~\cref{thm:our_fully_dyn}---which are robust against an adaptive adversary---are the best achievable up to polylogarithmic factors. 

\paragraph{The $(k, z)$-clustering problem.}
Our fully dynamic bicriteria approximation algorithm of~\cref{thm:bicr_alg_merged} generalizes to the $(k, z)$-clustering problem (which includes $k$-median and $k$-means).\antonis{Elaborate in Appendix.} Hence, any constant-factor approximation $(k, z)$-clustering algorithm with worst-case guarantees and update time of~$\tilde{O}(n)$ can use our bicriteria approximation algorithm as a subroutine to achieve near-optimal guarantees.

The lower bound given by Costa and Farokhnejad \cite{CostaF25} states that there is no static deterministic $k$-median algorithm with constant-factor approximation that runs in $\tilde{O}(nk)$ time, for $k \leq n^{o(1)}$. We observe that an easy corollary is that the same lower bound applies to bicriteria approximation algorithms that produce solutions of size $\tilde{O}(k)$. The reason is that any $(O(1), \tilde{O}(k))$-bicriteria approximate solution could be used within a deterministic $k$-median algorithm with $\tilde{O}(n^2)$ time (e.g.,\cite{MettuP03online,JainV01}) to produce an $O(1)$-approximate $k$-median solution in time $\tilde{O}(nk + k^2) = \tilde{O}(nk)$, which would yield a contradiction.
In turn, even an incremental deterministic algorithm with the guarantees of~\cref{thm:bicr_alg_merged} is ruled out. As a consequence, our (generalized) \emph{randomized} fully dynamic bicriteria approximation algorithm for the $(k, z)$-clustering problem---which is robust against an adaptive adversary---is the best achievable up to polylogarithmic factors.

\subsection{Structure of the Paper}
In~\cref{sec:technical_overview}, we provide a technical overview and highlight the intuition behind our algorithm. Some basic preliminaries follow in~\cref{sec:preliminaries}. Our fully dynamic $k$-center algorithm consists of two subroutines.
The first subroutine which is the fully dynamic bicriteria approximation algorithm is developed in~\cref{sec:bicr_recourse,sec:bicr_update_time} and then in~\cref{sec:merged_bicr_alg} its final version is presented. The second subroutine which 
efficiently converts the bicriteria approximate solution into a $k$-center solution is described in~\cref{sec:near_linear_alg}, and in~\cref{sec:put_together_fully_kcenter} our final fully dynamic $k$-center algorithm is demonstrated.

\section{Technical Overview} \label{sec:technical_overview}
In this section, we provide a high-level overview of our techniques and outline the intuition guiding our algorithms. 
Our fully dynamic $k$-center algorithm consists of two subroutines. 
The first subroutine builds on the static bicriteria approximation algorithm by Mettu and Plaxton~\cite{mettu2004optimal} (the MP-bi algorithm) and on its dynamic extension developed by Bhattacharya, Costa, Lattanzi, and Parotsidis~\cite{bhattacharya2023fully} (the dynamic MP-bi algorithm). The second subroutine is based on the fully dynamic $k$-center algorithm designed by Forster and Skarlatos~\cite{forster_skarlatos2025}.
The idea behind the two subroutines is to first maintain a bicriteria approximate solution, which acts as a sparsifier, and then run a near-linear time $k$-center algorithm on it.

\subsection{Fully Dynamic Bicriteria Approximation Algorithm} \label{subsec:fully_tech}
A high-level overview of the static MP-bi algorithm by Mettu and Plaxton~\cite{mettu2004optimal} is presented in the next paragraph.

\paragraph{MP-bi algorithm.}
Let $U_0 \coloneqq P$ be the input point set.
The MP-bi algorithm performs at most $t \coloneqq O(\log \frac{n}{k})$ iterations, where at the $i$-th iteration:
\begin{enumerate}
    \item A set $S_i$ is constructed by sampling a small subset of $U_i$ with $|S_i| = \tilde{O}(k)$.
    
    \item The smallest radius $\nu_i$ is computed such that the ball of radius $\nu_i$ around $S_i$ contains a constant fraction $\beta$ of points from $U_i$ (i.e., $\lvert\ball[U_i, S_i, \nu_i]\rvert \geq \beta \myhs |U_i|$).
    
    \item The $\ball[U_i, S_i, \nu_i]$ is \emph{removed}, and the MP-bi algorithm recurses on the remaining points $U_{i+1} \coloneq U_i \setminus \ball[U_i, S_i, \nu_i]$.
\end{enumerate}

\vspace{0.5em}
We refer to the $i$-th iteration of the MP-bi algorithm as the \emph{$i$-th level}, and to the set $U_i$ as the \emph{$i$-th execution set}. The fully dynamic MP-bi algorithm by Bhattacharya, Costa, Lattanzi, and Parotsidis~\cite{bhattacharya2023fully} is an elegant \emph{lazy} version of the MP-bi algorithm. For a level $i \in [0, t]$, let $\lazy(U_i)$ denote the lazy modification of the $i$-th execution set $U_i$, which undergoes adversarial point updates. Informally, a lazy set does not include new samples or newly constructed balls; it is updated only in response to adversarial point updates.\footnote{A more formal definition is given in the analysis (see also~\cref{def:lazy_set}).} For a fixed level $i \in [0,t]$ and a small constant $\epsilon \in (0,1)$, the crucial property identified in~\cite{bhattacharya2023fully} can be intuitively described as follows:

\paragraph{Lazy property.} 
The $i$-th execution set $U_i$ can be modified by $\epsilon \myhs |U_i|$ points, and the approximation ratio of the $i$-th $\ball[U_i, S_i, \nu_i]$ (constructed in the initial $U_i$ before the modifications) is affected only by a constant factor with respect to $\lazy(U_i)$. In other words, if the size of the symmetric difference between $U_i$ and $\lazy(U_i)$ is at most $\epsilon \myhs |U_i|$ (i.e., $|U_i \myhs \triangle \lazy(U_i)| \leq \epsilon \myhs |U_i|$), the lazy property guarantees that the $i$-th ball constructed with respect to $U_i$ still provides a good approximation with respect to $\lazy(U_i)$.

\vspace{0.5em}
The dynamic MP-bi algorithm in~\cite{bhattacharya2023fully} performs adversarial point updates on all relevant execution sets, allowing each to remain lazy for a number of adversarial updates proportional to its size. 
Once a specified threshold for an execution set $U_i$ is exceeded, the corresponding balls for all levels $j \in [i, t]$ are periodically rebuilt. 

Since this approach yields amortized guarantees, in the next subsections we elaborate on the primary intuition behind achieving constant worst-case recourse and near-optimal worst-case update time. The MP-bi algorithm is well suited for lazy updates, provided that the execution sets $U_i$ are of sufficiently large size---a feature also used in the work of Bhattacharya, Costa, Farokhnejad, Lattanzi, and Parotsidis~\cite{bhattacharya2025alm_opt_kcenter}.
The key idea is to extend the ``lazy'' approach introduced in~\cite{bhattacharya2023fully}. Roughly speaking, we show that making the execution sets lazier, along with appropriate data structures and careful adjustments, yields worst-case guarantees.

\subsubsection{Constant Worst-Case Recourse}
The challenge in achieving worst-case recourse guarantees lies in the fact that the set $S = \bigcup_{i=0}^t S_i$ must be periodically reconstructed to maintain a constant-factor approximation. By being lazy for $\Theta(|U_i|)$ adversarial updates at a fixed level $i \in [0, t]$, the desired amortized recourse guarantees are achieved. We push this idea further by maintaining the set $S$ which is allowed to be refreshed, and a \emph{lazy set} $\hat{S}$ which synchronizes with $S$ during the \emph{transition phases}. This lazy set $\hat{S}$ serves as our bicriteria approximate solution, and it changes by only a constant number of points per adversarial update (i.e., constant worst-case recourse).

The $\ell$-th transition phase is responsible for aligning $\hat{S}$ with the target set $S^\ellexp$, and may span multiple adversarial point updates. At the beginning of the $\ell$-th transition phase, the target set $S^\ellexp$ is initialized to the most recent version of $S$. Thereafter, the target set $S^\ellexp$ changes only due to adversarial point updates, even if~$S$ is later reconstructed. Notice that it is crucial for $S^\ellexp$ not to be reconstructed, as it takes time for~$\hat{S}$ to match $S^\ellexp$ when only a constant number of changes per adversarial update is allowed (i.e., constant worst-case recourse).

Our bicriteria approximate solution $\hat{S}$ is \emph{simulating} the target set $S^\ellexp$, rather than the set $S$ which is allowed to be reconstructed. Because the target set $S^\ellexp$ is constructed within execution sets that have been modified by adversarial point updates, we rely on the lazy property of the execution sets to argue that $\hat{S}$ is a constant-factor approximate solution. The core structural observations we make are the following: 
\begin{enumerate}
    \item The number of adversarial updates required for $\hat{S}$ to match $S^\ellexp$ (i.e., the length of the transition phase) is $O(k \log n \log\frac{n}{k})$. Conceptually, the length of the transition phase is proportional to the size of the bicriteria approximate solution.

    \item For every level $i \in [0, t)$,\footnote{Since $U_t \subseteq S$, the last $t$-th level is a trivial special case.} the size of the $i$-th execution set $U_i$ is $\Omega(k \log n \log\frac{n}{k})$. In turn, the lazy property guarantees that each execution set can be lazy
    for $\Theta(k \log n \log\frac{n}{k})$ adversarial updates.
\end{enumerate}
Therefore, we can deduce that each execution set can remain lazy during a transition phase, in which the solution $\hat{S}$ gradually aligns with the target set $S^\ellexp$. Essentially, we leverage the fact that the laziness of each execution set is proportional to, or larger than, the size of the bicriteria approximate solution. 

During the $\ell$-th transition phase, the solution $\hat{S}$ transitions from the previous target set $S^{(\ell-1)}$ to the current target set $S^\ellexp$. Consequently, our proofs exploit that the fully dynamic algorithm can be even lazier, and the technical details involved in establishing a constant-factor approximation become rather intricate.

\subsubsection{Near-Optimal Worst-Case Update Time}
By remaining lazy for $\Theta(|U_i|)$ adversarial updates at a fixed level $i \in [0, t]$, the desired amortized update time guarantees are also achieved. We remind the reader that the challenge in achieving worst-case recourse guarantees arises because the set $S$ is periodically reconstructed. To address this, during the $\ell$-th transition phase a fixed target set $S^\ellexp$ is defined\footnote{In this context, \emph{fixed} means that the target set $S^\ellexp$ is never completely rebuilt, even though it may change via adversarial point updates. We later refer to this as a \emph{non-rebuilt} set (in~\cref{def:non-rebuilt}).} and the fully dynamic algorithm gradually aligns the solution~$\hat{S}$ to the fixed~$S^\ellexp$. The intention is to apply a similar approach to achieve worst-case update time by fixing an execution set $U_i$ and gradually performing the computation of the $i$-th level. This computation involves the $i$-th sampled set, the $i$-th radius, and the $i$-th ball within $U_i$, as described previously in the MP-bi algorithm.
However, choosing which execution set to fix is more subtle.
 
 The first attempt is to fix the $i$-th execution set $U_i$ and construct the corresponding balls for all levels~$j \in [i, t]$ within this fixed set. The issue with this approach is that the subsequent levels $j \in [i+1, t]$ cannot be as lazy as the $i$-th level, yet they depend on the fixed $i$-th execution set. In the following, we provide a more detailed explanation of the obstacles and how to resolve them.

\paragraph{Obstacles and ideas.}
For fixed levels $i \in [0, t]$ and $j \in [i, t]$, constructing the $j$-th ball
within the $j$-th execution set $U_j$ requires $O(|S_j| |U_j|) = O(k \myhs |U_j| \log n)$ time. Since $j \geq i$ it holds that $|U_j| \leq |U_i|$, and thus $O(k \myhs |U_j| \log n)$ is upper bounded by $O(k \myhs |U_i| \log n)$. Consider a fixed level $i \in [0, t]$, and assume that the corresponding threshold for the $i$-th execution set is exceeded. Since there are at most $t = O(\log \frac{n}{k})$ levels, the time required to rebuild the corresponding balls for all levels $j \in [i, t]$ is upper bounded by~$O(k \myhs |U_i| \log n \myhs \log \frac{n}{k})$.

Consider the $i$-th execution set $U^\tauexp_i$ at the $\tau$-th adversarial update. The set $\lazy(U^\tauexp_i)$ denotes the lazy version of $U^\tauexp_i$ that undergoes adversarial point updates, and the set $U_i$ represents the current $i$-th execution set. Note that $U_i$ can be completely different from both $U^\tauexp_i$ and $\lazy(U^\tauexp_i)$, as the algorithm may produce new execution sets when the corresponding thresholds are exceeded. Based on the lazy property, if the size of the symmetric difference between $U^\tauexp_i$ and $\lazy(U^\tauexp_i)$ is at most $\epsilon \myhs |U^\tauexp_i|$ for a small constant $\epsilon \in (0, 1)$, then the $i$-th ball constructed with respect to $U^\tauexp_i$ still provides a good approximation with respect to $\lazy(U^\tauexp_i)$. In other words, the set $U^\tauexp_i$ can remain lazy for $\epsilon \myhs |U^\tauexp_i|$ adversarial updates, while the approximation ratio is affected only up to a constant factor.

Therefore, since (1) rebuilding the balls for all levels $j \in [i, t]$ starting from $U^\tauexp_i$ requires~$O(k \myhs |U^\tauexp_i| \log n \myhs \log \frac{n}{k})$ time and (2) the set $U^\tauexp_i$ can be maintained lazily for
$\epsilon \myhs |U^\tauexp_i|$ adversarial updates, we need  $\Theta(\frac{k}{\epsilon} \log n \myhs \log \frac{n}{k})$ time per adversarial update. Hence, all $j$-th balls within $U^\tauexp_i$ (where $j \in [i, t]$) can be constructed with a worst-case update time of $O(\frac{k}{\epsilon} \log n \myhs \log \frac{n}{k})$, while ensuring that each $j$-th ball provides a good approximation with respect to $\lazy(U^\tauexp_i)$. In this naive analysis though, we have not specified the execution sets in which each $j$-th ball is constructed. This is very important, as in the correctness analysis we need to associate the radius of the $j$-th ball with the optimal radius of the respective execution set.

\subparagraph{Subtle case.}
Let us clarify the subtlety that needs careful treatment. During the lazy construction starting from the set $U^\tauexp_i$, assume that the $i$-th ball $B_i$ has been constructed. Since $U^\tauexp_i$ is viewed as a fixed set, for every level $j > i$ it is natural to construct the $j$-th ball $B_j$   within the set $U^\tauexp_j \coloneqq U^\tauexp_i \setminus \bigcup_{\xi=i}^{j-1} B_\xi$. Notice that the size of the symmetric difference between the sets $U^\tauexp_j$ and $\lazy(U^\tauexp_j)$ is proportional to~$\epsilon \myhs |U^\tauexp_i|$. However, for some much larger level $j \gg i$ it follows that $|U^\tauexp_j| \ll |U^\tauexp_i|$, and thus the value~$\epsilon \myhs |U^\tauexp_i|$ can be much larger than $\epsilon \myhs |U^\tauexp_j|$. As a consequence, there is no guarantee from the lazy property that the $j$-th ball $B_j$ constructed with respect to $U_j^\tauexp$ provides a good approximation with respect to $\lazy(U_j^\tauexp)$. In other words, the $j$-th level cannot be as lazy as the $i$-th level. 

\paragraph{The remedy.}
To address this issue during the lazy construction starting from the fixed set $U^\tauexp_i$, we proceed as follows. For every level $j \in [i, t]$, each $j$-th ball $B_j$ is constructed within the fixed set $U_j^i \coloneqq \lazy(U^\tauexp_i) \setminus \bigcup_{\xi=i}^{j-1} B_\xi$ instead of within~$U_j^\tauexp$. Namely, the algorithm starts constructing the $j$-th ball~$B_j$ from the current subset of points~$U_j^i$, which can remain lazy for $\epsilon \myhs |U_j^i|$ adversarial point updates (based on the lazy property). The crucial observation is that the construction of $B_j$ within $U_j^i$ requires $O(k \myhs |U_j^i| \log n)$ time, and furthermore the construction of all subsequent balls within $\lazy(U_j^i)$ requires $O(k \myhs |U_j^i| \log n \myhs \log \frac{n}{k})$ time. 

Therefore, with a worst-case update time of $O(\frac{k}{\epsilon} \log n \myhs \log \frac{n}{k})$, each $j$-th ball $B_j$ constructed in $U_j^i$ provides a good approximation with respect to $\lazy(U_j^i)$ (since $|U_j^i \myhs \triangle \lazy(U_j^i)| \leq \epsilon \myhs |U_j^i|$). To summarize, when the lazy rebuild starting from the set $U^\tauexp_i$ finishes, each $j$-th ball has been constructed within an execution set that remained lazy for a number of adversarial updates proportional to its size. Thus, the lazy property can be applied to ensure a constant-factor approximation.

\vspace{0.5em}
Observe that the lazy construction starting from the fixed set $U_i^\tauexp$ concerns only the $i$-th level. For this reason, in our formal algorithm presented in~\cref{sec:bicr_update_time}, each level $i$ is associated with its own $i$-th algorithm~$\mathcal{BD}_i$. For every level $i \in [0, t]$, the fully dynamic algorithm maintains the $i$-th algorithm $\mathcal{BD}_i$, which is responsible for the lazy rebuild of the $i$-th level (which involves the construction of all subsequent balls~$B_j$ within $\lazy(U_i^\tauexp)$ where $j \in [i, t]$).
As a result, the worst-case update time of our fully dynamic bicriteria approximation algorithm becomes $O(k \log n \myhs (\log \frac{n}{k})^2)$.

\subsection{$k$-Center Clustering with Near-Linear Update Time}
The second subroutine efficiently converts
the bicriteria approximate solution into a $k$-center solution. Our second subroutine is based on the fully dynamic algorithm by Forster and Skarlatos~\cite{forster_skarlatos2025}, which has a worst-case update time of $O(nk)$. We modify the component responsible for this runtime, reducing the worst-case update time to $O(n \log n)$; the original guarantees on the approximation ratio and worst-case recourse are preserved.

\paragraph{The fully dynamic algorithm in~\cite{forster_skarlatos2025} and the expensive component.}
The fully dynamic algorithm in~\cite{forster_skarlatos2025} maintains a set of $k$ centers, the corresponding clusters (i.e., the points each center is responsible for), and a radius $r$. Initially, all centers and clusters are \emph{regular}, meaning that the points in each cluster are within distance $r$ from their center. This property can be violated if a center is deleted and replaced by another point from the cluster, in which case the cluster becomes a \emph{zombie} cluster. To ensure the accuracy of the radius $r$, such a replacement must be located at a distance greater than $r$ from all other centers. If this is not possible in the original algorithm, the costliest component of the algorithm is executed. 

In this component, the algorithm tries to find a sequence of $2l - 1$ points $p_1, c_2, p_2, \ldots, c_l, p_l$, where~$p_1$ is in the cluster of the deleted center $c_1$ and each point $p_i$ is in the cluster of $c_i$ with $\dist(c_i, p_i) > r$. Additionally, every point $p_i$ (where $1 \leq i \leq l-1$) is close to the subsequent center, namely $\dist(p_i, c_{i+1}) \le r$. Finally, the last point~$p_l$ in the sequence is not within distance $r$ from any other center. The point $p_l$ then  becomes the new center replacing $c_1$, and the responsibilities of the centers are shifted according to the sequence. When no such sequence can be found, all relevant points along the sequence are close to some center. Therefore, the relevant points can be reassigned to those centers and the corresponding clusters become regular again. 

Detecting such a sequence or the lack thereof can be seen as a graph search problem where the points are the vertices and the  graph has two types of edges. On the one hand, there are edges within a cluster that connect the centers to the points in their clusters. On the other hand, there are edges across clusters that connect points with close centers of other clusters. Hence, there could be $\Theta(nk)$ edges across clusters, leading to the aforementioned update time. 

\paragraph{Our alternative approach.}
We show that instead of testing all possible sequences, it suffices to explore the underlying graph without taking all edges across clusters into consideration. To this end, we pick an arbitrary close center $c_{i+1}$ for each point $p_i$ in a cluster that is either the starting cluster or was reached previously via an edge across clusters. If we find a sequence with the desired properties in this way, we  shift the centers along it, as done in~\cite{forster_skarlatos2025}. 

Otherwise, we ensure that all points in the encountered clusters are close (i.e., within distance $r$) to some center. In turn, these points can be assigned to a different cluster, and all the encountered clusters become regular. Specifically, all points that were in the cluster of the deleted center are assigned to some regular cluster. Since the valid set of sequences identified by our algorithm is a subset of the valid sequences identified by~\cite{forster_skarlatos2025}, our analysis carefully addresses this subtle situation. We remark that the underlying graph is not constructed explicitly, and our sequence construction ensures that the execution of our component requires $O(n \log n)$ time per adversarial update. 

\section{Preliminaries} \label{sec:preliminaries}
Consider a point set $P$ in a metric space $(\mathcal{X}, \dist)$ with $|P| = n$. We omit the notation $(\mathcal{X}, \dist)$ when the context is clear.
For any point $p \in P$ and any subset $S \subseteq P$, the \emph{distance} from $p$ to $S$ is $\dist(p, S) = \min_{q \in S} \dist(p, q)$.

\begin{Definition}[$k$-center clustering problem]\label{def:kcenter}
    Given a point set $P$ in a metric space and an integer $k \geq 1$, 
    the goal of the $k$-center clustering problem is to output a subset $S \subseteq P$ of at most $k$ points
    such that the value $\max_{p \in P} \dist(p, S)$ is minimized.
\end{Definition}

The points in $S$ are referred to as \emph{centers}, and the ($k$-center) \emph{cost} of $S$ is $\cost(S) \coloneqq \max_{p \in P} \dist(p, S)$.
The \emph{optimal} $k$-center cost is denoted by $R^* \coloneqq \min_{|S| \leq k} \cost(S)$. A $k$-center instance is the tuple $(P, k)$, consisting of the input point set and the integer parameter $k$. 

\begin{definition}[$(\alpha, \beta)$-bicriteria approximation] \label{def:bicriteria}
    Given a $k$-center instance, an $(\alpha, \beta)$-bicriteria approximate solution is a subset of points $S \subseteq P$ such that:
    \[
        \cost(S) \leq \alpha \cdot R^* \;\text{ and } \; |S| \leq \beta \cdot k.
    \]
\end{definition}

For a fixed set of points $U$, a subset $S \subseteq U$, and a positive real number $r$, the set $\ball(U, S, r) \coloneqq \{p \in U \mid \dist(p, S) < r\}$ denotes the \emph{open ball} of radius $r$ around $S$ in $U$, and the set $\ball[U, S, r] \coloneqq \{p \in U \mid \dist(p, S) \leq r\}$ denotes the \emph{closed ball} of radius $r$ around $S$ in $U$.

In the \emph{dynamic setting}, the point set is subject to adversarial point updates (i.e., point insertions and point deletions), and 
the adaptive adversary has full knowledge of the dynamic algorithm.
The \emph{recourse} is defined as the number of changes/swaps made to the set of centers after an adversarial update.
More formally, let~$S^\old$ and $S$ be the solution before and after the adversarial update. Then assuming that $|S^\old| = |S|$, the recourse is half the size of the symmetric difference between $S^\old$ and~$S$ (i.e., equal to $\frac{|S^\old \triangle \myhs S|}{2}$).
For solutions of different sizes or multisets, the recourse equals the size of the corresponding symmetric difference.
\section{Fully Dynamic Bicriteria Approximation with Worst-Case Recourse}\label{sec:bicr_recourse}
In this section, we develop a fully dynamic $(8, O(\log n \myhs \log \frac{n}{k}))$-bicriteria approximation algorithm for the $k$-center clustering problem with \emph{constant worst-case recourse}, as demonstrated in~\cref{thm:bicr_alg_worst_case}. A pseudocode of our fully dynamic bicriteria approximation algorithm is provided in~\cref{alg:aux_rebuild_proced,alg:point_ins,alg:point_del,alg:apply_point_repl,alg:lazy_sync}.\footnote{In the pseudocodes, multisets use .add() and .remove(), while sets use ``add to'' and ``remove from''.}

\begin{restatable}{theorem}{bicralgworstcase}\label{thm:bicr_alg_worst_case}
    There is a randomized fully dynamic algorithm against an adaptive adversary that, given a point set $P$ in a metric space subject to point updates and an integer $k \geq 1$, maintains a subset of points $\hat{S} \subseteq P$ such that:
    \begin{itemize}
        \item The set $\hat{S}$ is with high probability an $(8, O(\log n \myhs \log \frac{n}{k}))$-bicriteria approximate solution to the $k$-center clustering problem.

        \item The amortized update time is $O(k \log n \myhs \log \frac{n}{k})$.

        \item The worst-case recourse is $O(1)$.
    \end{itemize}
\end{restatable}

Our fully dynamic bicriteria approximation algorithm partially simulates the dynamic MP-bi algorithm from~\cref{lem:prior_bicr_alg}. Specifically, when the adaptive adversary sends a point update, the counters $\textit{cnt}_i$ for the affected execution sets $U_i$ at levels $i \in [0, t]$ are incremented by one. Whenever any counter $\textit{cnt}_i$  exceeds its threshold $\lambda \myhs n_i$, the procedure \texttt{RebuildFromLayer}($i$) (in~\cref{alg:aux_rebuild_proced}) is executed, producing a new set $S$ with worst-case recourse of $\Theta(k \log n \myhs \log \frac{n}{k})$ (see \texttt{CheckThresholdRebuild}$(\hspace{0.5pt})$ in~\cref{alg:aux_rebuild_proced}). The parameters~$\lambda, \beta \in (0, 1)$ are small positive constants and the parameter $\alpha \geq 1$ is a sufficiently large constant.\footnote{Tighter bounds on $\lambda$ are given in the analysis.} Note also that the value of the last level $t$ may change after an execution of~\texttt{RebuildFromLayer}$(\hspace{0.5pt})$, but the algorithm ensures that $t = O(\log \frac{n}{k})$ throughout its execution.

In order to achieve constant worst-case recourse, our bicriteria approximate solution $\hat{S}$ is updated lazily during the transition phases (see \texttt{LazySync}$(\hspace{0.5pt})$ in~\cref{alg:lazy_sync}). In particular during the $\ell$-th transition phase, the solution $\hat{S}$ gradually transitions to a \emph{non-rebuilt} version of the set $S$, changing by only a constant number of points per adversarial update. 

\begin{definition}[non-rebuilt set] \label{def:non-rebuilt}
    A set $\tilde{S}$ is called non-rebuilt if it does not change due to \texttt{RebuildFromLayer}$(\hspace{0.5pt})$, but only due to adversarial point updates.
\end{definition}

In order to perform the $\ell$-th transition phase more elegantly, we maintain the bicriteria approximate solution $\hat{S}$ as a multiset~(i.e., allowing multiple occurrences of the same point). Notice that the recourse of a multiset is an upper bound on the recourse of its corresponding set. For this reason, we can maintain the solution $\hat{S}$ as a multiset and analyze its worst-case recourse bounds accordingly.

\subsection{Transition Phases} \label{sec:trans_phase}
A transition phase can last for multiple adversarial point updates. At the beginning of the $\ell$-th transition phase, the target set $S^{\ellexp}$ is initialized to the set $S$. We remark that the target set $S^\ellexp$ is a non-rebuilt set, even though the set $S$ changes also due to \texttt{RebuildFromLayer}$(\hspace{0.5pt})$.
The inserted set $I^\ellexp$ denotes the adversarially inserted points during the $\ell$-th transition phase. The $\ell$-th transition phase consists of two subphases, as described in \texttt{LazySync}$(\hspace{0.5pt})$ in~\cref{alg:lazy_sync}:
\begin{enumerate}
    \item During the first subphase of $\ell$, the target set $S^\ellexp$ is gradually added to the solution $\hat{S}$.
    
    \item During the second subphase of $\ell$, both the target set $S^{(\ell-1)}$ and the inserted set $I^{(\ell-1)}$ are gradually removed from the solution $\hat{S}$.
\end{enumerate}
The term ``gradually'' refers to the worst-case recourse, which is $O(\frac{1}{\lambda^2})$. The bicriteria approximate solution~$\hat{S}$ also changes due to adversarial updates, but only by $O(1)$ points per adversarial update. Since $\lambda$ is a small positive constant, this implies that the worst-case recourse is constant.

\begin{algorithm}[H]\footnotesize
\algnewcommand{\LineComment}[1]{\State \(\triangleright\) #1}
\algrenewcommand\algorithmiccomment[1]{\hspace{1em} \(\triangleright\) #1}
\caption{\textsc{Lazy Synchronization}{}}\label{alg:lazy_sync}

\begin{algorithmic}[1]

\Function{LazySync}{\hspace{0.5pt}}
    \For{$i \in \{1, \ldots, \frac{32}{\lambda^2}\}$}    
        \If{$\textit{subPhase} = 1$}
            \State $\hat{S}.\textit{add}(S^\ellexp[\textit{pos}])$\label{algline:add_s_l}
            \State $\textit{pos} \gets \textit{pos} + 1$
            \vspace{0.3em}
            \If{$\textit{pos} \geq |S^\ellexp|$}
                \State $\hat{\sigma} \gets \sigma^\ellexp$
                \State $\textit{subPhase} \gets 2$ 
                \State $\textit{pos} \gets 0$ 
                \State $\textit{DelSet} \gets S^{(\ell-1)} \cup I^{(\ell-1)}$ \label{algline:del_set}
            \EndIf
        \vspace{0.3em}
        \ElsIf{$\textit{subPhase} = 2$}
            \State $\hat{S}.\textit{remove}(\textit{DelSet}[\textit{pos}])$ \label{algline:rem_del_set}
            \State $\textit{pos} \gets \textit{pos} + 1$
            \vspace{0.3em}
            \If{$\textit{pos} \geq \lvert\textit{DelSet}\rvert$}
                \State $\ell \gets \ell + 1$
                \State $S^\ellexp \gets S$ \label{algline:target_set_ell_init_S}
                \State $\sigma^\ellexp \gets \sigma$

                \vspace{0.3em}
                \State $\textit{subPhase} \gets 1$
                \State $\textit{pos} \gets 0$
            \EndIf
        \EndIf
    \EndFor
\EndFunction

\end{algorithmic}
\end{algorithm}

\paragraph{Point insertion.}
When the adversary inserts a new point $p$ into the point set $P$ (i.e., $P \coloneqq P \cup \{p\}$) during either subphase, the fully dynamic bicriteria approximation algorithm adds $p$ to the set $S$, to the maintained bicriteria approximate solution $\hat{S}$, and to the inserted set $I^\ellexp$. Additionally, the algorithm adds the newly inserted point $p$ to every execution set $U_i$ and increments each counter $\mathit{cnt}_i$ by one, for every level $i \in [0, t]$. 

\begin{algorithm}[H]\footnotesize
\algnewcommand{\LineComment}[1]{\State \(\triangleright\) #1}
\algrenewcommand\algorithmiccomment[1]{\hspace{1em} \(\triangleright\) #1}
\caption{\textsc{Point Insertion}{}}\label{alg:point_ins}

\begin{algorithmic}[1]

\Function{PointInsertion}{$p$}
    \State Insert $p$ into $P$
    \vspace{0.3em}
    
    \For{$i \in \{0, \ldots, t\}$} 
        \State Add $p$ to $U_i$ \label{algline:add_to_Ui}
        \State $cnt_i \gets cnt_i + 1$ \label{algline:increm_cnt_pnt_ins}
    \EndFor

    \vspace{0.3em}
    \State Add $p$ to $S$ and $I^\ellexp$
    \State $\hat{S}.add(p)$
    \State $\sigma(p) \gets p, \hat{\sigma}(p) \gets p$

    \vspace{0.3em}
    \State \Call{CheckThresholdRebuild}{\hspace{0.5pt}}
    \State \Call{LazySync}{\hspace{0.5pt}}
\EndFunction

\end{algorithmic}
\end{algorithm}

\paragraph{Point deletion.}
When the adversary deletes a point $p$ from the point set $P$ (i.e., $P \coloneqq P \setminus \{p\}$) during either subphase, the fully dynamic bicriteria approximation algorithm finds the largest level $i \in [0, t]$ such that the point $p$ belongs to the $i$-th execution set $U_i$. The algorithm then removes the deleted point $p$ from every execution set $U_j$ and increments each counter $\mathit{cnt}_j$ by one, for every level $j \in [0, i]$.

Next, through the \texttt{ApplyPointReplacement}$(\hspace{0.5pt})$ procedure (in~\cref{alg:apply_point_repl}), the fully dynamic algorithm replaces the deleted point $p$ with another point in the set $S$ (see~\linecref{algline:repl_in_S} in~\cref{alg:point_del}), in the target sets $S^{(\ell-1)}$ and $S^\ellexp$~(see~\linescref{algline:repl_in_S_ell-1}{algline:repl_in_S_ell} in~\cref{alg:point_del}), and in the maintained bicriteria approximate solution~$\hat{S}$ (see~\linescref{algline:remove_p_hat_S}{algline:add_p_hat_S} in~\cref{alg:apply_point_repl}). These replacements are performed only if $p$ belongs to the respective sets, and $p$ is replaced by another point from the cluster induced by the respective assignment (see~\linecref{algline:rem_for_replace} in~\cref{alg:apply_point_repl}).\footnote{The updates of the assignments are performed implicitly in $O(1)$ time (rather than $\Theta(n)$) by maintaining the id, center, and size of each  induced cluster (see also Claim 3.7 in~\cite{bhattacharya2023fully}).}
Furthermore, these replacements are performed while ensuring the following properties:
\begin{itemize}
    \item During the first subphase of $\ell$, the target sets $S^{(\ell-1)}$ and $S^\ellexp$ are subsets of the point set $P$, and the solution $\hat{S}$ contains only points from~$P$.\footnote{We avoid referring to $\hat{S}$ as a ``subset'' of $P$, since $\hat{S}$ is a multiset and $P$ is a set.} Moreover during the first subphase of $\ell$, the target set $S^{(\ell-1)}$ remains a subset of the solution $\hat{S}$.
   
    \item During the second subphase of $\ell$, the target set  $S^\ellexp$ remains a subset of the point set $P$, and the solution $\hat{S}$ contains only points from~$P$. Moreover during the second subphase of $\ell$, the target set $S^\ellexp$ remains a subset of the solution $\hat{S}$.
\end{itemize}

\begin{algorithm}[H]\footnotesize
\algnewcommand{\LineComment}[1]{\State \(\triangleright\) #1}
\algrenewcommand\algorithmiccomment[1]{\hspace{1em} \(\triangleright\) #1}
\caption{\textsc{Point Replacement}{}}\label{alg:apply_point_repl}

\begin{algorithmic}[1]

\Function{ApplyPointReplacement}{$p, \tilde{S}, \tilde{\sigma}$}
    \If{$p \in \tilde{S}$} \label{algline:if_remove_p_tilde_S}
        \State Remove $p$ from $\tilde{S}$ \label{algline:remove_p_tilde_S}
        \State $\textit{Rem} \gets \left(\tilde{\sigma}\right)^{-1}(p) \setminus \{p\}$ \label{algline:rem_for_replace} \Comment{Preimage of $p$ under $\tilde{\sigma}$ excluding $p$ itself}
        \vspace{0.3em}
        
        \If{$\textit{Rem} \neq \emptyset$}
            \State Let $c$ be a point in $\textit{Rem}$
            \State Add $c$ to $\tilde{S}$ \label{algline:add_p_tilde_S}
            \State $\forall q \in \textit{Rem}: \tilde{\sigma}(q) \gets c$   
            \vspace{0.3em}
            \LineComment{$\hat{S}$ is modified only due to the target sets $S^{(\ell-1)}$ and $S^\ellexp$, and not due to $S$}
            \vspace{0.1em}
            \If{$(\tilde{\sigma} = \sigma^{(\ell-1)} \textbf{ or } \tilde{\sigma} = \sigma^\ellexp) \textbf{ and } p \in \hat{S}$}
                \State $\hat{S}.\textit{remove}(p)$ \label{algline:remove_p_hat_S}
                \State $\hat{S}.\textit{add}(c)$ \label{algline:add_p_hat_S}
                \vspace{0.3em}
                \LineComment{$\hat{\sigma}$ simulates $\sigma^{(\ell-1)}$ during the first subphase and $\sigma^\ellexp$ during the second subphase}
                \If{($\textit{subPhase} = 1$ \textbf{ and } $\tilde{\sigma} = \sigma^{(\ell-1)}$) \textbf{ or }
                ($\textit{subPhase} = 2$ \textbf{ and } $\tilde{\sigma} = \sigma^\ellexp$)}
                    \State $\forall q \in \textit{Rem}: \hat{\sigma}(q) \gets c$ \label{algline:hat_sigma_simul}
                \EndIf
            \EndIf
        \EndIf
    \EndIf
\EndFunction

\end{algorithmic}
\end{algorithm}
\begin{algorithm}[H]\footnotesize
\algnewcommand{\LineComment}[1]{\State \(\triangleright\) #1}
\algrenewcommand\algorithmiccomment[1]{\hspace{1em} \(\triangleright\) #1}
\caption{\textsc{Point Deletion}{}}\label{alg:point_del}

\begin{algorithmic}[1]

\Function{PointDeletion}{$p$}
    \State Delete $p$ from $P$
    \vspace{0.3em}
    \State $i \gets \max\{j \in [0, t] \mid p \in U_j\}$
    \vspace{0.3em}
    
    \For{$j \in \{0, \dots, i$\}}
        \State Remove $p$ from $U_j$ \label{algline:remove_from_Uj}
        \State $\mathit{cnt}_j \gets \mathit{cnt}_j + 1$ \label{algline:increm_cnt_pnt_del}
    \EndFor

    \vspace{0.3em}
    \LineComment{Apply point replacement to the set $S$, the target sets $S^{(\ell-1)}$ and $S^\ellexp$, and the solution $\hat{S}$}
    \State \Call{ApplyPointReplacement}{$p, S, \sigma$} \label{algline:repl_in_S}
    \vspace{0.3em}
    
    \If{$\textit{subPhase} = 1$} \label{algline:if_S_ell-1}
        \State \Call{ApplyPointReplacement}{$p, S^{(\ell-1)}, \sigma^{(\ell-1)}$} \label{algline:repl_in_S_ell-1}
    \EndIf

    \vspace{0.3em}
    \State \Call{ApplyPointReplacement}{$p, S^\ellexp, \sigma^\ellexp$} \label{algline:repl_in_S_ell} 
    \vspace{0.3em}
    \State $I^{(\ell-1)} \gets I^{(\ell-1)} \cap P$ \label{algline:I_ell-1_subset_P}
    \State $I^\ellexp \gets I^\ellexp \cap P$ \label{algline:I_ell_subset_P}
    \State $\hat{S} \gets \hat{S} \cap P$ \label{algline:hat_S_subset_P}

    \vspace{0.3em}
    \State \Call{CheckThresholdRebuild}{\hspace{0.5pt}}
    \State \Call{LazySync}{\hspace{0.5pt}}
\EndFunction

\end{algorithmic}
\end{algorithm}

\vspace{1em}
At the end of each adversarial point update, the fully dynamic algorithm invokes \texttt{CheckThresholdRebuild}$(\hspace{0.5pt})$ which scans all levels $i \in [0, t]$ to detect whether a counter $\mathit{cnt}_i$ has exceeded its threshold $\lambda n_i$. For the smallest level $i \in [0, t]$ such that $\textit{cnt}_i > \lambda n_i$,\footnote{Notice that for some adversarial point updates, there may not be a level whose counter exceeds its threshold.} the procedure \texttt{RebuildFromLayer}($i$)
is executed, producing a new set $S$. The fully dynamic algorithm then invokes \texttt{LazySync}$(\hspace{0.5pt})$, which gradually converts the solution $\hat{S}$ to $S^\ellexp \cup I^\ellexp$, modifying $\hat{S}$ by only $\frac{32}{\lambda^2}$ points per adversarial update.

\begin{algorithm}[H]\footnotesize
\algnewcommand{\LineComment}[1]{\State \(\triangleright\) #1}
\algrenewcommand\algorithmiccomment[1]{\hspace{1em} \(\triangleright\) #1}
\caption{\textsc{Auxiliary Rebuilding Procedures}{}}\label{alg:aux_rebuild_proced}

\begin{algorithmic}[1]

\LineComment{The algorithm has global access to all data structures}
\LineComment{The solution $\hat{S}$ is a multiset}
\State $\textit{pos} \gets 0, \ell \gets 0$
\State $\textit{subPhase} \gets 1$
\vspace{1em}

\Function{RebuildFromLayer}{$i$} 
    \State $j \gets i$
    \vspace{0.3em}
    
    \While{$|U_j| > \alpha \myhs k \log n \myhs \log \frac{n}{k}$} \label{algline:while_rebuild}
        \State $n_j \gets |U_j|$ \label{algline:nj_init_size_of_Uj}
        \State $cnt_j \gets 0$ \label{algline:counter_init_zero}

        \State Construct $S_j$ by sampling each $p \in U_j$ independently with probability $\min\Big(\frac{\alpha \myhs k \log n}{|U_j|}, 1\Big)$ \label{line:sample_Si} 
        
        \State $\nu_j \gets \min \{r \geq 0 \mid \lvert\ball[U_j, S_j, r]\rvert \geq \beta \myhs |U_j|\}$ \label{line:min_rad_nu}
         
        \State $B_j \gets \ball[U_j, S_j, \nu_j]$ \label{line:construct_Bj}
        \State For every $p \in B_j: \sigma(p) \gets \argmin_{q \in S_j} \dist(p, q)$ \label{algline:sigma_in_rebuild}

        \vspace{0.3em}
        \State $U_{j+1} \gets U_j \setminus B_j$ \label{line:new_exec_set}
        \State $j \gets j + 1$
    \EndWhile
    
    \vspace{0.3em}
    \State $t \gets j$
    \State $S_t \gets U_t$, $B_t \gets U_t$,
    $\nu_t \gets 0$, $n_t \gets |U_t|$, $cnt_t \gets 0$, $\forall p \in U_t: \sigma(p) \gets p$ \label{line:last_level}
    \State $S \gets \bigcup_{j=0}^{t} S_j$ \label{line:final_S}
\EndFunction

\vspace{1em}

\Function{CheckThresholdRebuild}{\hspace{0.5pt}}
    \State $i \gets 0$
    \vspace{0.3em}
    
    \While{$i \leq t$ \textbf{ and } $cnt_i \leq \lambda \myhs n_i$} \label{algline:check_thres}
        \State $i \gets i + 1$
    \EndWhile

    \vspace{0.3em}
    \If{$i \leq t$}
        \State \Call{RebuildFromLayer}{$i$} \label{algline:exec_rebuild}
    \EndIf
\EndFunction

\end{algorithmic}
\end{algorithm}

\paragraph{End of a transition phase.}
Recall that the $\ell$-th transition phase can last for multiple adversarial point updates, and that the maintained bicriteria approximate solution $\hat{S}$ changes by only $O(\frac{1}{\lambda^2})$ points per adversarial update.
At the end of the $\ell$-th transition phase, the solution $\hat{S}$ is the union of  the target set $S^{\ellexp}$ and the inserted set $I^\ellexp$, namely $\hat{S} = S^{\ellexp} \cup I^\ellexp$. Once the $\ell$-th transition phase ends, the transition phase index $\ell$ is incremented by one, and the $(\ell+1)$-th transition phase begins. We note that the end of the $\ell$-th transition phase coincides with the beginning of the $(\ell+1)$-th transition phase.

\subsection{Preprocessing Phase}
During the preprocessing phase, the following steps are performed. First, the $0$-th execution set $U_0$ is initialized to $P$. Then, the procedure \texttt{RebuildFromLayer}($0$) is executed, producing a set~$S$. Our solution~$\hat{S}$ and the target sets $S^{(-1)}$ and $S^{(0)}$ are all initialized to $S$. All assignments $\sigma^{(-1)}, \sigma^{(0)}$, and $\hat{\sigma}$ are initialized to the assignment $\sigma$ produced by~\texttt{RebuildFromLayer}($0$). The assignment $\hat{\sigma}$ corresponds to the bicriteria approximate solution $\hat{S}$. Next, the transition phase index $\ell$ is initialized to $0$, and the inserted set $I^{(-1)}$ is initialized to the empty set. Finally, the $0$-th transition phase begins. 

\subsection{Analysis of the Fully Dynamic Algorithm with Worst-Case Recourse}
In this section, we prove Theorem~\ref{thm:bicr_alg_worst_case} by analyzing our fully dynamic bicriteria approximation algorithm, described in~\cref{alg:aux_rebuild_proced,alg:point_ins,alg:point_del,alg:apply_point_repl,alg:lazy_sync}. The analysis is divided into five subsections, as outlined below.
First, in~\cref{sec:size_bicr} we analyze the size of our bicriteria approximate solution $\hat{S}$. Next, in~\cref{sec:upper_bound} we derive an upper bound on the cost of $\hat{S}$. Then, in~\cref{sec:rel_upp_low_bound} we establish a relationship between this upper bound and certain values $\mu(\cdot, \cdot)$. In~\cref{sec:lower_bound}, we associate these values $\mu(\cdot, \cdot)$ with a lower bound on the optimal $k$-center cost. Finally, in~\cref{sec:finish_proof} we conclude the proof of Theorem~\ref{thm:bicr_alg_worst_case}.

\subsubsection{Size of the Bicriteria Approximate Solution} \label{sec:size_bicr}

\begin{observation} \label{obs:size_Ut_upper_bound}
    After an adversarial update, it holds that $n_t \leq \alpha \myhs k \myhs \log n \myhs \log \frac{n}{k}$.
\end{observation}
\begin{proof}
The value of $n_t$ is the size of $U_t$ immediately after the execution of \texttt{RebuildFromLayer}$(\hspace{0.5pt})$. Based on~\linescref{algline:while_rebuild}{line:last_level} in~\cref{alg:aux_rebuild_proced}, we can deduce that $n_t \leq \alpha \myhs k \myhs \log n \myhs \log \frac{n}{k}$. 
\end{proof}

\begin{observation}[Lemma 3.2 in~\cite{bhattacharya2023fully}] \label{obs:last_level_t}
    After an adversarial update, it holds that $t = O(\log\frac{n}{k})$.
\end{observation}
\begin{proof}
    See Appendix B of~\cite{bhattacharya2023fully}.
\end{proof}

\begin{observation} \label{obs:size_S_rebuild} 
    Immediately after the execution of \texttt{RebuildFromLayer}$(\hspace{0.5pt})$, it holds with high probability that $|S| \leq 3 \alpha \myhs k \myhs \log n \myhs \log \frac{n}{k}$.
\end{observation}
\begin{proof}
    Observe that $S = \bigcup_{i=0}^{t-1} S_i \cup S_t$ (see~\linecref{line:final_S} in~\cref{alg:aux_rebuild_proced}). With high probability, it holds that $|S_i| \leq 2 \alpha \myhs k \log n$ for every level $i < t$ (by applying a Chernoff bound). In addition, we have $|S_t| = n_t \leq \alpha \myhs k \log n \myhs \log \frac{n}{k}$. Since $t = O(\log \frac{n}{k})$ by~\cref{obs:last_level_t}, the claim follows.
\end{proof}

\begin{lemma} \label{lem:len_diff_S}
    Between two consecutive executions of \texttt{RebuildFromLayer}$(\hspace{0.5pt})$, at most $\lambda \myhs \alpha \myhs k \log n \myhs \log \frac{n}{k}$ adversarial point insertions occur.
\end{lemma}
\begin{proof}
    Let $\tau$ be a moment at which \texttt{RebuildFromLayer}$(\hspace{0.5pt})$ is executed.
    Consider the last level $t$ and the $t$-th execution set $U_t$. 
    Based on~\cref{obs:size_Ut_upper_bound}, at the $\tau$-th update we have $n_t = |U_t| \leq \alpha \myhs k \log n \myhs \log \frac{n}{k}$. Hence, the threshold for the $t$-th counter $cnt_t$ is at most $\lambda \myhs n_t \leq \lambda \myhs \alpha \myhs k \log n \myhs \log \frac{n}{k}$. When $cnt_t$ exceeds its threshold, \texttt{RebuildFromLayer}$(t)$ is executed (see~\linescref{algline:check_thres}{algline:exec_rebuild} in~\cref{alg:aux_rebuild_proced}). Observe that $cnt_t = 0$ at the $\tau$-th update (see~\linecref{line:last_level} in~\cref{alg:aux_rebuild_proced}), and that $cnt_t$ is incremented by one with each adversarial point insertion (see~\linecref{algline:increm_cnt_pnt_ins} in~\cref{alg:point_ins}). As a result, after at most $\lambda \myhs \alpha \myhs k \log n \myhs \log \frac{n}{k}$ adversarial point insertions,\antonis{This should be the old $n$.} another execution of \texttt{RebuildFromLayer}$(\hspace{0.5pt})$ occurs.
\end{proof}

\begin{lemma} \label{lem:size_targ_set}
    For every transition phase index $\ell \geq 0$,\footnote{The symbol $\ell$ is used to denote the current transition phase index in the pseudocode. For simplicity, we reuse $\ell$ in the analysis to range over all transition phases encountered by the algorithm.} it holds that $|S^\ellexp| \leq (3 + \lambda) \myhs \alpha \myhs k \log n \myhs \log \frac{n}{k}$.
\end{lemma}
\begin{proof}
    Consider an arbitrary transition phase index $\ell \geq 0$. The target set $S^\ellexp$ is the set $S$ at the beginning of the $\ell$-th transition phase (see~\linecref{algline:target_set_ell_init_S} in~\cref{alg:lazy_sync}). Note that $S^\ellexp$ is a non-rebuilt set and remains unchanged by adversarial point insertions. Moreover, while adversarial point deletions can modify $S^\ellexp$, its size cannot increase (see~\linescref{algline:repl_in_S_ell-1}{algline:repl_in_S_ell} in~\cref{alg:point_del}, and~\linescref{algline:remove_p_tilde_S}{algline:add_p_tilde_S} in~\cref{alg:apply_point_repl}). As a result, it follows that $|S^\ellexp|$ is upper bounded by the maximum size of $S$ at any moment.
    
    By~\cref{obs:size_S_rebuild}, immediately after the execution of \texttt{RebuildFromLayer}$(\hspace{0.5pt})$, the size of $S$ is reset and upper bounded by $3\alpha \myhs k \log n \myhs \log \frac{n}{k}$.
    Moreover based on~\cref{lem:len_diff_S}, at most $\lambda \myhs \alpha \myhs k \log n \myhs \log \frac{n}{k}$ new points are added to $S$ between two consecutive executions of \texttt{RebuildFromLayer}$(\hspace{0.5pt})$. Since under point deletions the size of $S$ does not increase (see~\linecref{algline:repl_in_S} in~\cref{alg:point_del}, and~\linescref{algline:remove_p_tilde_S}{algline:add_p_tilde_S} in~\cref{alg:apply_point_repl}), the maximum size of~$S$ at any moment is upper bounded by $3\alpha \myhs k \log n \myhs \log \frac{n}{k} + \lambda \myhs \alpha \myhs k \log n \myhs \log \frac{n}{k}$. Therefore, we conclude that $|S^\ellexp| \leq (3 + \lambda) \myhs \alpha \myhs k \log n \myhs \log \frac{n}{k}$.
\end{proof}

\begin{lemma} \label{lem:len_trans}
    For every transition phase index $\ell \geq 0$, the number of adversarial updates in the $\ell$-th transition phase is at most: 
    \[\big(|S^{(\ell-1)}| + |I^{(\ell-1)}| + |S^\ellexp|\big) \cdot \frac{\lambda^2}{16}.\]
\end{lemma}
\begin{proof}
    The $\ell$-th transition phase consists of two subphases.
    During the first subphase, the target set $S^{\ellexp}$ is added to $\hat{S}$. During the second subphase, the target set $S^{(\ell-1)}$ and the inserted set $I^{(\ell-1)}$ are both removed from $\hat{S}$. Hence, the total relevant changes to the solution $\hat{S}$ during the first and second subphases are at most $|S^{(\ell-1)}| + |I^{(\ell-1)}| + |S^\ellexp|$. Since the solution $\hat{S}$ changes by $\frac{32}{\lambda^2}$ points per adversarial update, the number of adversarial updates required to perform these relevant changes is at most:
    \begin{align*}
        \Big\lceil \big(|S^{(\ell-1)}| + |I^{(\ell-1)}| + |S^\ellexp|\big) \cdot \frac{\lambda^2}{32} \Big\rceil &\;\leq\; \big(|S^{(\ell-1)}| + |I^{(\ell-1)}| + |S^\ellexp|\big) \cdot \frac{\lambda^2}{32} \;+\; 1 \\
        &\;\leq\; 2 \, \big(|S^{(\ell-1)}| + |I^{(\ell-1)}| + |S^\ellexp|\big) \cdot \frac{\lambda^2}{32} \\
        &\;=\; \big(|S^{(\ell-1)}| + |I^{(\ell-1)}| + |S^\ellexp|\big) \cdot \frac{\lambda^2}{16}.
    \end{align*}
    For the second inequality, we assumed that $1 \leq |S^{(\ell-1)}| + |I^{(\ell-1)}| + |S^\ellexp|\big) \cdot \frac{\lambda^2}{32}$; otherwise, the number of adversarial updates in the $\ell$-th transition phase would be a constant, resulting in a stronger guarantee without affecting the validity of our results.
\end{proof}

The next lemma states an upper bound on the \emph{length of each transition phase}, defined as the number of adversarial point updates required until the solution $\hat{S}$ becomes the union of the target set and the inserted set of the current phase.

\begin{lemma} \label{lem:num_updates_phase}
    For every transition phase index $\ell \geq 0$, there are at most $\lambda \myhs \alpha \myhs k \log n \myhs \log \frac{n}{k}$ adversarial updates during the $\ell$-th transition phase.
\end{lemma}
\begin{proof}
    The proof proceeds by induction on the transition phases. The base case is before the beginning of the $0$-th transition phase, when the claim trivially holds. The induction step is for the transition phase index~$\ell \geq 0$.
    Notice that the number of inserted points during any transition phase is trivially upper bounded by the number of adversarial updates within that phase. Hence by the induction hypothesis, we have $|I^{(\ell-1)}| \leq \lambda \myhs \alpha \myhs k \log n \myhs \log \frac{n}{k}$ (note that $|I^{(-1)}| = 0$). Based on~\cref{lem:size_targ_set}, it holds that $|S^{(\ell-1)}| \leq (3 + \lambda) \myhs \alpha \myhs k \log n \myhs \log \frac{n}{k}$ and $|S^\ellexp| \leq (3 + \lambda) \myhs \alpha \myhs k \log n \myhs \log \frac{n}{k}$. Combining this with~\cref{lem:len_trans} and the fact that $0 < \lambda < 1$, the number of adversarial updates during the $\ell$-th transition phase is at most:
    \begin{align*}
        \big(|S^{(\ell-1)}| + |I^{(\ell-1)}| + |S^\ellexp|\big) \cdot \frac{\lambda^2}{16} &\;\leq\; \big((3 + \lambda) \myhs \alpha \myhs k \log n \myhs \log \frac{n}{k} + \lambda \myhs \alpha \myhs k \log n \myhs \log \frac{n}{k} + (3 + \lambda) \myhs \alpha \myhs k \log n \myhs \log \frac{n}{k}\big) \cdot \frac{\lambda^2}{16} \\
        &\;\leq\; \big(2\myhs (3 + \lambda) \myhs \alpha \myhs k \log n \myhs \log \frac{n}{k} + \lambda \myhs \alpha \myhs k \log n \myhs \log \frac{n}{k}\big) \cdot \frac{\lambda^2}{16} \\
        &\;\leq\; \big((6 + 3\lambda) \myhs \alpha \myhs k \log n \myhs \log \frac{n}{k}\big) \cdot \frac{\lambda^2}{16} \\
        &\;\leq\; \lambda \myhs \alpha \myhs k \log n \myhs \log \frac{n}{k}.
    \end{align*}
\end{proof}

At the end of the $\ell$-th transition phase, we have 
$\hat{S} = S^{\ellexp} \cup I^\ellexp$.
The union of the target set $S^{\ellexp}$ and the inserted set $I^\ellexp$ is denoted by $\hat{S}^\ellexp \coloneqq S^{\ellexp} \cup I^\ellexp$. We remark that the set $\hat{S}^\ellexp$ is a dynamic set.

\begin{lemma} \label{lem:sizeS_trans_phase}
    For every transition phase index $\ell \geq 0$, it holds that $|\hat{S}^\ellexp| \leq (3 + 2\lambda) \myhs \alpha \myhs k \log n \myhs \log \frac{n}{k}$.
\end{lemma}
\begin{proof}
    Consider an arbitrary transition phase index $\ell \geq 0$.  By~\cref{lem:size_targ_set}, we have $|S^\ellexp| \leq (3 + \lambda) \myhs \alpha \myhs k \log n \myhs \log \frac{n}{k}$. Moreover, since the size of $I^\ellexp$ is at most the number of adversarial updates within the $\ell$-th transition phase, it follows from~\cref{lem:num_updates_phase} that $|I^\ellexp| \leq \lambda \myhs \alpha \myhs k \log n \myhs \log \frac{n}{k}$. Therefore, it holds that:
    \begin{align*}
        |\hat{S}^\ellexp| \;=\; |S^{\ellexp}| + |I^\ellexp| 
        &\;\leq\; (3 + \lambda) \myhs \alpha \myhs k \log n \myhs \log \frac{n}{k} + \lambda \myhs \alpha \myhs k \log n \myhs \log \frac{n}{k} \\ &\;=\; (3 + 2\lambda) \myhs \alpha \myhs k \log n \myhs \log \frac{n}{k}.
    \end{align*}
\end{proof}

\begin{lemma} \label{lem:size_of_S_hat}
    After an adversarial update, it holds that $|\hat{S}| \leq (6 + 4\lambda) \myhs \alpha \myhs k \log n \myhs \log \frac{n}{k}$. 
\end{lemma}
\begin{proof}
    During an arbitrary $\ell$-th transition phase where $\ell \geq 0$, the maximum size of $\hat{S}$ is upper bounded by $|\hat{S}^{(\ell-1)}| + |\hat{S}^\ellexp|$ as a result of its two subphases. Thus by~\cref{lem:sizeS_trans_phase}, it holds that $|\hat{S}|$ is always at most $(6 + 4\lambda) \myhs \alpha \myhs k \log n \myhs \log \frac{n}{k}$.
\end{proof}

\subsubsection{Upper Bound on the Cost of the Bicriteria Approximate Solution} \label{sec:upper_bound}
For a fixed transition phase index $\ell \geq 0$, let $P^\ellexp, \nu^\ellexp_i$, and $ U^\ellexp_i$ denote the point set, the $i$-th radius, and the $i$-th execution set respectively, at the beginning of the $\ell$-th transition phase. The maximum $\nu^\ellexp_i$ is denoted by $\nu^\ellexp \coloneqq \max_{i \in [0, t]} \nu_i^\ellexp$. Recall also the definition of the following dynamic sets: The target set $S^\ellexp$ is the dynamic non-rebuilt set initialized to $S$ at the beginning of the $\ell$-th transition phase. The inserted set $I^\ellexp$ is the dynamic set of all points inserted during the $\ell$-th transition phase. The set $\hat{S}^\ellexp$ is the dynamic union of~$S^\ellexp$ and $I^\ellexp$.

\begin{restatable}{lemma}{firstSecSubphase} \label{lem:first_sec_subphase}
    For every transition phase index $\ell \geq 0$, the following three properties are satisfied:
    \begin{enumerate}[label=$P_\arabic*$.]
        \item During the first subphase of $\ell$, it holds that $\hat{S}^{(\ell-1)} \subseteq \hat{S}$.

        \item During the second subphase of $\ell$, it holds that $\hat{S}^\ellexp \subseteq \hat{S}$.

        \item During the $\ell$-th transition phase, it holds that $I^\ellexp \subseteq \hat{S}$.
    \end{enumerate}
\end{restatable}
\begin{proof}
    Consider a fixed transition phase index $\ell \geq 0$.
    The end of the $(\ell-1)$-th transition phase coincides with the beginning of the $\ell$-th transition phase. Namely at the beginning of the $\ell$-th transition phase, by an induction argument we have $\hat{S}^{(\ell-1)} \subseteq \hat{S}$, where $\hat{S}^{(\ell-1)} = S^{(\ell-1)} \cup I^{(\ell-1)}$. 
    
    \subparagraph{Validity of $P_1$.}
    During the $\ell$-th transition phase, the set $\hat{S}^{(\ell-1)}$ is not affected by adversarial point insertions. During the first subphase of $\ell$, no points are removed from the solution $\hat{S}$ due to~\texttt{LazySync}$(\hspace{0.5pt})$ (in~\cref{alg:lazy_sync}). Hence during the first subphase of $\ell$, the algorithm removes a point $p$ from $\hat{S}$ only if the adversary deletes $p$ from the point set $P$. As a consequence, the validity of $P_1$ holds under adversarial point insertions.

    To analyze adversarial point deletions, consider a deleted point $p$. Observe that the target set $S^{(\ell-1)}$ is modified only if the deleted point $p$ belongs to $S^{(\ell-1)}$ (see~\linecref{algline:repl_in_S_ell-1} in~\cref{alg:point_del}, and \linescref{algline:if_remove_p_tilde_S}{algline:remove_p_tilde_S} in~\cref{alg:apply_point_repl}). Then, another point $c$ (if it exists) replaces $p$ in $S^{(\ell-1)}$ (see~\linecref{algline:add_p_tilde_S} in~\cref{alg:apply_point_repl}). Since at the beginning of the $\ell$-th transition phase we have $S^{(\ell-1)} \subseteq \hat{S}$, by an induction argument we can infer that the deleted point $p$ belongs to the solution $\hat{S}$ as well. As in this case $\tilde{\sigma} = \sigma^{(\ell-1)}$ and $p \in \hat{S}$, the same point $c$ replaces $p$ in $\hat{S}$ too (see~\linescref{algline:remove_p_hat_S}{algline:add_p_hat_S} in~\cref{alg:apply_point_repl}). Moreover, the solution $\hat{S}$ is potentially modified due to~\linecref{algline:repl_in_S_ell} in~\cref{alg:point_del} and \linecref{algline:remove_p_hat_S} in~\cref{alg:apply_point_repl}, where the target set $S^\ellexp$ is handled. However, since in~\cref{alg:point_del} the operation at~\linecref{algline:repl_in_S_ell-1} for $S^{(\ell-1)}$ is executed before the operation at~\linecref{algline:repl_in_S_ell} for $S^\ellexp$ during the first subphase of $\ell$, the deleted point $p$ must appear more than once in the multiset $\hat{S}$. Therefore, a potential update to $\hat{S}$ due to~\linecref{algline:repl_in_S_ell} in~\cref{alg:point_del} does not affect the relation $\hat{S}^{(\ell-1)} \subseteq \hat{S}$.
    
    Furthermore, note that due to~\linescref{algline:I_ell-1_subset_P}{algline:hat_S_subset_P} in~\cref{alg:point_del}, the inserted set $I^{(\ell-1)}$ remains a subset of the solution $\hat{S}$ during the first subphase of $\ell$. As a result, the set $\hat{S}^{(\ell-1)} = I^{(\ell-1)} \cup S^{(\ell-1)}$ remains a subset of the bicriteria approximate solution $\hat{S}$ during the first subphase of $\ell$, and thus the validity of $P_1$ follows. 

    \subparagraph{Setup for $P_2$ and validity of $P_3$.}
    Before proceeding with the validity of $P_2$, we prove that the relation $\hat{S}^\ellexp \subseteq \hat{S}$ holds by the end of the first subphase of $\ell$. By ignoring the adversarial updates, note that the first subphase of $\ell$ ends once the target set $S^\ellexp$ becomes a subset of the solution $\hat{S}$. For this reason, we analyze the relation $\hat{S}^\ellexp \subseteq \hat{S}$ under adversarial updates. Observe that the target set $S^\ellexp$ is not affected by adversarial point insertions. Also, every newly inserted point is added to both the inserted set $I^\ellexp$ and the solution $\hat{S}$. Thus under adversarial point insertions, the relation $\hat{S}^\ellexp \subseteq \hat{S}$ holds by the end of the first subphase of $\ell$.  

    To analyze adversarial point deletions, consider a deleted point $p$. The argument is similar to the reasoning for the validity of $P_1$. In particular, if the deleted point $p$ belongs to the target set $S^\ellexp$ then it is replaced by another point $c$ (if it exists) in $S^\ellexp$. If the deleted point $p$ belongs to the solution $\hat{S}$ as well, then the algorithm applies the same replacement to the solution $\hat{S}$ too (see~\linecref{algline:repl_in_S_ell} in~\cref{alg:point_del} and the procedure \texttt{ApplyPointReplacement}$(\hspace{0.5pt})$).\footnote{During the execution of~\texttt{LazySync}$(\hspace{0.5pt})$, the deleted point $p$ may not yet have been added to $\hat{S}$; however, by the end of the first subphase of $\ell$ its replacement $c$ enters $\hat{S}$ so that $S^\ellexp \subseteq \hat{S}$.} Moreover, the solution $\hat{S}$ is potentially modified due to~\linecref{algline:repl_in_S_ell-1} in~\cref{alg:point_del} and \linecref{algline:remove_p_hat_S} in~\cref{alg:apply_point_repl}, where the target set $S^{(\ell-1)}$ is handled. 
    Based on the validity of $P_1$ and the construction, if the deleted point $p$ belongs to $S^{(\ell-1)}$ then $p$ belongs to $\hat{S}$ as an occurrence of~$S^{(\ell-1)}$. This implies that $p$ as an occurrence of~$S^\ellexp$ is replaced in $\hat{S}$ only with the same point $c$ (if it exists) in~$S^\ellexp$. Hence,
    whenever the solution $\hat{S}$ is updated because of both target sets $S^{(\ell-1)}$ and $S^\ellexp$, the deleted point $p$ must have appeared more than once in the multiset $\hat{S}$. Therefore, a potential update to $\hat{S}$ due to~\linecref{algline:repl_in_S_ell-1} in~\cref{alg:point_del} does not affect the relation $\hat{S}^\ellexp \subseteq \hat{S}$.
    
    Furthermore, note that according to~\linescref{algline:I_ell_subset_P}{algline:hat_S_subset_P} in~\cref{alg:point_del}, the inserted set $I^\ellexp$ remains a subset of the solution $\hat{S}$ during the $\ell$-th transition phase, which in turn proves the validity of $P_3$. As a consequence, the set $\hat{S}^\ellexp$ becomes a subset of the bicriteria approximate solution $\hat{S}$ by the end of the first subphase of $\ell$. 
    
    \paragraph{Validity of $P_2$.}
    The end of the first subphase of $\ell$ coincides with the beginning of the second subphase of~$\ell$. Namely at the beginning of the second subphase of $\ell$, we have $\hat{S}^\ellexp \subseteq \hat{S}$. By ignoring the adversarial updates, observe that during the second subphase of $\ell$, none of the points in $\hat{S}^\ellexp = S^\ellexp \cup I^\ellexp$ are removed from $\hat{S}$ due to~\texttt{LazySync}$(\hspace{0.5pt})$ (see~\linescref{algline:del_set}{algline:rem_del_set} in~\cref{alg:lazy_sync}).\footnote{We also remark that each point is considered unique.} The argument regarding adversarial updates follows similar reasoning to that in the setup for $P_2$. 
    
    Notice that during the second subphase of $\ell$, the solution $\hat{S}$ is not updated because of the previous target set $S^{(\ell-1)}$ (see~\linecref{algline:if_S_ell-1} in~\cref{alg:point_del}). This is important, as the validity of $P_1$ does not hold during the second subphase of $\ell$; hence, the adversarially deleted point in $\hat{S}$ is treated as an occurrence of $S^\ellexp$. Therefore, we conclude that the set $\hat{S}^\ellexp$ remains a subset of the bicriteria approximate solution $\hat{S}$ during the second subphase of $\ell$, and thus the validity of $P_2$ follows.
\end{proof}

\begin{lemma}[\cite{bhattacharya2023fully,bhattacharya2025alm_opt_kcenter}] \label{lem:approx_of_S}
    After an adversarial update, for every point $p \in P$ it holds that $\dist(p, S) \,\leq\, 2 \myhs \max_{i \in [0, t]} \nu_i$.
\end{lemma}

\begin{lemma} \label{lem:dist_p_S_subphases}
    For every transition phase index $\ell \geq 0$, every point $p \in P$ satisfies the following three properties:
    \begin{enumerate}
        \item During the first subphase of $\ell$, it holds that $\dist(p, \hat{S}) \leq 2 \myhs \nu^{(\ell-1)}$.\footnote{When $\ell = 0$, we set $\nu^{(-1)} = \nu^{(0)}$.}

        \item During the second subphase of $\ell$, it holds that $\dist(p, \hat{S}) \leq 2 \myhs \nu^\ellexp$.

        \item At the end of the second subphase of $\ell$, it holds that $\dist(p, \hat{S}^\ellexp) \leq 2 \myhs \nu^\ellexp$.
    \end{enumerate}
\end{lemma}
\begin{proof}
    The proof proceeds by induction on the transition phases. The base case is before the beginning of the $0$-th transition phase, when all three properties are trivially satisfied. The induction step is for the transition phase index $\ell \geq 0$. Note that due to the third property $P_3$ of~\cref{lem:first_sec_subphase}, all points inserted into the point set $P$ during the $\ell$-th transition phase, belong to the solution $\hat{S}$ throughout the entire phase. It remains to consider a point $p \in P$ that was inserted before the $\ell$-th transition phase. The analysis for this fixed point $p$ is divided into two cases, depending on which subphase is active:
    \begin{itemize}
        \item Assume that the first subphase of $\ell$ is active. The end of the second subphase of the $(\ell-1)$-th transition phase coincides with the beginning of the first subphase of the $\ell$-th transition phase. 
        This implies that the point $p$ had been inserted by the end of the $(\ell-1)$-th transition phase. Hence
        by the induction hypothesis, it holds that $\dist(p, \hat{S}^{(\ell-1)}) \leq 2 \nu^{(\ell-1)}$ at the beginning of the $\ell$-th transition phase. Based on~\cref{lem:first_sec_subphase} we have $\hat{S}^{(\ell-1)} \subseteq \hat{S}$ during the first subphase of $\ell$, which means that $\dist(p, \hat{S}) \leq \dist(p, \hat{S}^{(\ell-1)})$ during the first subphase of $\ell$. In turn, we obtain $\dist(p, \hat{S}) \leq \dist(p, \hat{S}^{(\ell-1)}) \leq 2 \nu^{(\ell-1)}$ at the beginning of the $\ell$-th transition phase. However during the first subphase of $\ell$, we must be careful with adversarial point deletions, as the closest point in $\hat{S}$ to $p$ may move farther away. 

        Notice that each cluster induced by $\sigma$ in~\linecref{algline:sigma_in_rebuild} of~\cref{alg:aux_rebuild_proced} has radius $\nu_j$ and diameter $2\nu_j$ by the triangle inequality, where $j \in [0, t]$ is a fixed level. As a result, any replacement within this induced cluster preserves the diameter $2\nu_j$. For an adversarial point deletion, observe that in~\linecref{algline:rem_for_replace} of~\cref{alg:apply_point_repl}, the set $\textit{Rem}$ is defined based on the remaining points inside the cluster induced by $\tilde{\sigma}$. Additionally, in~\linecref{algline:hat_sigma_simul} of~\cref{alg:apply_point_repl}, the assignment $\hat{\sigma}$ of our bicriteria approximate solution $\hat{S}$ simulates $\tilde{\sigma} = \sigma^{(\ell-1)}$. Therefore, it follows that $\dist(p, \hat{S}) \leq 2 \myhs \nu^{(\ell-1)}$ during the first subphase of $\ell$.
        
        \item Assume that the second subphase of $\ell$ is active. Since the point $p$ had been inserted by the end of the $(\ell-1)$-th transition phase, it holds that $p \in P^\ellexp$. The target set $S^\ellexp$
        is the set $S$ when the point set is $P^\ellexp$, and so~\cref{lem:approx_of_S} implies that the point $p$ is within distance $2\myhs \nu^\ellexp$ from the target set $S^\ellexp$. 
        The argument concerning adversarial updates in $S^\ellexp$ follows similar reasoning to that in~$S^{(\ell-1)}$.
        By~\cref{lem:first_sec_subphase}, we have $\hat{S}^\ellexp \subseteq \hat{S}$ during the second subphase of $\ell$, and thus it follows that $\dist(p, \hat{S}) \leq 2 \myhs \nu^\ellexp$, as needed.
    \end{itemize}
    At the end of the $\ell$-th transition phase, we have $\hat{S} = \hat{S}^\ellexp$.
    Consequently, at the end of the second subphase of~$\ell$, it holds that $\dist(p, \hat{S}^\ellexp) \leq 2 \myhs \nu^\ellexp$ as well.
\end{proof}

\begin{corollary} \label{cor:cost_hat_S_nu}
    For every transition phase index $\ell \geq 0$, it holds that $\cost(\hat{S}) \leq 2 \myhs \max(\nu^{(\ell-1)}, \nu^\ellexp)$ during the $\ell$-th transition phase.
\end{corollary}
\begin{proof}
    Based on~\cref{lem:dist_p_S_subphases}, it follows that $\cost(\hat{S}) = \max_{p \in P} \dist(p, \hat{S}) \leq 2 \myhs \max(\nu^{(\ell-1)}, \nu^\ellexp)$.
\end{proof}

\subsubsection{Relationship Between Upper Bound and Lower Bound} \label{sec:rel_upp_low_bound}

In the following definition the value $\mu(\cdot, \cdot)$ is introduced, which is later associated with a lower bound on the optimal $k$-center cost.\footnote{We emphasize that in~\cref{def:mu_i}, the corresponding set $X$ must be a subset of $P$ and not just of $U$.}

\begin{definition} \label{def:mu_i}
    Given a subset of points $U \subseteq P$ and a constant $0 < \zeta < 1$, we define $\mu(U, \zeta)$ as the minimum nonnegative real number such that there exists a subset of points $X \subseteq P$ with $|X| \leq k$, for which the following properties hold:
\begin{enumerate}
    \item $\lvert\ball[U, X, \mu(U, \zeta)]\rvert \geq \zeta \myhs |U|$.
    \item $\lvert U \setminus \ball(U, X, \mu(U, \zeta))\rvert \geq (1 - \zeta) \myhs |U|$. 
\end{enumerate}
\end{definition}

For a fixed transition phase index $\ell \geq 0$, recall that $P^\ellexp, \nu^\ellexp_i$, and $U^\ellexp_i$ denote the point set, the $i$-th radius, and the $i$-th execution set respectively, at the beginning of the $\ell$-th transition phase. Also recall that $\nu^\ellexp \coloneqq \max_{i \in [0, t]} \nu_i^\ellexp$. 
For some $\tau \in \mathbb{N}$ and a fixed level $i \in [0, t]$, let $P^\tauexp$, $U^\tauexp_i$ and $\nu^\tauexp_i$ denote the point set, the $i$-th execution set, and the $i$-th radius respectively, at the $\tau$-th adversarial update. The maximum~$\nu^\tauexp_i$ is denoted by $\nu^\tauexp \coloneqq \max_{i \in [0, t]} \nu_i^\tauexp$.

We fix a positive real number $\gamma$ such that $\beta < \gamma < 1$, where $\beta$ is the parameter associated with all radii~$\nu_i$~(see~\linecref{line:min_rad_nu} in~\cref{alg:aux_rebuild_proced}). The following relationship between $\nu^\tauexp_i$ and  $\mu(U_i^\tauexp, \gamma)$ is derived from the analysis by Mettu and Plaxton in~\cite[Section 3]{mettu2004optimal}.

\begin{lemma}[Lemma 3.3 in~\cite{mettu2004optimal}] \label{lem:nu_2mu}
    Let $\tau_i$ be a moment at which~\texttt{RebuildFromLayer}$(\hspace{0.5pt})$ is invoked and the $i$-th level is scanned, for some level $i \in [0, t]$. Then it holds with high probability that $\nu^{(\tau_i)}_i \leq 2 \myhs \mu(U^{(\tau_i)}_i, \gamma)$.
\end{lemma}

\begin{corollary} \label{cor:rel_nu_mu}
    Let $\tau_i$ be the most recent moment at which~\texttt{RebuildFromLayer}$(\hspace{0.5pt})$ (in~\cref{alg:aux_rebuild_proced}) is invoked and the $i$-th level is scanned, and let $\tau$ be the current moment. Then it holds with high probability that: 
    \[
        \nu^\tauexp \leq 2 \myhs \max_{i \in [0, t]} \mu(U^{(\tau_i)}_i, \gamma) \;\text{ and also }\; \nu^\tauexp \leq 2 \myhs \max_{i \in [0, t)} \mu(U^{(\tau_i)}_i, \gamma).
    \]
\end{corollary}
\begin{proof}
    Since the $i$-th radius $\nu_i$
    is modified only when ~\texttt{RebuildFromLayer}$(\hspace{0.5pt})$ scans the $i$-th level (see~\linecref{line:min_rad_nu} in~\cref{alg:aux_rebuild_proced}), we have
    $\nu_i^\tauexp = \nu_i^{(\tau_i)}$ for every level $i \in [0, t]$.
    By the definition of $\nu^\tauexp$ and~\cref{lem:nu_2mu}, it follows that $\nu^\tauexp = \max_{i \in [0, t]} \nu_i^{(\tau_i)} \leq 2 \myhs \max_{i \in [0, t]} \mu(U^{(\tau_i)}_i, \gamma)$, as needed. 
    Furthermore, by the construction of \texttt{RebuildFromLayer}$(\hspace{0.5pt})$ it holds that $\nu_t^\tauexp = 0$ (see~\linecref{line:last_level} in~\cref{alg:aux_rebuild_proced}), and thus $\nu^\tauexp = \max_{i \in [0, t)} \nu_i^{(\tau_i)} \leq 2 \myhs \max_{i \in [0, t)} \mu(U^{(\tau_i)}_i, \gamma)$, as required. Note that the high-probability guarantee of~\cref{lem:nu_2mu} for each individual level extends to all $t$ levels by a union bound, since $t = O(\log\frac{n}{k})$ by~\cref{obs:last_level_t}.
\end{proof}

Even though~\cref{cor:rel_nu_mu} provides an upper bound for $\nu^\tauexp$,
the cost of the bicriteria approximate solution $\hat{S}$ is upper bounded with respect to $\max(\nu^{(\ell-1)}, \nu^\ellexp)$ (see~\cref{cor:cost_hat_S_nu}). The $i$-th radius $\nu^\ellexp_i$ (defined at the beginning of the $\ell$-th transition phase) is the $i$-th radius $\nu^\tauexp_i$ from some (earlier) $\tau$-th adversarial update. Observe that the $i$-th execution set $U_i^\ellexp$ at the beginning of the $\ell$-th transition phase can be different from the (earlier) $i$-th execution set $U_i^\tauexp$ at the $\tau$-th adversarial update (i.e., at time/moment $\tau$). In turn, the corresponding values $\mu(\cdot, \cdot)$---which are
associated with a lower bound on the optimal $k$-center cost---can change.

In the next subsection, we exploit the fact that each $i$-th execution set $U_i$ can remain lazy for a number of adversarial updates proportional to its size (i.e., for $\Theta(|U_i|)$ adversarial updates), in order to provide the desired bounds. Namely, we prove that the corresponding values $\mu(\cdot, \cdot)$ change only slightly in a desired way.

\subsubsection{Lower Bound on the Optimal Solution} \label{sec:lower_bound}

Based on the discussion at the end of~\cref{sec:rel_upp_low_bound}, the value of $\nu^\ellexp$ (defined at the beginning of the $\ell$-th transition phase) equals the value of $\nu^\tauexp$ from some (earlier) $\tau$-th adversarial update. For a fixed level~$i \in [0, t]$, let $\tau_i^\ellexp$ be the adversarial update at which $\nu_i^\ellexp$ is computed, that is, $\nu_i^\ellexp = \nu_i^{(\tau_i^\ellexp)}$.\footnote{Given the way we use the verbs \emph{computed} and \emph{defined}, note that the value of $\nu_i^\ellexp$ may be computed before $\nu_i^\ellexp$ is defined.} Then the set $U_i^{(\tau_i^\ellexp)}$ denotes the $i$-th the execution set at the moment $\nu_i^\ellexp$ is computed. To simplify the notation, let $\hat{U}_i^\ellexp \coloneqq U_i^{(\tau_i^\ellexp)}$. Observe that $\hat{U}_i^\ellexp$ is defined immediately after an execution of~\texttt{RebuildFromLayer}$(\hspace{0.5pt})$ (in~\cref{alg:aux_rebuild_proced}) that scans the~$i$-th level (at time~$\tau_i^\ellexp$), which is when the value of $\nu_i^\ellexp$ is computed.

At the beginning of the $\ell$-th transition phase, the $i$-th execution set is denoted by $U_i^\ellexp$. During the $\ell$-th transition phase, the procedure \texttt{RebuildFromLayer}$(\hspace{0.5pt})$ can be executed within \texttt{CheckThresholdRebuild}$(\hspace{0.5pt})$ in~\cref{alg:aux_rebuild_proced}. As a result, the execution sets and the value of the last level $t$ can completely change. This means that there may no longer be a direct correspondence between $U_i^\ellexp$, $\hat{U}_i^\ellexp$ and the currently maintained $i$-th execution set $U_i$. Since the upper bound on the cost of the bicriteria approximate solution $\hat{S}$ depends on $\nu^\ellexp$ (see~\cref{lem:dist_p_S_subphases,cor:cost_hat_S_nu}), our goal is to show that maintaining $\hat{U}_i^\ellexp$ in a lazy manner (without rebuilding it) affects the approximation ratio only up to a constant factor. 

For this reason, we define $\lazy(\hat{U}_i^\ellexp)$ for a fixed transition phase index $\ell \geq 0$ and a fixed level $i \in [0, t]$, as follows:
\begin{itemize}
    \item $\lazy(\hat{U}_i^\ellexp)$ is initialized to $\hat{U}_i^\ellexp$ at time $\tau_i^\ellexp$.
    \item $\lazy(\hat{U}_i^\ellexp)$ is subject to changes due to~\linecref{algline:add_to_Ui} in~\cref{alg:point_ins} and~\linecref{algline:remove_from_Uj} in~\cref{alg:point_del} (i.e., subject to adversarial point updates). In other words, any newly inserted point is added to $\lazy(\hat{U}_i^\ellexp)$, and any deleted point that belongs to $\lazy(\hat{U}_i^\ellexp)$ is removed from it.
\end{itemize}
We remark that $\lazy(\hat{U}_i^\ellexp)$ is an auxiliary non-rebuilt set (see~\cref{def:non-rebuilt}) that can be thought of as the current $U_i$, assuming \texttt{RebuildFromLayer}$(\hspace{0.5pt})$ has not been executed since the $\tau_i^\ellexp$-th adversarial update. 

\begin{observation} \label{obs:equal_lazy_sets} 
    For every transition phase index $\ell \geq 0$, consider a fixed level $i \in [0, t]$. Then~$\lazy(\hat{U}_i^\ellexp)$ is the same set as $\lazy(U_i^\ellexp)$ (i.e., $\lazy(\hat{U}_i^\ellexp) = \lazy(U_i^\ellexp)$).
\end{observation}
\begin{proof}
    The moment $\hat{U}_i^\ellexp$ is defined is the most recent adversarial update at or before the beginning of the $\ell$-th transition phase at which an execution of \texttt{RebuildFromLayer}$(\hspace{0.5pt})$ (in~\cref{alg:aux_rebuild_proced}) scans the~$i$-th level.
    Hence it holds that $U_i^\ellexp = \lazy(\hat{U}_i^\ellexp)$, and by the definition of $\lazy(\cdot)$ it follows that $\lazy(U_i^\ellexp) = \lazy(\hat{U}_i^\ellexp)$ as well.\footnote{The definition of~$\lazy()$ is naturally extended to any set $U$; see also~\cref{def:lazy_set}, where we reuse this concept.}
\end{proof}

The following lemma states that the size of the symmetric difference between the execution set $\hat{U}_i^\ellexp$ (defined at the moment $\nu_i^\ellexp$ is computed) and the execution set $U_i^\ellexp$ (defined at the moment  $\nu_i^\ellexp$ is defined, which occurs at the beginning of the $\ell$-th transition phase) is upper bounded by $\lambda \myhs |\hat{U}_i^\ellexp|$. 

\begin{lemma} \label{lem:sym_Ui_upper_bound}
    For every transition phase index $\ell \geq 0$, consider a fixed level $i \in [0, t]$.
    Then it holds that $|\hat{U}_i^\ellexp \triangle \myhs U^\ellexp_i| \leq \lambda \myhs |\hat{U}_i^\ellexp|$.
\end{lemma}
\begin{proof}
    Consider a fixed level $i \in [0, t]$, and recall that~$\tau_i^\ellexp$ denotes the moment when $\hat{U}_i^\ellexp$ is defined (i.e., the moment when the value of $\nu_i^\ellexp$ is computed). Notice that $\tau_i^\ellexp$ is the most recent moment at or before the beginning of the $\ell$-th transition phase at which an execution of \texttt{RebuildFromLayer}$(\hspace{0.5pt})$ (in~\cref{alg:aux_rebuild_proced}) scans the~$i$-th level. At time~$\tau_i^\ellexp$, the $i$-th counter $\textit{cnt}_i$ is initialized to zero (see~\linecref{algline:counter_init_zero} in~\cref{alg:aux_rebuild_proced}) and its threshold becomes $\lambda \myhs n_i = \lambda \myhs |\hat{U}_i^\ellexp|$ (see~\linescref{algline:nj_init_size_of_Uj}{algline:check_thres} in~\cref{alg:aux_rebuild_proced}). Observe that $\textit{cnt}_i$ is incremented by one with each adversarial point update that modifies the $i$-th execution set $U_i$~(see~\linescref{algline:add_to_Ui}{algline:increm_cnt_pnt_ins} in~\cref{alg:point_ins} and~\linescref{algline:remove_from_Uj}{algline:increm_cnt_pnt_del} in~\cref{alg:point_del}). 
    According to~\linescref{algline:check_thres}{algline:exec_rebuild} in~\cref{alg:aux_rebuild_proced} (see~\texttt{CheckThresholdRebuild}$(\hspace{0.5pt})$), after at most $\lambda \myhs n_i = \lambda \myhs |\hat{U}_i^\ellexp|$ changes to the $i$-th execution set $\hat{U}^\ellexp_i$ (i.e., when $\textit{cnt}_i$ exceeds its threshold) the procedure \texttt{RebuildFromLayer}$(i)$ is executed again. 
    
    Suppose to the contrary that $|\hat{U}_i^\ellexp \triangle \myhs U^\ellexp_i| > \lambda \myhs |\hat{U}_i^\ellexp|$.
    Since $U_i^\ellexp$ denotes the $i$-th execution set $U_i$ at the beginning of the $\ell$-th transition phase, there must exist a moment $\tau > \tau_i^\ellexp$ at or before the beginning of the $\ell$-th transition phase at which \texttt{RebuildFromLayer}$(i)$ scans the $i$-th level. However, this contradicts the fact that $\tau_i^\ellexp$ is the most recent such moment at or before the beginning of the $\ell$-th transition phase, and thus the claim follows. 
\end{proof}

For a fixed level $i \in [0, t]$, the value of $n_i$ is the initial size of the $i$-th execution set $U_i$ immediately after an execution of \texttt{RebuildFromLayer}$(\hspace{0.5pt})$ that 
scans the $i$-th level (see~\linecref{algline:nj_init_size_of_Uj} in~\cref{alg:aux_rebuild_proced}).

\begin{observation} \label{obs:Ui_big_rebuild}
    At the moment when $n_i$ is assigned a value, it holds that $n_i \geq \alpha \myhs k \myhs \log n \myhs \log \frac{n}{k}$ for every level~$i \in [0, t)$.\footnote{Note that $n$ (i.e., the number of points) may change over time; hence, we refer specifically to the moment when $n_i$ is assigned a value.}
\end{observation}
\begin{proof}
    For all levels $i \in [0, t-1]$, the claim holds trivially due to the condition in~\linecref{algline:while_rebuild} of~\texttt{RebuildFromLayer}$(\hspace{0.5pt})$ (see~\cref{alg:aux_rebuild_proced}). 
\end{proof}

\begin{observation} \label{obs:sym_diff_prop}
    Given three sets $X, Y, Z$, it holds that
    $|X \triangle\myhs Z| \leq |X \triangle\myhs Y| + |Y \triangle\myhs Z|$.
\end{observation}
\begin{proof}
    The right-hand side of the inequality is evaluated as follows: 
     \begin{align*}
         |X \triangle\myhs Y| + |Y \triangle\myhs Z| &\,\ge\, |X \triangle\myhs Y| + |Y \triangle\myhs Z| - 2 \myhs |(X \triangle\myhs Y) \cap (Y \triangle\myhs Z)|\\
         &\,=\, \big|\big((X \triangle\myhs Y) \cup (Y \triangle\myhs Z)\big) \setminus \big((X \triangle\myhs Y) \cap (Y \triangle\myhs Z)\big)\big|\\
         &\,=\, | (X \triangle\myhs Y) \triangle\myhs (Y \triangle\myhs Z)| \,=\, |X \triangle\myhs Z|,
     \end{align*}
     where the last equality follows by the associativity of $\triangle$ and the fact that $A \triangle \myhs A = \emptyset$ and $A \triangle \myhs \emptyset  = A$ for all sets~$A$. 
\end{proof}

\begin{lemma} \label{lem:symmetric_diff}
    For every transition phase index $\ell \geq 0$, consider a fixed level $i \in [0, t)$ during the $\ell$-th transition phase. Then both of the following inequalities hold:
    \begin{itemize}
        \item $|\hat{U}^{(\ell-1)}_i \triangle \lazy(\hat{U}_i^{(\ell-1)})| \leq 3\myhs \lambda \myhs |\hat{U}^{(\ell-1)}_i|.$

        \item $|\hat{U}^\ellexp_i \triangle \lazy(\hat{U}_i^\ellexp)| \leq 2 \myhs \lambda \myhs |\hat{U}^\ellexp_i|.$
    \end{itemize}
\end{lemma}
\begin{proof}
    Recall that the set $\hat{U}^{(\ell-1)}_i$ denotes the $i$-th execution set at the moment $\nu_i^{(\ell-1)}$ is computed, and that~$U_i^{(\ell-1)}$ denotes the $i$-th execution set at the beginning of the $(\ell-1)$-th transition phase (i.e., at the moment $\nu_i^{(\ell-1)}$ is defined). Based on~\cref{lem:num_updates_phase}, there are at most $\lambda \myhs \alpha \myhs k \log n \myhs \log \frac{n}{k}$ adversarial updates during a transition phase. The statement asserts that the $\ell$-th transition phase is active, which means that there have been at most $2 \myhs \lambda \myhs \alpha \myhs k \log n \myhs \log \frac{n}{k}$ adversarial updates since $U^{(\ell-1)}_i$ was defined. Hence, an upper bound on the size of the symmetric difference between $U^{(\ell-1)}_i$ and $\lazy(U_i^{(\ell-1)})$ is given by $2 \myhs \lambda \myhs \alpha \myhs k \log n \myhs \log \frac{n}{k}$ (the maximum number of adversarial point updates during two transition phases). Moreover by~\cref{lem:sym_Ui_upper_bound} it holds that $|\hat{U}_i^{(\ell-1)} \triangle \myhs U^{(\ell-1)}_i| \leq \lambda \myhs |\hat{U}_i^{(\ell-1)}|$, and thus based on~\cref{obs:sym_diff_prop} we obtain that:
    \begin{align*}
        |\hat{U}^{(\ell-1)}_i \triangle \lazy(U_i^{(\ell-1)})| \;&\leq\; |\hat{U}^{(\ell-1)}_i \triangle\myhs U_i^{(\ell-1)}| \,+\, |U_i^{(\ell-1)} \triangle \lazy(U_i^{(\ell-1)})| \\
        \;&\leq\; \lambda \myhs |\hat{U}_i^{(\ell-1)}| \,+\, 2 \myhs \lambda \myhs \alpha \myhs k \log n \myhs \log \frac{n}{k}.
    \end{align*}
    According to~\cref{obs:equal_lazy_sets} we have $\lazy(\hat{U}_i^{(\ell-1)}) = \lazy(U_i^{(\ell-1)})$, and so it follows that:
    \[
        |\hat{U}^{(\ell-1)}_i \triangle \lazy(\hat{U}_i^{(\ell-1)})| \;\leq\; \lambda \myhs |\hat{U}_i^{(\ell-1)}| \,+\, 2 \myhs \lambda \myhs \alpha \myhs k \log n \myhs \log \frac{n}{k}.
    \]
    Furthermore,~\cref{obs:Ui_big_rebuild} states that $n_i \geq \alpha \myhs k \myhs \log n \myhs \log \frac{n}{k}$, where $i \in [0, t-1]$. By construction and definition, both~$n_i$ and~$\hat{U}_i^{(\ell-1)}$ are defined immediately after an execution of \texttt{RebuildFromLayer}$(\hspace{0.5pt})$ that scans the $i$-th level, which implies that $|\hat{U}_i^{(\ell-1)}| \geq \alpha \myhs k \myhs \log n \myhs \log \frac{n}{k}$.
    Consequently, it holds that $|\hat{U}^{(\ell-1)}_i \triangle \lazy(\hat{U}_i^{(\ell-1)})| \leq 3 \myhs \lambda \myhs |\hat{U}^{(\ell-1)}_i|$, as needed.

    Note that $\hat{U}_i^\ellexp$ and $U_i^\ellexp$ denote more recent $i$-th execution sets compared to the $i$-th execution sets denoted by $\hat{U}_i^{(\ell-1)}$ and $U_i^{(\ell-1)}$. Using similar arguments, an upper bound on the size of the symmetric difference between $U^\ellexp_i$ and $\lazy(U_i^\ellexp)$ is given by the maximum number of adversarial point updates during one transition phase. Therefore, we deduce that $|\hat{U}^\ellexp_i \triangle \lazy(\hat{U}_i^\ellexp)| \leq 2 \lambda \myhs |\hat{U}_i^\ellexp|$, as required (see also Figure~\ref{fig:U_i_example}).
\end{proof}

\vspace{1em}
\begin{figure}[H]
    \centering
    \includegraphics[width=0.5\linewidth]{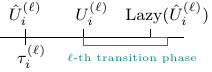}
    \vspace{1em}
    \caption{$\hat{U}_i^\ellexp$ is defined at time $\tau_i^\ellexp$ and $U_i^\ellexp$ is defined at the beginning of the $\ell$-th transition phase. The set $\lazy(\hat{U}_i^\ellexp)$ is depicted at some moment during the $\ell$-th transition phase.}
    \label{fig:U_i_example}
\end{figure}

We derive the following lemma (\cref{lem:rel_mu_old_curr}) using reasoning similar to that of~\cite[Lemma B.4]{bhattacharya2023fully}. However, since we maintain lazier versions of the execution sets $U_i$, the parameters must be adjusted appropriately.
Roughly speaking,~\cref{lem:rel_mu_old_curr} says the following for the execution sets $\hat{U}_i^\ellexp$ and $\lazy(\hat{U}_i^\ellexp)$, whose symmetric difference is bounded as established in~\cref{lem:symmetric_diff}. Assume that there is a sufficiently large ball of some radius in $\lazy(\hat{U}_i^\ellexp)$, defined according to~\cref{def:mu_i} for some $\tilde{\gamma} \geq \gamma$. Then there is also a sufficiently large ball of a similar radius in $\hat{U}_i^\ellexp$ for the smaller parameter $\gamma \leq \tilde{\gamma}$. As a result, the corresponding values of~$\mu(\cdot,\cdot)$ (which are associated with the optimal $k$-center cost) can be compared in a suitable way. We later use this fact to argue that the optimal $k$-center cost does not change significantly during a transition phase.

\begin{lemma} \label{lem:rel_mu_old_curr}
    For every transition phase index $\ell \geq 0$, consider a fixed level $i \in [0, t)$ during the $\ell$-th transition phase. Then for $\tilde{\gamma} \geq \frac{\gamma + 3 \myhs \lambda}{1 - 3 \myhs \lambda}$, $0 < \lambda < \frac{1}{6}$, and $0 < \gamma < 1 - 6\lambda$, it holds that:
    \[
        \mu(\hat{U}^{(\ell-1)}_i, \gamma) \leq 2 \cdot \mu(\lazy(\hat{U}_i^{(\ell-1)}), \tilde{\gamma}) \;\text{ and }\;\mu(\hat{U}^\ellexp_i, \gamma) \leq 2 \cdot \mu(\lazy(\hat{U}_i^\ellexp), \tilde{\gamma}).
    \]
\end{lemma}
\begin{proof}
    Let $X_i \subseteq P$ be the set that satisfies the requirements for $\mu(\lazy(\hat{U}_i^{(\ell-1)}), \tilde{\gamma})$ in~\cref{def:mu_i}. To simplify the notation, let $\hat{\mu}_i^{(\ell-1)} \coloneqq \mu(\lazy(\hat{U}_i^{(\ell-1)}), \tilde{\gamma})$. 
    In turn, we have: \[\lvert\ball[\lazy(\hat{U}_i^{(\ell-1)}), X_i, \hat{\mu}_i^{(\ell-1)}]\rvert \,\geq\, \tilde{\gamma} \myhs \lvert\lazy(\hat{U}_i^{(\ell-1)})\rvert.\]
    Since by~\cref{lem:symmetric_diff} we have $|\hat{U}^{(\ell-1)}_i \triangle \lazy(\hat{U}_i^{(\ell-1)})| \leq 3 \myhs \lambda \myhs |\hat{U}^{(\ell-1)}_i|$, the size of $\ball[\hat{U}_i^{(\ell-1)}, X_i, \hat{\mu}_i^{(\ell-1)}]$ is evaluated as follows:
    \begin{align*}
        \lvert\ball[\hat{U}_i^{(\ell-1)}, X_i, \hat{\mu}_i^{(\ell-1)}]\rvert &\;=\; \lvert\ball[\hat{U}_i^{(\ell-1)} \cup \lazy(\hat{U}_i^{(\ell-1)}), X_i, \hat{\mu}_i^{(\ell-1)}] \setminus \ball[\lazy(\hat{U}_i^{(\ell-1)}) \setminus \hat{U}_i^{(\ell-1)}, X_i, \hat{\mu}_i^{(\ell-1)}]\rvert \\
        &\;\geq \lvert\ball[\lazy(\hat{U}_i^{(\ell-1)}), X_i, \hat{\mu}_i^{(\ell-1)}]\rvert - \lvert\ball[\lazy(\hat{U}_i^{(\ell-1)}) \setminus \hat{U}_i^{(\ell-1)}, X_i, \hat{\mu}_i^{(\ell-1)}]\rvert \\
        &\;\geq\; \tilde{\gamma}\myhs \lvert\lazy(\hat{U}_i^{(\ell-1)})\rvert - \lvert\lazy(\hat{U}_i^{(\ell-1)}) \setminus \hat{U}_i^{(\ell-1)}\rvert
        \;\geq\; \frac{\gamma + 3 \myhs \lambda}{1 - 3 \myhs \lambda} \myhs \lvert\lazy(\hat{U}_i^{(\ell-1)})\rvert - 3 \myhs \lambda \myhs |\hat{U}_i^{(\ell-1)}| \\
        &\;\geq\; \frac{\gamma + 3 \myhs \lambda}{1 - 3 \myhs \lambda} \myhs \big(|\hat{U}_i^{(\ell-1)}| - 3 \myhs \lambda \myhs |\hat{U}_i^{(\ell-1)}|\big) - 3 \myhs \lambda \myhs |\hat{U}_i^{(\ell-1)}| \\
        &\;=\; \gamma \myhs |\hat{U}_i^{(\ell-1)}|.
    \end{align*}
    Notice that $X_i$ may not be a subset of the point set $P^{(\tau_i^{(\ell-1)})}$ at the moment $\hat{U}_i^{(\ell-1)}$ is defined (i.e., at the $\tau_i^{(\ell-1)}$-th adversarial update). However using the triangle inequality, we can find a subset $Y_i$ of~$P^{(\tau_i^{(\ell-1)})}$ such that $\ball[\hat{U}_i^{(\ell-1)}, X_i, \hat{\mu}_i^{(\ell-1)}] \subseteq \ball[\hat{U}_i^{(\ell-1)}, Y_i, 2\myhs\hat{\mu}_i^{(\ell-1)}]$. In other words, there exists a subset of points $Y_i \subseteq~P^{(\tau_i^{(\ell-1)})}$ with $|Y_i| \leq k$ such that $\lvert\ball[\hat{U}_i^{(\ell-1)}, Y_i, 2\myhs\hat{\mu}_i^{(\ell-1)}]\rvert \geq \gamma\myhs|\hat{U}_i^{(\ell-1)}|$. Therefore based on~\cref{def:mu_i} for $\mu(\hat{U}_i^{(\ell-1)}, \gamma)$, this implies that $\mu(\hat{U}_i^{(\ell-1)}, \gamma) \leq 2\myhs\hat{\mu}_i^{(\ell-1)} = 2 \cdot \mu(\lazy(\hat{U}_i^{(\ell-1)}), \tilde{\gamma})$, as needed.

    Note that $\hat{U}_i^\ellexp$ denotes a more recent $i$-th execution set compared to the $i$-th execution set denoted by~$\hat{U}_i^{(\ell-1)}$, and by~\cref{lem:symmetric_diff} we have $|\hat{U}^\ellexp_i \triangle \lazy(\hat{U}_i^\ellexp)| \leq 2 \myhs \lambda \myhs |\hat{U}_i^\ellexp|$. By applying the same arguments, we obtain $\lvert\ball[\hat{U}_i^\ellexp, X_i, \mu(\lazy(\hat{U}_i^\ellexp), \tilde{\gamma})]\rvert \geq \gamma \myhs |\hat{U}_i^\ellexp|$. Thus, we also conclude that $\mu(\hat{U}^\ellexp_i, \gamma) \leq 2 \cdot \mu(\lazy(\hat{U}_i^\ellexp), \tilde{\gamma})$, as required.
\end{proof}

The next lemma demonstrates that during the $\ell$-th transition phase, for any level $i \in [0, t]$ both $\mu(\lazy(\hat{U}_i^{(\ell-1)}), \tilde{\gamma})$ and $\mu(\lazy(\hat{U}_i^\ellexp), \tilde{\gamma})$ serve as lower bounds on the current optimal $k$-center cost $R^*$, where $\tilde{\gamma}$ is chosen appropriately.

\begin{lemma} \label{lem:rel_OPT_mu}
    For every transition phase index $\ell \geq 0$, consider a fixed level $i \in [0, t]$ during the $\ell$-th transition phase. Then for $\tilde{\gamma} = \frac{\gamma + 3 \myhs \lambda}{1 - 3 \myhs \lambda}$, $0 < \lambda < \frac{1}{6}$, and $0 < \gamma < 1 - 6 \lambda$, it holds that:
    \[
        R^* \;\geq\; \max\big(\mu(\lazy(\hat{U}_i^{(\ell-1)}), \tilde{\gamma}),\; \mu(\lazy(\hat{U}_i^\ellexp), \tilde{\gamma})\big).
    \]
\end{lemma}
\begin{proof}
    Let $C^* \subseteq P$ be a current optimal $k$-center solution. Since $R^* = \max_{p \in P} \dist(p, C^*)$ and $\lazy(\hat{U}_i^\ellexp) \subseteq P$, it follows that:
    \begin{enumerate}
        \item $\lvert\ball[\lazy(\hat{U}_i^\ellexp), C^*, R^*]\rvert \geq \lvert\lazy(\hat{U}_i^\ellexp)\rvert$.
        \item $\lvert\lazy(\hat{U}_i^\ellexp) \setminus \ball(\lazy(\hat{U}_i^\ellexp), C^*, R^*)\rvert \geq 0$.
    \end{enumerate}
    Note that as $0 < \tilde{\gamma} < 1$, we have $\tilde{\gamma} \myhs \lvert\lazy(\hat{U}_i^\ellexp)\rvert < \lvert\lazy(\hat{U}_i^\ellexp)\rvert$ and $(1 - \tilde{\gamma}) \myhs \lvert\lazy(\hat{U}_i^\ellexp)\rvert > 0$. Hence, there is a value $R \leq R^*$ such that: 
    \begin{enumerate}
        \item $\lvert\ball[\lazy(\hat{U}_i^\ellexp), C^*, R]\rvert \,\geq\, \tilde{\gamma} \myhs \lvert\lazy(\hat{U}_i^\ellexp)\rvert$.
        \item $\lvert\lazy(\hat{U}_i^\ellexp) \setminus \ball(\lazy(\hat{U}_i^\ellexp), C^*, R)\rvert \,\geq\, (1-\tilde{\gamma}) \myhs \lvert\lazy(\hat{U}_i^\ellexp)\rvert$.
    \end{enumerate}
    Observe that $\mu(\lazy(\hat{U}_i^\ellexp), \tilde{\gamma})$ is the minimum value that satisfies the requirements in~\cref{def:mu_i} for $\lazy(\hat{U}_i^\ellexp)$ and $\tilde{\gamma}$. Therefore, we conclude that $\mu(\lazy(\hat{U}_i^\ellexp), \tilde{\gamma}) \,\leq\, R \,\leq\, R^*$. By applying the same arguments to $\mu(\lazy(\hat{U}_i^{(\ell-1)}), \tilde{\gamma})$, we can obtain that $\mu(\lazy(\hat{U}_i^{(\ell-1)}), \tilde{\gamma}) \,\leq\, R^*$ as well, and thus the claim follows.
\end{proof}

\subsubsection{Finishing the Proof of~\cref{thm:bicr_alg_worst_case}} \label{sec:finish_proof}

By combining the previous observations and lemmas, and by setting $\tilde{\gamma} = \frac{\gamma + 3 \myhs \lambda}{1 - 3 \myhs \lambda}$, $0 < \lambda < \frac{1}{6}$, and $0 < \gamma < 1 - 6 \lambda$, we establish the approximation ratio of our bicriteria approximate solution $\hat{S}$ for the $k$-center clustering problem. Finally, we finish the proof of~\cref{thm:bicr_alg_worst_case} which we restate for convenience.

\begin{lemma} \label{lem:approx_ratio}
    For every transition phase index $\ell \geq 0$, it holds with high probability that
    $\cost(\hat{S}) \leq 8 \myhs R^*$ during the $\ell$-th transition phase, where $R^*$ is the current optimal $k$-center cost.
\end{lemma}
\begin{proof}
    By~\cref{cor:cost_hat_S_nu}, it holds that $\cost(\hat{S}) \leq 2 \myhs \max(\nu^{(\ell-1)}, \nu^{(\ell)})$. 
    Let $\tau^\ellexp$ be the moment at which the $\ell$-th transition phase begins, and consider the $\tau_i^\ellexp$-th adversarial update at which the value of $\nu_i^\ellexp$ is computed and the set $\hat{U}_i^\ellexp$ is defined. By definition, we have $\nu^\ellexp = \nu^{(\tau^\ellexp)} = \max_{i \in [0, t]} \nu_i^{(\tau_i^\ellexp)}$. 
    By the construction of \texttt{RebuildFromLayer}$(\hspace{0.5pt})$ it holds that $\nu_t^\ellexp = 0$ (see~\linecref{line:last_level} in~\cref{alg:aux_rebuild_proced}), and thus it follows that $\nu^\ellexp = \nu^{(\tau^\ellexp)} = \max_{i \in [0, t)} \nu_i^{(\tau_i^\ellexp)}$.
    Notice that at time $\tau_i^\ellexp$, the fully dynamic algorithm must have invoked~\texttt{RebuildFromLayer}$(\hspace{0.5pt})$ and the $i$-th level must have been scanned. Hence applying~\cref{cor:rel_nu_mu} at time $\tau^\ellexp$, and as $U_i^{(\tau_i^\ellexp)} = \hat{U}_i^\ellexp$ by definition, we obtain that:
    \[
        \nu^\ellexp \,=\,\nu^{(\tau^\ellexp)} \;\leq\; 2 \myhs \max_{i \in [0, t)} \mu(U_i^{(\tau_i^\ellexp)}, \gamma) \;=\; 2 \myhs \max_{i \in [0, t)} \mu(\hat{U}_i^\ellexp, \gamma).\footnote{The value of the last level $t$ corresponds to time $\tau^\ellexp$; since it does not affect the analysis, we leave it as $t$ to simplify the notation.}
    \]
    The same argument can be replicated for $\nu^{(\ell-1)}$ to infer that:
    \[
        \cost(\hat{S}) \leq 4 \cdot \max\left( \max_{i \in [0, t)} \mu(\hat{U}^{(\ell-1)}_i, \gamma),\; \max_{i \in [0, t)} \mu(\hat{U}^{(\ell)}_i, \gamma) \right).
    \]
    Based on~\cref{lem:rel_mu_old_curr}, for every level $i \in [0, t)$ the corresponding terms are upper bounded as follows:
    \[
        \mu(\hat{U}^{(\ell-1)}_i, \gamma) \;\leq\; 2 \cdot \mu\big(\lazy(\hat{U}_i^{(\ell-1)}), \tilde{\gamma}\big) \;\text{ and }\; \mu(\hat{U}^\ellexp_i, \gamma) \;\leq\; 2 \cdot \mu\big(\lazy(\hat{U}_i^\ellexp), \tilde{\gamma}\big).
    \]
    Therefore, the $k$-center cost of the bicriteria approximate solution $\hat{S}$ can be upper bounded by:
    \[
        \cost(\hat{S}) \;\leq\; 8 \cdot \max\left( \max_{i \in [0, t)} \mu\big(\lazy(\hat{U}_i^{(\ell-1)}), \tilde{\gamma}\big),\; \max_{i \in [0, t)} \mu\big(\lazy(\hat{U}_i^\ellexp), \tilde{\gamma}\big) \right).
    \]
    Finally by~\cref{lem:rel_OPT_mu}, it follows that:
    \[
        \cost(\hat{S}) \;\leq\; 8 \cdot \max_{i \in [0, t)} R^* \;=\; 8 \cdot R^*.
    \]
\end{proof}

\bicralgworstcase*
\begin{proof}
    The procedure \texttt{LazySync}$(\hspace{0.5pt})$ requires only a constant number of operations per adversarial point update.
    Specifically,~\texttt{LazySync}$(\hspace{0.5pt})$ runs for $\frac{32}{\lambda^2}$ steps per adversarial update, where $\lambda$ is a small positive constant. Together with the analyses in~\cite{bhattacharya2023fully,bhattacharya2025alm_opt_kcenter}, the claim on the amortized update time follows.
    The remaining guarantees follow from~\cref{lem:size_of_S_hat,lem:approx_ratio}, and the construction of the fully dynamic algorithm.
\end{proof}

\subsection{Smaller Size of Bicriteria Approximate Solution} \label{sec:different_sampling}
\antonis{Add the ``Improved Approximation Ratio and'' to the title when we prove it.}
For each level $i \in [0, t]$, our fully dynamic bicriteria approximation algorithm (in~\cref{thm:bicr_alg_worst_case}) constructs the set $S_i$ by sampling each point from $U_i$ independently with probability~$\min\Big(\frac{\alpha \myhs k \log n}{|U_i|}, 1\Big)$ (see~\linecref{line:sample_Si} in \texttt{RebuildFromLayer}($i$) of~\cref{alg:aux_rebuild_proced}). Based on the idea of Bhattacharya, Costa, Farokhnejad, Lattanzi, and Parotsidis~\cite{bhattacharya2025alm_opt_kcenter}, we can sample in the following way:
\paragraph{Sampling in~\cite{bhattacharya2025alm_opt_kcenter}.}
    For a fixed level $i \in [0, t]$, a temporary set $\tilde{S}_i$ is constructed by sampling uniformly at random a subset of $U_i$ with size $2k$ (i.e., $\tilde{S}_i \subseteq U_i$ and $|\tilde{S}_i| = 2k$). A temporary radius $\tilde{\nu}_i$ and temporary ball $\tilde{B}_i$ are then constructed according to~\linescref{line:min_rad_nu}{line:construct_Bj} in \texttt{RebuildFromLayer}$(\hspace{0.5pt})$  of~\cref{alg:aux_rebuild_proced}.
    
    For the $i$-th level, this sampling process is repeated $\Theta(\log n)$ times. For the smallest temporary radius~$\tilde{\nu}_i$, the $i$-th set $S_i$, the $i$-th ball $B_i$, and the $i$-th radius $\nu_i$ are then assigned the corresponding
    temporary set, ball, and radius associated with $\tilde{\nu}_i$.

\vspace{1em}
Based on the arguments from~\cite{bhattacharya2025alm_opt_kcenter}, this sampling process satisfies the guarantees required for the dynamic MP-bi algorithm (e.g., $\nu_i \leq 2 \myhs \mu(U_i, \gamma)$ with high probability). The benefit of this sampling process is that 
the extra $\log n$ factor is eliminated from the size of the bicriteria approximate solution $S$, and thus we have $|S| = O(k \log \frac{n}{k})$.

\antonis{Explain how we can improve $8$ approx to $4$.}

\section{Fully Dynamic Bicriteria Approximation with Worst-Case Update~Time}\label{sec:bicr_update_time}
Our fully dynamic bicriteria approximation algorithm from~\cref{sec:bicr_recourse} (\cref{thm:bicr_alg_worst_case}) guarantees worst-case recourse but only amortized update time. The reason is that the execution of the procedure \texttt{RebuildFromLayer}($i$) (in~\cref{alg:aux_rebuild_proced}) scans all levels $j \in [i, t]$ and constructs the corresponding balls $B_j$. For fixed levels $i \in [0, t]$ and $j \in [i, t]$, both the computation of the $j$-th radius~$\nu_j$ (in~\linecref{line:min_rad_nu} of~\cref{alg:aux_rebuild_proced}) and the construction of the $j$-th ball $B_j$ within the $j$-th execution set $U_j$ (in~\linecref{line:construct_Bj} of~\cref{alg:aux_rebuild_proced}) require $O(|S_j| \myhs |U_j|) = O(k \myhs |U_j| \log n)$ time.\footnote{The $j$-th radius $\nu_j$ can be computed using a linear-time selection algorithm.} Since $j \geq i$ it holds that $|U_j| \leq |U_i|$, and thus $O(k \myhs |U_j| \log n)$ is upper bounded by $O(k \myhs |U_i| \log n)$. As there are at most $t = O(\log \frac{n}{k})$ levels, the procedure \texttt{RebuildFromLayer}($i$) requires $O(k \myhs |U_i| \log n \log \frac{n}{k})$ time.

Hence, whenever an adversarial point update causes a counter $\textit{cnt}_i$ to exceed its threshold $\lambda \myhs n_i$, the procedure \texttt{RebuildFromLayer}($i$) (in~\cref{alg:aux_rebuild_proced}) is executed to produce a new set $S$, incurring a worst-case update time of $\Theta(k \myhs |U_i| \log n \log \frac{n}{k})$ (see \texttt{CheckThresholdRebuild}$(\hspace{0.5pt})$ in~\cref{alg:aux_rebuild_proced}). Recall also that the value of the last level $t$ may change after an invocation of~\texttt{RebuildFromLayer}$(\hspace{0.5pt})$, but it is always upper bounded by $O(\log \frac{n}{k})$.

In this section, we develop a fully dynamic $(8, O(\log n \myhs \log \frac{n}{k}))$-bicriteria approximation algorithm for the $k$-center clustering problem with \emph{near-optimal worst-case update time}, as demonstrated in~\cref{thm:bicr_alg_worst_case_update_time}.

\begin{restatable}{theorem}{bicralgworstcaseupdatetime} \label{thm:bicr_alg_worst_case_update_time}
    There is a randomized fully dynamic algorithm against an adaptive adversary that, given a point set $P$ in a metric space subject to point updates and an integer $k \geq 1$, maintains a subset of points $S \subseteq P$ such that:
    \begin{itemize}
        \item The set $S$ is with high probability an $(8, O(\log n \myhs \log \frac{n}{k}))$-bicriteria approximate solution to the $k$-center clustering problem.

        \item The worst-case update time is $O\big(k \log n \myhs (\log \frac{n}{k})^2\big)$.
    \end{itemize}
\end{restatable}
\noindent In order to achieve near-optimal worst-case update time, our fully dynamic bicriteria approximation algorithm rebuilds its data structures lazily, by spending $\Theta\big(\frac{k}{\epsilon} \log n \myhs (\log \frac{n}{k})^2\big)$ time per adversarial point update. The parameter $\epsilon$ is a small positive constant less than $1$ (i.e., $\epsilon \in (0, 1)$). In particular, the execution sets are maintained in an appropriately lazy manner (without rebuilding them) such that:
\begin{enumerate}
    \item The approximation ratio is affected only up to a constant factor.
      
    \item The worst-case update time of $\Theta\big(\frac{k}{\epsilon} \log n \myhs (\log \frac{n}{k})^2\big)$ is sufficient to build all the data structures until the execution sets can no longer be maintained lazily.
\end{enumerate}
For this reason, we extend the definition of $\lazy(\cdot)$ from~\cref{sec:lower_bound} to arbitrary sets of points. 
\begin{definition}[lazy set]\label{def:lazy_set}
   For any subset of points $U \subseteq P$, the set $\lazy(U)$ is defined as follows:
\begin{itemize}
    \item $\lazy(U)$ is initialized to $U$.
    \item $\lazy(U)$ is subject to changes due to~\linecref{algline:add_to_Ui} in~\cref{alg:point_ins} and~\linecref{algline:remove_from_Uj} in~\cref{alg:point_del} (i.e., subject to adversarial point updates). In other words, any newly inserted point is added to $\lazy(U)$, and any deleted point that belongs to $\lazy(U)$ is removed from it.
\end{itemize} 
\end{definition}
\noindent
Essentially, the set $\lazy(U)$ is a  non-rebuilt set (see~\cref{def:non-rebuilt})  that can be thought of as the current $U$ undergoing adversarial point updates, provided that $U$ is not modified by \texttt{RebuildFromLayer}$(\hspace{0.5pt})$. 

For every level $i \in [0, t]$, our fully dynamic bicriteria approximation algorithm maintains a dynamic algorithm $\mathcal{BD}_i$. In~\cref{subsec:ith_alg} we describe the $i$-th algorithms $\mathcal{BD}_i$ for $i \in [0, t]$, and in~\cref{subsec:lazy_constr_fully_dynam_alg} we explain how the fully dynamic bicriteria approximation algorithm utilizes them. 

\subsection{The $i$-th Algorithm $\mathcal{BD}_i$} \label{subsec:ith_alg}
The $i$-th algorithm $\mathcal{BD}_i$ simulates the procedure~\texttt{RebuildFromLayer}($i$) (see~\cref{alg:aux_rebuild_proced}) in a lazy manner, starting from an execution set of the $i$-th level. A pseudocode of the $i$-th lazy algorithm~$\mathcal{BD}_i$ is provided in~\texttt{BuildLazily}$(\cdot)$ of~\cref{alg:build_lazily}. 
For some $\tau \in \mathbb{N}$, let $U^\tauexp_i$ be the $i$-th execution set at the $\tau$-th adversarial update. The lazy rebuild of the $i$-th level starts from a fixed execution set $U^\tauexp_i$, which is denoted by $U_i^i$. For every level $j \in [i, t^i]$ (where $t^i$ is the last level of $\mathcal{BD}_i$), each $j$-th ball~$B_j^i$ is constructed within the \emph{fixed} execution set $U_j^i \coloneqq \lazy(U^i_i) \setminus \bigcup_{\xi=i}^{j-1} B_\xi^i$. Namely, the algorithm starts constructing \emph{lazily} the $j$-th ball~$B^i_j$ within the fixed subset of points $U_j^i$, which can remain lazy for $\epsilon \myhs |U_j^i|$ adversarial updates.\footnote{Intuitively, the notion of a set remaining \emph{lazy} captures the idea that the corresponding values  $\mu(\cdot, \cdot)$ (of~\cref{def:mu_i}) change only slightly in a desired way.} 

In particular, the $i$-th algorithm $\mathcal{BD}_i$ spends $\Theta(\frac{k}{\epsilon} \log n \myhs \log \frac{n}{k})$ time per adversarial update towards the construction of each $B^i_j$ within $U_j^i$, and then waits until the next adversarial update. This implies that the worst-case update time of $\mathcal{BD}_i$ is $O(\frac{k}{\epsilon} \log n \myhs \log \frac{n}{k})$. When the $i$-th algorithm $\mathcal{BD}_i$ has completed all of its computations, it reports completion.

\paragraph{Notion of \emph{fixed} execution set.}
In this paragraph, we clarify the term \emph{fixed} for the execution sets~$U^i_j \coloneqq \lazy(U^i_i) \setminus \bigcup_{\xi=i}^{j-1} B_\xi^i$, where $0 \leq i \leq t$ and $i \leq j \leq t^i$ (with $t^i$ denoting the last level of $\mathcal{BD}_i$).
We remark that the execution set $U_j^i$ is considered fixed once the $i$-th algorithm $\mathcal{BD}_i$ begins constructing its data structures at the $j$-th level (i.e., the fixed $U_j^i$ is defined in~\linecref{line:new_exec_set_lazy}, and the construction of the $j$-th level begins in~\linecref{algline:while_rebuild_lazy} of~\cref{alg:build_lazily}). Prior to this, $U_j^i$ has not yet been defined, and the corresponding region of points is modified by adversarial updates (this explains the presence of $\lazy(U^i_i)$ in the definition of $U^i_j$). 

Hence when an adversarial point update is passed to $\mathcal{BD}_i$, the execution sets $U_j^i$ that have already been defined (i.e., for which the construction of the $j$-th level has begun) are not modified. Instead, only the lazy part $\lazy(U_i^i)$ is updated according to the definition of $\lazy(\cdot)$ (see~\cref{def:lazy_set}), in order to produce the subsequent fixed execution sets.

\paragraph{Equivalent definition of execution sets.}
Notice that it is equivalent to define the fixed execution set~$U^i_j$ based on the previous lazy execution set of the $i$-th algorithm~$\mathcal{BD}_i$, that is, $U^i_j \coloneqq \lazy(U^i_{j-1}) \setminus B^i_{j-1}$. Thus, assuming that  $\mathcal{BD}_i$ is currently constructing~$B^i_j$, an equivalent perspective is that only the lazy set~$\lazy(U_j^i)$ is updated upon an adversarial point update to~$\mathcal{BD}_i$. In~\cref{subsec:adv_upd_BDi}, we describe in more detail how $\mathcal{BD}_i$ handles the adversarial point updates.

\subsubsection{Adversarial Point Updates to $\mathcal{BD}_i$} \label{subsec:adv_upd_BDi}
The dynamic lazy algorithms $\mathcal{BD}_i$ affect each other; this interaction is discussed later in~\cref{subsec:adv_point_upd}, which addresses adversarial point updates to the fully dynamic algorithm. For this reason, the $i$-th algorithm~$\mathcal{BD}_i$ actually maintains all lazy sets $\lazy(U^i_j)$ for all $j \in [i, t^i]$. The value of $t^i$ is the current last level of $\mathcal{BD}_i$ (see~\linescref{algline:curr_level_t_lazy_init}{algline:curr_level_t_lazy_new} in~\cref{alg:build_lazily}). 

Hence when the adaptive adversary sends a point update, the $i$-th algorithm $\mathcal{BD}_i$ updates each lazy execution set $\lazy(U^i_j)$ for every $j \in [i, t^i]$ according to the definition of $\lazy(\cdot)$ (in~\cref{def:lazy_set}), as follows:
\begin{itemize}
    \item Any newly inserted point is added to $\lazy(U^i_j)$.
    
    \item Any deleted point that belongs to $\lazy(U^i_j)$ is removed from it. 
\end{itemize}
\noindent 
Regarding the adversarial point deletions, if the deleted point $p$ belongs to a set $S_j^i$ (where $j \in [i, t^i]$) then the $i$-th algorithm $\mathcal{BD}_i$ replaces it with another point. Specifically, if $p \in S_j^i$ then the adversarially deleted point~$p$ is replaced by another point from the cluster induced by the respective assignment $\sigma^i$. This is a process similar to~\texttt{ApplyPointReplacement}$(\hspace{0.5pt})$ in~\cref{alg:apply_point_repl}, using
~$S_j^i$ and $\sigma^i$ in place of~$\tilde{S}$ and $\tilde{\sigma}$ respectively (e.g., see~\linecref{algline:rem_for_replace} in~\cref{alg:apply_point_repl} for the candidate replacements of the deleted point $p$).

At the end of each adversarial point update, the $i$-th algorithm $\mathcal{BD}_i$ resumes the procedure~\texttt{BuildLazily}$(\cdot)$ of~\cref{alg:build_lazily} from its previous state, spending $\Theta(\frac{\alpha}{\epsilon} \myhs k \log n \myhs \log \frac{n}{k})$ time per adversarial update.\footnote{The $i$-th algorithm $\mathcal{BD}_i$ begins from $U_i^i$ and is currently lazily constructing its data structures at its internal $j$-th level (i.e., $j = t^i$).}

\begin{algorithm}[H]\footnotesize
\algnewcommand{\LineComment}[1]{\State \(\triangleright\) #1}
\algrenewcommand\algorithmiccomment[1]{\hspace{1em} \(\triangleright\) #1}
\caption{\textsc{Build Lazily}{}}\label{alg:build_lazily}

\begin{algorithmic}[1]

\LineComment{The algorithm has global access to all its data structures}
\LineComment{$\beta \in (0, 1)$ is a small positive constant and $\alpha \geq 1$ is a sufficiently large constant}
\vspace{0.6em}

\Function{BuildLazily}{$U^i_i$} 
    \State $j \gets i$
    \State $t^i \gets i$ \label{algline:curr_level_t_lazy_init}
    \vspace{0.3em}
    
    \While{$|U^i_j| > \alpha \myhs k \log n \myhs \log \frac{n}{k}$} \label{algline:while_rebuild_lazy}
        \State For each adversarial update, spend $\Theta(\frac{\alpha \myhs k}{\epsilon} \log n \myhs \log\frac{n}{k})$ time towards:
        
        \State \hspace*{2em} Construct $S^i_j$ by sampling each $p \in U^i_j$ independently with probability $\min\Big(\frac{\alpha \myhs k \log n}{|U^i_j|}, 1\Big)$ \label{line:sample_Si_lazy} 
        
        \State \hspace*{2em} $\nu^i_j \gets \min \{r \geq 0 \mid \lvert\ball[U^i_j, S^i_j, r]\rvert \geq \beta \myhs |U^i_j|\}$ \label{line:min_rad_nu_lazy}
         
        \State \hspace*{2em} $B^i_j \gets \ball[U^i_j, S^i_j, \nu^i_j]$ \label{line:construct_Bj_lazy}
        \State \hspace*{2em} For every $p \in B^i_j: \sigma^i(p) \gets \argmin_{q \in S^i_j} \dist(p, q)$ \label{algline:sigma_in_rebuild_lazy}
        
        \vspace{0.4em}
        \State When the previous steps have been completed: 
        \State \hspace*{2em} $U^i_{j+1} \gets \lazy(U^i_j) \setminus B^i_j$ \label{line:new_exec_set_lazy}
        \State \hspace*{2em} $j \gets j + 1$
        \State \hspace*{2em} $t^i \gets j$ \label{algline:curr_level_t_lazy_new}
    \EndWhile
    
    \vspace{0.3em}
    \State $S^i_t \gets U^i_t$, $B^i_t \gets U^i_t$,
    $\nu^i_t \gets 0$, $\forall p \in U^i_t: \sigma^i(p) \gets p$ \label{line:last_level_lazy}
    \vspace{0.2em}
    \State \textbf{report} completion
    \vspace{0.1em}
\EndFunction

\end{algorithmic}
\end{algorithm}

\subsection{Lazy Construction} \label{subsec:lazy_constr_fully_dynam_alg}
The last level $t$ depends on  $t^i$ for the $i$-th algorithm $\mathcal{BD}_i$ that most recently reported completion. For every level $i \in [0, t]$, the fully dynamic bicriteria approximation algorithm maintains the $i$-th algorithm~$\mathcal{BD}_i$ (see~\texttt{BuildLazily}$(\cdot)$ in~\cref{alg:build_lazily}). The bicriteria approximate solution $S$ is the union of the sets $S_j$, namely $S \coloneqq \bigcup_{j=0}^t S_j$. For each level $j \in [0, t]$, the $j$-th set $S_j$ is equal to some $S^i_j$ which is the $j$-th set of the $i$-th algorithm $\mathcal{BD}_i$, where $i \in [0, j]$. 

\paragraph{Additional components of the fully dynamic algorithm.}
For each level $j \in [0, t]$, the fully dynamic algorithm maintains \emph{implicitly} the $j$-th ball $B_j$, simulating the internal ball $B_j^i$ of $\mathcal{BD}_i$ for the level~$i$ such that $S_j = S_j^i$. Furthermore, the fully dynamic algorithm maintains \emph{explicitly} the $j$-th assignment $\sigma_j$, simulating the internal assignment $\sigma^i$ of $\mathcal{BD}_i$ for the level $i$ such that $S_j = S_j^i$. In other words, for each global level~$j \in [0, t]$, the fully dynamic algorithm uses the data structures of an $i$-th algorithm $\mathcal{BD}_i$, where~$i \in [0, j]$.

The $j$-th balls $B_j$ are maintained implicitly because, once the $i$-th algorithm $\mathcal{BD}_i$ reports completion, we want to transfer its data structures to the fully dynamic algorithm in a single adversarial update. Hence, the desired worst-case update time of $\Theta\big(k \log n \myhs (\log \frac{n}{k})^2\big)$ is inadequate, as there may be $\Theta(n)$ points within the corresponding internal balls. Nevertheless, maintaining the $j$-th assignments $\sigma_j$ explicitly is efficient and sufficient for the fully dynamic algorithm.

\subsubsection{Adversarial Point Updates} \label{subsec:adv_point_upd}
When the adaptive adversary inserts a new point into $P$ or deletes an existing point from $P$, the fully dynamic bicriteria approximation algorithm proceeds as follows. Initially, this adversarial point update is forwarded to all the algorithms $\mathcal{BD}_i$, where $0 \leq i \leq t$. The fully dynamic algorithm then finds the smallest level $i \in [0, t]$ for which the $i$-th algorithm $\mathcal{BD}_i$ reports completion,\footnote{There is always such a level; the reason is explained in a later paragraph.} and updates the last level $t$ to $t^i$.

Next, the fully dynamic bicriteria approximation algorithm  replaces all sets $S_j$ where $i \leq j \leq t$, with the corresponding sets $S_j^i$ of the $i$-th algorithm $\mathcal{BD}_i$. Also for every (global) level $j \in [i, t]$, the $j$-th assignment $\sigma_j$ is updated to~$\sigma^i$; hence the corresponding $j$-th ball $B_j$ is \emph{implicitly} updated to $B_j^i$. 
Moreover for every level $j \in [i, t]$,
each $j$-th algorithm~$\mathcal{BD}_j$ is restarted with input the execution set $\lazy(U_j^i)$, which is the lazy execution set of the internal $j$-th level of the $i$-th algorithm $\mathcal{BD}_i$. To restart each $\mathcal{BD}_j$, the fully dynamic algorithm sets the fixed execution set~$U_j^j$ to $\lazy(U_j^i)$, and then invokes~\texttt{BuildLazily}$(U_j^j)$ (in~\cref{alg:build_lazily}).\footnote{Note that the sets $U^i_j$ can be redefined. The concept of a \emph{fixed} execution set is clarified in a later paragraph.}

In turn, the last $t$-th algorithm $\mathcal{BD}_t$ is restarted with $\lazy(U_t^i)$ as input. Subsequently, the fully dynamic algorithm continues based on the size of the updated execution set~$U_t^t \coloneqq \lazy(U_t^i)$, as follows:
\begin{itemize}
    \item If $|U^t_t| \leq \alpha \myhs k \log n \log \frac{n}{k}$,
    then the last $t$-th algorithm $\mathcal{BD}_t$ immediately reports completion.

    \item Otherwise if $|U^t_t| > \alpha \myhs k \log n \log \frac{n}{k}$, then the value of the last level $t$ is incremented by one (i.e., a new level is created). For the updated value of $t$, a new $t$-th algorithm $\mathcal{BD}_t$ is then invoked with $\{p\}$ as input. As a result, this new $t$-th algorithm $\mathcal{BD}_t$ immediately reports completion.
\end{itemize}
Once the (new) last $t$-th algorithm $\mathcal{BD}_t$ reports completion, the set $S_t$ is replaced by the corresponding set~$S_t^t$, where $t$ is the (updated) last level. Furthermore, the last $t$-th assignment $\sigma_t$ is also updated to $\sigma^t$, which means that the last $t$-th ball $B_t$ is implicitly updated to $B_t^t$. The (new) last $t$-th algorithm $\mathcal{BD}_t$ is then ready to begin at the next adversarial point update.

\paragraph{The smallest level $i$ that reports completion.}
In this paragraph, we explain why there is always a level $i$ for which the $i$-th algorithm $\mathcal{BD}_i$ reports completion. Consider the last level $t$ and the size of~$\lazy(U^t_t)$.
Under adversarial point deletions the size of $\lazy(U^t_t)$ is reduced, and thus it still holds that $\lvert\lazy(U^t_t)\rvert = |U^t_t| \leq \alpha \myhs k \log n \log \frac{n}{k}$.
This implies that the last $t$-th algorithm $\mathcal{BD}_t$ immediately reports completion. For adversarial point insertions, there are two possible scenarios depending on the size of~$\lazy(U^t_t)$:
\begin{itemize}
    \item Assume that after the adversarial point insertion, it still holds that $\lvert\lazy(U^t_t)\rvert = |U^t_t| \leq \alpha \myhs k \log n \log \frac{n}{k}$.
    In this scenario, the last $t$-th algorithm $\mathcal{BD}_t$ immediately reports completion.

    \item Assume that after the adversarial insertion of a new point $p$, it holds that $\lvert\lazy(U^t_t)\rvert = |U^t_t| > \alpha \myhs k \log n \log \frac{n}{k}$. In this scenario, the value of the last level $t$ is incremented by one (i.e., a new level is created). For the updated value of $t$, a new $t$-th algorithm $\mathcal{BD}_t$ is then invoked with $\{p\}$ as input. As a result, this new $t$-th algorithm $\mathcal{BD}_t$ immediately reports completion.
\end{itemize}
Thus, the (new) last $t$-th algorithm $\mathcal{BD}_t$ always reports completion for the (updated) last level~$t$, and it is then ready to begin at the next adversarial update with the same updated input $U^t_t$. The conclusion is that since $\mathcal{BD}_t$ reports completion immediately upon an adversarial point update, and is then ready to begin with the same updated input, it will report completion in the next adversarial update as well.
Therefore for every adversarial point update, there is always a level $i \in [0, t]$ for which the $i$-th algorithm~$\mathcal{BD}_i$ reports completion.

\paragraph{Modification of the bicriteria approximate solution.}
In this paragraph, we elaborate more on how the bicriteria approximate solution $S$ is affected
by the adversarial point updates. Consider the two types of adversarial point updates:
\begin{itemize}
    \item If the adversary deletes an existing point $p \in S^i_j$ for some levels $0 \leq i \leq t$ and $i \leq j \leq t^i$ (with $t^i$ denoting the last level of $\mathcal{BD}_i$), then the $i$-th algorithm~$\mathcal{BD}_i$ handles this adversarial point deletion internally. In turn, the corresponding $j$-th set~$S_j$ of the solution~$S$ is updated appropriately (if needed).

    \item If the adversary inserts a new point $p$, then let $i \in [0, t]$ be the smallest level such that the $i$-th algorithm $\mathcal{BD}_i$ reports completion. The newly inserted point $p$ has been added to $\lazy(U_t^i)$ via~$\mathcal{BD}_i$. By the construction of the fully dynamic algorithm, the (new) last $t$-th algorithm $\mathcal{BD}_t$ is restarted either with $\lazy(U_t^i)$ or with $\{p\}$ as input. Thus, this (new) last $t$-th algorithm $\mathcal{BD}_t$ contains $p$ in its input and
    always reports completion for the (updated) last level~$t$. 
    
    This implies that $p \in S_t^t$ and $S_t \coloneqq S_t^t$, where $t$ is the (updated) last level. Therefore, the fully dynamic algorithm adds the newly inserted point $p$ to the maintained solution~$S$ by possibly creating a new level.
\end{itemize}

\paragraph{Notion of \emph{fixed} execution set.}
As mentioned in~\cref{subsec:ith_alg}, the execution set $U_j^i \coloneqq \lazy(U^i_{j-1}) \setminus B^i_{j-1}$ (where $0 \leq i  \leq t$ and $i \leq j \leq t^i$) is considered fixed once the $i$-th algorithm $\mathcal{BD}_i$ begins constructing its data structures at the $j$-th level; prior to this, $U_j^i$ has not yet been defined. However, the execution set $U^i_j$ can be redefined by a new execution of the $i$-th algorithm $\mathcal{BD}_i$, which lazily constructs its internal $j$-th level (e.g., the ball $B_j^i$ within~$U^i_j$).
Once the execution set $U_j^i$ is redefined, it remains fixed until its next redefinition.
In addition, whenever the execution set $U_j^i$ is redefined, the corresponding lazy execution set $\lazy(U_j^i)$ is also updated according to the definition of $\lazy(\cdot)$ (in~\cref{def:lazy_set}), as follows:
\begin{itemize}
    \item $\lazy(U_j^i)$ is initialized to the most recent $U_j^i$.
    \item For every adversarial point update after its most recent initialization, the lazy execution set $\lazy(U_j^i)$ is updated in the following way: 
    \begin{itemize}
        \item Any newly inserted point is added to $\lazy(U_j^i)$.

        \item Any deleted point that belongs to $\lazy(U_j^i)$ is removed from it.
    \end{itemize}
\end{itemize}

\subsubsection{Preprocessing Phase}
In the preprocessing phase, the $0$-th execution set $U_0$ is initialized to $P$, and the procedure \texttt{RebuildFromLayer}($0$) of~\cref{alg:aux_rebuild_proced} is executed producing our initial solution~$S$. All assignments $\sigma_i$ (where $i \in [0, t]$) are initialized to the assignment $\sigma$ produced by~\texttt{RebuildFromLayer}($0$). Afterwards, for every level~$i \in [0, t]$ the $i$-th algorithm $\mathcal{BD}_i$ begins with the $i$-th execution set $U_i$ as input, where $U_i$ was produced by~\texttt{RebuildFromLayer}($0$).

\subsection{Analysis of the Fully Dynamic Algorithm with Worst-Case Update Time}
In this section, we prove~\cref{thm:bicr_alg_worst_case_update_time} by analyzing our fully dynamic bicriteria approximation algorithm, described in~\cref{subsec:lazy_constr_fully_dynam_alg}.
The analysis is divided into five subsections, as outlined below.
First, in~\cref{subsec:aux_cont_upd} we provide some auxiliary context and definitions. Next, in~\cref{subsec:upper_bound_upd} we derive an upper bound on the cost of the solution~$S$. 
Then, in~\cref{sec:size_bicr_worst_upd} we analyze the size of our bicriteria approximate solution $S$.
In~\cref{subsec:lower_bound_worst_case_upd}, we associate certain values $\mu(\cdot, \cdot)$ with a lower bound on the optimal $k$-center cost. Finally, in~\cref{sec:finish_proof_update_time} we conclude the proof of Theorem~\ref{thm:bicr_alg_worst_case_update_time}.

\subsubsection{Auxiliary Context and Definitions} \label{subsec:aux_cont_upd}
The fully dynamic bicriteria approximation algorithm maintains
the bicriteria approximate solution~$S$, which is the union of the $j$-th sets $S_j$. For a fixed level $j \in [0, t]$ we refer to each such $j$-th set $S_j$ as \emph{global}, while for a fixed $i$-th algorithm $\mathcal{BD}_i$ (where $i \in [0, j]$) we refer to its corresponding $j$-th set $S_j^i$ as \emph{internal}. For the global $j$-th set, let $i_j \in [0, j]$ denote the level such that $S_j = S_j^{i_j}$. 
Namely, the global $j$-th set $S_j$ is equal to the internal $j$-th set of the $i_j$-th algorithm $\mathcal{BD}_{i_j}$. In turn, we have $S \coloneqq \bigcup_{j=0}^t S_j =  \bigcup_{j=0}^t S_j^{i_j}$. By the construction in~\cref{subsec:adv_point_upd}, the value of $i_j$ is the level $i \in [0, j]$ for which the $i$-th algorithm $\mathcal{BD}_i$ has reported completion most recently. 

Let $\tau_i \in \mathbb{N}$ be the most recent moment at which the $i$-th algorithm $\mathcal{BD}_i$ reports completion, and let $\tau_j^i \leq \tau_i$ be the most recent moment at which the fixed execution set $U_j^i$ is defined prior to the completion of~$\mathcal{BD}_i$. At time $\tau_j^i$, the $i$-th algorithm $\mathcal{BD}_i$ (i.e., the procedure~\texttt{BuildLazily}$(U_i^i)$ in~\cref{alg:build_lazily}) begins the lazy construction of its data structures at the internal $j$-th level. 
Since the construction is lazy, the value of the radius $\nu_j^i$ (in~\linecref{line:min_rad_nu_lazy} of~\cref{alg:build_lazily}) may be computed after the $\tau_j^i$-th adversarial update. However, an important detail is that the radius $\nu_j^i$ is computed with respect to the execution set~$U_j^i$ defined prior to the completion of $\mathcal{BD}_i$, which implies that $\tau_j^i$ is the relevant moment for~$\nu_j^i$. 

Notice that the $i$-th algorithm $\mathcal{BD}_i$ is restarted once it reports completion. Consequently, a new fixed execution set $U_j^i$ could be defined after the most recent completion of $\mathcal{BD}_i$. Since the fully dynamic algorithm can utilize a set $S_j^i$ (with the associated radius $\nu_j^i$) only after the completion of $\mathcal{BD}_i$, the moment when the new $U_j^i$ is defined is not yet relevant for $\nu_j^i$. This justifies why time $\tau_j^i$ is at most $\tau_i$. To distinguish between the current $U_j^i$ and the relevant $U_j^i$, we denote by $\hat{U}_j^i$ the relevant fixed execution set defined at time $\tau_j^i \leq \tau_i$. Similarly, we denote by $\hat{\nu}_j^i$, $\hat{B}^i_j$, and $\hat{S}_j^i$ the radius, the ball, and the internal $j$-th set respectively, computed with respect to~$\hat{U}_j^i$.

\subsubsection{Upper Bound on the Cost of the Bicriteria Approximate Solution} \label{subsec:upper_bound_upd}
The following two statements, \cref{obs:U_subset_within_BD,lem:within_BD_ball}, concern the behavior of the $i$-th algorithm~$\mathcal{BD}_i$ for some level $i \in [0, t]$. Afterwards, we analyze how these $\mathcal{BD}_i$ algorithms interact.

\begin{observation} \label{obs:U_subset_within_BD}
    After an adversarial point update, consider a fixed level $i \in [0, t]$ and the corresponding $i$-th algorithm $\mathcal{BD}_i$. Assume that $\mathcal{BD}_i$ is currently lazily constructing its data structures at the $j$-th level (i.e.,~$j$ is the current last internal level within $\mathcal{BD}_i$, and so $j = t^i$). Then it holds that $\lazy(U_\xi^i) \subseteq \lazy(U_{\xi-1}^i)$ for every internal level $\xi \in [i+1, j]$. 
\end{observation}
\begin{proof}
    By the construction of~$\mathcal{BD}_i$ in \linecref{line:new_exec_set_lazy} of~\cref{alg:build_lazily} (where $U^i_{\xi} \subseteq \lazy(U^i_{\xi-1})$), and by the way each lazy execution set is updated (see also~\cref{def:lazy_set}), the claim follows.
\end{proof}

\begin{lemma} \label{lem:within_BD_ball}
    After an adversarial point update, consider a fixed level $i \in [0, t]$ and the corresponding $i$-th algorithm $\mathcal{BD}_i$. Assume that $\mathcal{BD}_i$ is currently lazily constructing its data structures at the $j$-th level (i.e.,~$j$ is the current last internal level within $\mathcal{BD}_i$, and so $j = t^i$). Then every point $p \in \lazy(U_i^i)$ either belongs to exactly one ball~$B_\xi^i$ with $ \xi \in [i, j-1]$, or lies in $\lazy(U_j^i)$. 
\end{lemma}
\begin{proof}
    Consider an arbitrary point $p \in \lazy(U_i^i)$ and suppose to the contrary that $p \not \in  \bigcup_{\xi = i}^{j-1}B^i_{\xi} \cup \lazy(U_j^i)$. Then either $p \in U^i_j$ or $p$ has been inserted by the adversary during the construction of the internal $j$-th level. Both cases lead to a contradiction: in the former because $\lazy(U_i^i) \cap U^i_j \subseteq \lazy(U^i_j)$, and in the latter because~$p$ is directly added to $\lazy(U^i_j)$.
\end{proof}

The following statements concern the interaction between the dynamic lazy algorithms $\mathcal{BD}_i$ within the fully dynamic bicriteria approximation algorithm.

\begin{lemma} \label{lem:U_subset_between_alg}
    After an adversarial point update, it holds that $\lazy(U_i^i) \subseteq \lazy(U_{i-1}^{i-1})$ for every level $i \in [1, t]$.
\end{lemma}
\begin{proof}
    The proof is by induction on the adversarial point updates. The base case is immediately after the preprocessing phase, where the statement holds. Regarding the induction step, consider a fixed level~$i \in [0, t]$ during the current adversarial point update. We analyze the two possible cases based on whether there is a level~$j \leq i$ such that the $j$-th algorithm $\mathcal{BD}_j$ reports completion:
    \begin{itemize}
        \item Assume that there is no such level $j$, because the smallest level $j$ for which $\mathcal{BD}_j$ reports completion is greater than $i$ (i.e., $j > i$). Since the fully dynamic algorithm forwards every adversarial point update to all algorithms $\mathcal{BD}_\xi$
        and each $\mathcal{BD}_\xi$ (where $\xi \in [0, t]$) updates all its lazy execution sets, it follows from the induction hypothesis that $\lazy(U_i^i) \subseteq \lazy(U_{i-1}^{i-1})$.

        \item Assume that there is such a level $j \leq i$. The fully dynamic algorithm then sets all $U_\xi^\xi$ to $\lazy(U_\xi^j)$ for all $\xi \in [j, t]$.
        If $j = i$ then the claim follows by the induction hypothesis and the way the lazy execution sets are updated. Otherwise if $j < i$, based on~\cref{obs:U_subset_within_BD} it holds that $\lazy(U_i^j) \subseteq \lazy(U_{i-1}^j)$.
        Since the set $U_{i-1}^{i-1}$ is updated to $\lazy(U_{i-1}^j)$ and the set $U_i^i$ to $\lazy(U_i^j)$, according to~\cref{def:lazy_set} the claim follows. 
    \end{itemize}
\end{proof}

After an adversarial point update, let $i \in [0, t]$ be the smallest level such that $\mathcal{BD}_i$ reports completion. Then by the construction of the fully dynamic bicriteria approximation algorithm, each ball $B_j^i$ (where~$j \in [i, t]$) is used as the $j$-th ball $B_j$, namely $B_j \coloneqq B_j^i$ for all levels $j \in [i, t]$. Roughly speaking, the next lemma (\cref{lem:between_alg_ball}) says the following. For the $i$-th algorithm $\mathcal{BD}_i$, which replaces all global $j$-th balls $B_j$ with its internal $j$-th balls $B^i_j$, every point $p$ appearing in some old $j$-th ball $B^\old_j$ must also appear in some internal $\xi$-th ball $B_\xi^i$. This is important because then every point in some affected global $j$-th ball still belongs to some updated global $\xi$-th ball.

Recall that the notation $\hat{B}_\xi^i$ denotes the \emph{relevant} balls, as defined in~\cref{subsec:aux_cont_upd}. Throughout the analysis, we use $\old$ as a superscript for the states before the adversarial update, and omit the $\old$ superscript for the states after the fully dynamic algorithm processes the adversarial update.

\begin{lemma} \label{lem:between_alg_ball}
    After an adversarial point update, let $i \in [0, t]$ be the smallest level such that $\mathcal{BD}_i$ reports completion. Then for each level $j \in [i, t]$ and for every point $p \in B^\old_j \cap P$ it holds that:
    \begin{itemize}
        \item either $p$ belongs to exactly one ball $\hat{B}_\xi^i$ for some level $\xi \in [i, t - 1]$, 
        \item or $p$ belongs to $S_t$. 
    \end{itemize}
\end{lemma}
\begin{proof}
Consider a fixed level $j \in [i, t]$. To simplify the notation, let~$\zeta \coloneqq i_j^\old$ denote the level such that~$B^\old_j = \bigl(\hat{B}_j^\zeta\bigr)^\old$ before the adversarial update; similarly let $\tau^\old \coloneqq \tau_\zeta^\old$ and $U^\old_\xi \coloneqq \bigl(\hat{U}_\xi^\zeta\bigr)^\old$ for all levels~$\xi \in [\zeta, (t^\zeta)^\old]$ before the adversarial update.
Consider a point $p \in P$ that belonged to the $j$-th ball~$B^\old_j$, which is replaced upon the adversarial update. 
Since $p \in B^\old_j$ and $B^\old_j \subseteq U^\old_j$ by construction (in~\linecref{line:construct_Bj_lazy} of~\cref{alg:build_lazily}), we have~$p \in U^\old_j$. In turn as $p \in P$, it holds that $p \in \lazy\big(U^\old_j)$ at time~$\tau^\old$. 
We consider two cases depending on whether the level $\zeta$ is greater than $i$:
\begin{itemize}
    \item If $\zeta \leq i$, then at time $\tau^\old$ the $i$-th algorithm $\mathcal{BD}_i$ is restarted with $\lazy(U^\old_i)$ as input. 
    Based on~\cref{obs:U_subset_within_BD} we have $\lazy\big(U^\old_j\big) \subseteq \lazy\big(U^\old_i\big)$, and so $p \in \lazy\big(U^\old_i\big)$. Hence, the point~$p$ belongs to the execution set $\lazy(\hat{U}_i^i) = \lazy(U^\old_i)$ (see also the left Figure~\ref{fig:BD_algs_left}). 

    \item Otherwise if $i < \zeta$, then by~\cref{obs:U_subset_within_BD} we have $\lazy\big(U^\old_j\big) \subseteq \lazy\big(U^\old_\zeta\big)$, and so $p \in \lazy\big(U^\old_\zeta\big)$. Since by~\cref{lem:U_subset_between_alg} we have $\lazy\big(U^\old_\zeta\big) \subseteq \lazy\big(\hat{U}_i^i\big)$, it holds that $p \in \lazy\big(\hat{U}_i^i\big)$ (see also the right Figure~\ref{fig:BD_algs_right}). 
\end{itemize}
In both cases, the $i$-th algorithm $\mathcal{BD}_i$ had been invoked with $\hat{U}_i^i$ as input, and $p \in \lazy(\hat{U}_i^i)$. Based on~\cref{lem:within_BD_ball} and the fact that $t = t^i$, it follows that $p$ either belongs to exactly one ball $\hat{B}_\xi^i$ with $\xi \in [i, t - 1]$, or lies in~$\lazy(\hat{U}_t^i)$.\footnote{In the lemmas we invoke, we can use the $\hat{U}$ version in place of some $U$, because the statements hold for an arbitrary moment in time.} In the former case, the claim follows; hence, assume that $p \in \lazy(\hat{U}_t^i)$. By the construction of the fully dynamic algorithm (see~\cref{subsec:adv_point_upd}), the current last $t$-th algorithm $\mathcal{BD}_t$ is restarted with~$\lazy(\hat{U}_t^i)$ as input. We analyze separately the two types of adversarial point updates:
\begin{itemize}
    \item Assume that the adversary deletes an existing point. In this case, it holds that $\lvert\lazy(\hat{U}_t^i)\rvert < |\hat{U}_t^i|$ and the current last $t$-th algorithm reports completion with $p \in \hat{U}_t^t$. This implies that $p \in S_t^t$ and thus it holds that $p \in S_t$, as needed.

    \item Assume that the adversary inserts a new point $p^+$. If no new level is created then all points in $\lazy(\hat{U}_t^i)$ (including $p$ and $p^+$) belong to $S_t^t$, and thus it holds that $p \in S_t$. Otherwise a new level is created, and the (new) last $t$-th algorithm $\mathcal{BD}_t$ is invoked with $\{p^+\}$ as input, which reports completion for the (updated) last level~$t$. Observe that for the (updated) last level~$t$, we have $\hat{U}_{t-1}^i \cup \{p^+\} = \lazy(\hat{U}_{t-1}^i)$, and $\hat{B}_{t-1}^i = \hat{U}_{t-1}^i$ when $\mathcal{BD}_i$ reported completion.
    Therefore, the point $p \in B_j^\old \cap P$ either belongs to~$\hat{B}_{t-1}^i$ or to $S_t^t = S_t$, as required. 
\end{itemize}
\end{proof}

\vspace{1em}
\begin{figure}[H]
    \centering
    \begin{subfigure}{.5\textwidth}
      \centering
      \includegraphics[width=\linewidth]{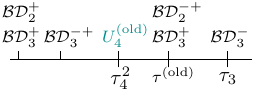}
      \caption{}
      \label{fig:BD_algs_left}
    \end{subfigure}%
    \hfill
    \raisebox{-0.3em} {%
        \begin{subfigure}{.45\textwidth}
        \centering
        \includegraphics[width=\linewidth]{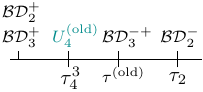}
        \caption{}
        \label{fig:BD_algs_right}
        \end{subfigure}%
    }
    \vspace{1em}
    \caption{Illustration of successive moments in time. Superscript $^+$ indicates that the algorithm begins, $^-$~indicates that the algorithm reports completion, and $^{-+}$ indicates that the algorithm reports completion and then restarts. We just use $^-$ or $^+$ in place of $^{-+}$ whenever additional details are unnecessary for understanding the example provided.
     \\[0.5em]
     \underline{Left figure}: We have $j = 4, \zeta = 2, i = 3$. Just before time $\tau_3$, the global $4$-th ball uses $\hat{B}_4^2$ constructed within $\hat{U}_4^2 = U_4^\old$ defined at time $\tau_4^2$. Subsequently, $\mathcal{BD}_3$ reports completion at time $\tau_3$, replacing  $\hat{B}_4^2$ with $\hat{B}_4^3$.
     \\[0.5em]
     \underline{Right figure}: We have $j = 4, \zeta = 3, i = 2$. Just before time $\tau_2$, the global $4$-th ball uses $\hat{B}_4^3$ constructed within $\hat{U}_4^3 = U_4^\old$ defined at time $\tau_4^3$. Subsequently, $\mathcal{BD}_2$ reports completion at time $\tau_2$, replacing  $\hat{B}_4^3$ with $\hat{B}_4^2$.}
    \label{fig:BD_algs}
\end{figure}

\begin{observation} \label{obs:i_of_t_is_t}
    After an adversarial point update, it holds that $S_t = S_t^t$ and $B_t = B_t^t$. Namely, the value of~$i_t$ is equal to $t$.
\end{observation}
\begin{proof}
    By the construction of the fully dynamic bicriteria approximation algorithm in~\cref{subsec:adv_point_upd}, both the global $t$-th set $S_t$ and the global $t$-th ball $B_t$ are updated to $S_t^t$ and $B_t^t$ respectively, due to the completion of $\mathcal{BD}_t$; consequently, it holds that $i_t = t$.
\end{proof}

The following lemma (\cref{lem:ball_partition_points}) essentially states that each point $p \in P$ belongs to exactly one global $j$-th ball $B_j$, where $j \in [0, t]$. 

\begin{lemma} \label{lem:ball_partition_points}
    After an adversarial point update, every point belongs to exactly one ball $\hat{B}_j^{i_j}$ for some level~$j \in [0, t]$, where $0 \leq i_j \leq j$. 
\end{lemma}
\begin{proof}
    Whenever a new point $p^+$ is inserted by the adversary, the fully dynamic algorithm adds the newly inserted point $p^+$ to $S_t^t = B_t^t = \hat{B}_t^t$, where $t$ is the (updated) last level (see the paragraph ``Modification of the bicriteria approximate solution'' in~\cref{subsec:adv_point_upd}). Based on~\cref{obs:i_of_t_is_t} we have $\hat{B}_t^{i_t} = \hat{B}_t^t$, which means that~$p^+$ belongs only to~$\hat{B}_t^{i_t}$, as needed for $p^+$. Consider another point $p \in P \setminus \{p^+\}$, and let $B^\old_j$ be the unique global $j$-th ball containing $p$ before the adversarial update. The existence of such a ball $B_j^\old$ follows by induction on the adversarial updates, using the induction hypothesis. To simplify the notation, let $\zeta \coloneqq i_j^\old$ denote the level such that $B^\old_j = \bigl(\hat{B}_j^\zeta\bigr)^\old$  before the adversarial update. The analysis continues based on whether the global $j$-th ball is replaced after the adversarial point update:
    \begin{itemize}
        \item Assume that the $j$-th ball is not replaced after the adversarial update (i.e., $B_j = B_j^\old$). In this case, it remains to show that the point~$p$ does not belong to any other ball $\hat{B}_\xi^{i_\xi}$ with $\xi \in [0, t]$ and $\xi \neq j$. Let~$i \in [0, t]$ be the smallest level such that~$\mathcal{BD}_i$ reports completion, and observe that $i > j$; otherwise, the $j$-th ball $B^\old_j$ would be replaced.
        Since $p \in B_j^\old$ then $p \notin \lazy(\hat{U}_{j+1}^\zeta)$ according to~\linecref{line:new_exec_set_lazy} in~\cref{alg:build_lazily}. Note that the $\zeta$-th algorithm $\mathcal{BD}_\zeta$ is the most recent one that reports completion such that $\zeta \leq j$. By the construction of the fully dynamic algorithm, once the $\zeta$-th algorithm $\mathcal{BD}_\zeta$ reports completion, the
        $(j+1)$-th algorithm $\mathcal{BD}_{j+1}$ is restarted with $\lazy(\hat{U}_{j+1}^\zeta)$ as input. Hence at time $\tau_\zeta$, we have $U_{j+1}^{j+1} = \lazy(\hat{U}_{j+1}^\zeta)$ and $p \notin U_{j+1}^{j+1}$. In turn, the point $p$ does not belong to $\lazy(\hat{U}_i^i)$, as otherwise based on~\cref{lem:U_subset_between_alg} the point $p$ would belong to $\lazy(U_{j+1}^{j+1})$ yielding a contradiction. Since $p \notin \lazy(\hat{U}_i^i)$, the point $p$ does not enter any $\xi$-th ball $\hat{B}_\xi^{i_\xi}$ with $\xi > j$ and $i_\xi \geq i > \zeta$, as needed.

        \item Otherwise upon the adversarial update,
        the $j$-th ball $B^\old_j$ is replaced by another ball $B_j^{i_j} = \hat{B}_j^{i_j}$, because the $i_j$-th algorithm $\mathcal{BD}_{i_j}$ reports completion.\footnote{Observe that $\bigl(\hat{B}_j^\zeta\bigr)^\old$ could be replaced by $\hat{B}_j^{i_j}$ even if $i_j^\old = i_j$ (i.e., $\zeta = i_j$). Such an event can occur if the $\zeta$-th algorithm~$\mathcal{BD}_\zeta$ reports completion, and its current $j$-th ball differs from its previous $j$-th ball.} 
        Therefore by~\cref{lem:between_alg_ball,obs:i_of_t_is_t}, the point $p$ belongs to exactly one ball $\hat{B}^{i_j}_\xi$ for some level $\xi \in [i_j, t]$ with $i_\xi = i_j$, as required.
    \end{itemize}
\end{proof}

\begin{lemma} \label{lem:cost_S_worst_case_update_time}
    After an adversarial point update, it holds that:
    \[
        \cost(S) \;\leq\; 2 \myhs \max_{j \in [0, t]} \hat{\nu}_j^{i_j}.
    \]
\end{lemma}
\begin{proof}
    For every level $j \in [0, t]$, the radius of the ball $\hat{B}_j^{i_j}$ is equal to~$\hat{\nu}_j^{i_j}$ according to~\linecref{line:construct_Bj_lazy} in~\cref{alg:build_lazily}.
    By the triangle inequality, the diameter of each $\hat{B}_i^{i_j}$ is at most~$2 \myhs \hat{\nu}_j^{i_j}$.
    Thus, every point $p \in \hat{B}_j^{i_j}$ is within distance $2 \myhs \hat{\nu}_j^{i_j}$ from the internal set~$\hat{S}_j^{i_j}$, which is also the global $j$-th set $S_j$ (i.e., $S_j = \hat{S}_j^{i_j}$). Based on~\cref{lem:ball_partition_points}, the $j$-th balls $B_j = \hat{B}_j^{i_j}$ partition the point set $P$.
    As a result, every point is within distance 
    $2 \cdot \max_{j \in [0, t]} \hat{\nu}_j^{i_j}$ from the maintained bicriteria approximate solution $S$.
\end{proof}

We fix a positive real number $\gamma$ such that $\beta < \gamma < 1$, where $\beta$ is the parameter associated with all radii~$\nu_j^i$~(see~\linecref{line:min_rad_nu_lazy} in~\cref{alg:build_lazily}). The following relationship between $\hat{\nu}_j^{i_j}$ and $\mu(\hat{U}_j^{i_j}, \gamma)$ for the fixed execution set $\hat{U}_j^{i_j}$ is derived (as in~\cref{sec:rel_upp_low_bound}) from the analysis by Mettu and Plaxton in~\cite[Section~3]{mettu2004optimal}. Recall that for a fixed level $j \in [0, t]$, the radius $\hat{\nu}_j^{i_j}$ is computed with respect to the fixed execution set $\hat{U}_j^{i_j}$ which is defined at the $\tau_j^{i_j}$-th adversarial update. Note that time $\tau_j^{i_j}$ is the most recent moment at which $\hat{U}_j^{i_j}$ is defined prior to the completion of $\mathcal{BD}_{i_j}$. 

\begin{lemma}[Lemma 3.3 in~\cite{mettu2004optimal}] \label{lem:nu_2mu_worst_case}
     After an adversarial point update, for every level $j \in [0, t]$ it holds with high probability that:
    \[
        \hat{\nu}_j^{i_j}  \;\leq\; 2 \, \mu(\hat{U}_j^{i_j}, \gamma).
    \]
\end{lemma}

\subsubsection{Size of the Bicriteria Approximate Solution} \label{sec:size_bicr_worst_upd}

The size of the maintained bicriteria approximate solution $S$ depends on the number of levels, which in turn depends on the underlying execution set sizes, as established by the following lemmas.

\begin{lemma} \label{lem:U_i_j_reductionbeta}
    After an adversarial point update, consider a fixed level $i \in [0, t]$ and the corresponding $i$-th algorithm $\mathcal{BD}_i$. Assume that $\mathcal{BD}_i$ is currently lazily constructing its data structures at the $j$-th level (i.e.,~$j$ is the current last internal level within $\mathcal{BD}_i$, and so $j = t^i$). 
    Then it holds with high probability that $|U_\xi^i| \leq (1 - (\beta - \epsilon)) \myhs |U_{\xi-1}^i|$ for every internal level $\xi \in [i+1, j].$
\end{lemma}
\begin{proof}
Consider a fixed internal level $\xi \in [i, j-1]$, and let $\mathsf{upd}$ be the number of adversarial point updates from the definition of $U^i_\xi$ until the definition of $U^i_{\xi + 1}$ (see Line~\ref{line:new_exec_set_lazy} in~\texttt{BuildLazily}$(U_i^i)$ of~\cref{alg:build_lazily}). At the moment the fixed execution set $U_{\xi+1}^i$ is defined, it holds that $\lvert\lazy(U^i_\xi) \setminus U^i_\xi\rvert \leq \mathsf{upd}$, and thus it follows that:
\begin{equation} \label{eq:bound_Ui_with_next}
    \frac{\big\lvert U^i_{\xi + 1}\big\rvert}{\big\lvert U^i_\xi\big\rvert} \;=\; \frac{\big\lvert\!\lazy(U^i_{\xi}) \setminus B^i_{\xi}\big\rvert}{\big\lvert U^i_\xi\big\rvert} \;\le\; \frac{\big\lvert U^i_\xi \setminus B^i_\xi\big\rvert + \mathsf{upd}}{\big\lvert U^i_\xi\big\rvert}.
\end{equation}
In order to bound the number of adversarial updates $\mathsf{upd}$, we need to consider the time \texttt{BuildLazily}($U_i^i$) spends on building the internal $\xi$-th level (the $(\xi - i + 1)$-th iteration of the while-loop at~\linecref{algline:while_rebuild_lazy} in~\cref{alg:build_lazily}). The corresponding time is dominated by the computation of the radius $\nu^i_\xi$ (in~\linecref{line:min_rad_nu_lazy} of~\cref{alg:build_lazily}) and the construction of the ball $B^i_\xi$ within the internal $\xi$-th execution set $U^i_\xi$ (in~\linecref{line:construct_Bj_lazy} of~\cref{alg:build_lazily}). These steps require $O(\alpha \myhs k \myhs |U^i_\xi| \log n)$ time with high probability, since with high probability $|S^i_\xi| = O(\alpha \myhs k \log n)$ by~\linecref{line:sample_Si_lazy} in~\cref{alg:build_lazily} and a Chernoff bound, and because the $\xi$-th radius $\nu^i_\xi$ can be computed using a linear-time selection algorithm.

By the construction of~\texttt{BuildLazily}$(U_i^i)$ in~\cref{alg:build_lazily}, the $i$-th algorithm $\mathcal{BD}_i$ spends $\Theta(\frac{\alpha \myhs k}{\epsilon} \log n \myhs \log \frac{n}{k})$ time per adversarial update. By choosing the constant hidden in the big-Theta term of the worst-case update time to be the same as the constant hidden in the big-O term of the computation time of the $\xi$-th level in $\mathcal{BD}_i$, we obtain that $\mathsf{upd} \leq \epsilon \myhs |U_\xi^i|$. 
By the construction in~\linescref{line:min_rad_nu_lazy}{line:construct_Bj_lazy} of~\cref{alg:build_lazily}, we have $\big\lvert B^i_\xi \big\rvert \geq \beta \myhs \big\lvert U^i_\xi \big\rvert$ and $B^i_\xi \subseteq U^i_\xi$.
Therefore, using Equation~(\ref{eq:bound_Ui_with_next}) it holds that: 
\begin{equation*}
    \frac{\big|U^i_{\xi + 1}|}{|U^i_{\xi}\big|} \;\le\; \frac{(1 - \beta + \epsilon )\big|U^i_{\xi}\big|}{\big|U^i_{\xi}\big|} \;=\; 1 - (\beta - \epsilon).
\end{equation*}
\end{proof}
\noindent By setting the parameters $\beta$ and $\epsilon$ such that $\beta - \epsilon  > 0$, the internal execution sets decrease monotonically in size. 
\begin{corollary} \label{lem:U_i_j_reduction}
    After an adversarial point update, consider a fixed level $i \in [0, t]$ and the corresponding $i$-th algorithm $\mathcal{BD}_i$. Assume that $\mathcal{BD}_i$ is currently lazily constructing its data structures at the $j$-th level (i.e.,~$j$ is the current last internal level within $\mathcal{BD}_i$, and so $j = t^i$). 
    Then it holds that $|U_\xi^i| \leq |U_{\xi-1}^i|$ for every internal level $\xi \in [i+1, j].$
\end{corollary}

\begin{lemma} \label{lem:num_level_i}
    After an adversarial point update, the number of levels of the $i$-th algorithm $\mathcal{BD}_i$ is at most~$O(\log \frac{n}{k})$. 
\end{lemma}
\begin{proof}
    To simplify the notation, let $\zeta = t^i - 1$ be the penultimate level of the $i$-th algorithm $\mathcal{BD}_i$. 
    Based on~\cref{lem:U_i_j_reductionbeta}, it holds that $|U_\zeta^i| \leq (1 - (\beta - \epsilon))^{(\zeta - i)} \cdot |U_i^i|$. Observe that $|U_i^i| \leq 2n$, $|U_\zeta^i| \geq k$, and we choose~$\beta$ and $\epsilon$ such that $\beta > \epsilon$ which implies that $0 < 1 - (\beta - \epsilon) < 1$. Hence it follows that $k \,\leq\, (1 - (\beta - \epsilon))^{(\zeta - i)} \cdot 2n \,\iff\, \zeta \,\leq\, i + \log_{1 - (\beta - \epsilon)} \frac{k}{2n} \,\iff\, t^i \,\leq\, i + 1 + \frac{-\log \frac{2n}{k}}{\log (1 - (\beta - \epsilon))}$. Therefore as $1 - (\beta - \epsilon)$ is a fixed constant in $(0, 1)$, the number of levels of $\mathcal{BD}_i$ is at most $t^i - i + 1 \,\leq\,  2 + \frac{-\log \frac{2n}{k}}{\log (1 - (\beta - \epsilon))} \,=\, O(\log \frac{n}{k})$.
\end{proof}

Consider the following hypothetical scenario for some levels $0 \leq i < j \leq t$. Suppose that a few moments before the $j$-th algorithm $\mathcal{BD}_j$ is about to report completion, the $i$-th algorithm $\mathcal{BD}_i$ reports completion. In this scenario, our fully dynamic bicriteria approximation algorithm restarts $\mathcal{BD}_j$ as a result of the completion of $\mathcal{BD}_i$. Hence, there could be more than $\epsilon \myhs |U^j_j|$ adversarial updates between two completions of the $j$-th algorithm $\mathcal{BD}_j$. Nevertheless, the following lemmas exploit the fact that once $\mathcal{BD}_j$ is restarted, the fixed execution set $U_j^j$ is redefined. 

Based on this hypothetical scenario, the $j$-th algorithm $\mathcal{BD}_j$ \emph{computes} its data structures only if it \emph{reports completion} without being \emph{interrupted} by another $i$-th algorithm~$\mathcal{BD}_i$ with $i < j$. If $\mathcal{BD}_j$ is interrupted, then the computation begins from scratch. Thus, there can be more adversarial updates between two completions of $\mathcal{BD}_j$ than needed to compute all its data structures.

\begin{lemma} \label{lem:BD_i_time}
    Consider a fixed level $i \in [0, t]$. Then the $i$-th algorithm $\mathcal{BD}_i$ requires $O(\alpha \myhs k \myhs |U^i_i| \log n \myhs \log\frac{n}{k})$ time with high probability to compute all its data structures.
\end{lemma}
\begin{proof}
    For a fixed level $i \in [0, t]$, the $i$-th algorithm $\mathcal{BD}_i$ is given in the input the fixed execution set~$U_i^i$ (i.e.,~\texttt{BuildLazily}$(U_i^i)$ in~\cref{alg:build_lazily}).
    The procedure \texttt{BuildLazily}($U_i^i$) scans all levels $j \in [i, t^i]$ and constructs the corresponding balls $B^i_j$. For a fixed level $j \in [i, t^i]$, both the computation of the internal $j$-th radius~$\nu^i_j$ (in~\linecref{line:min_rad_nu_lazy} of~\cref{alg:build_lazily}) and the construction of the internal $j$-th ball $B^i_j$ within the internal  $j$-th execution set $U^i_j$ (in~\linecref{line:construct_Bj_lazy} of~\cref{alg:build_lazily}) require $O(\alpha \myhs k \myhs |U^i_j| \log n)$ time with high probability. This follows from the fact that high probability $|S^i_j| = O(\alpha \myhs k \log n)$ by~\linecref{line:sample_Si_lazy} in~\cref{alg:build_lazily} and a Chernoff bound, and because the $j$-th radius $\nu^i_j$ can be computed using a linear-time selection algorithm.
    Since $j \geq i$ it holds that $|U^i_j| \leq |U^i_i|$ by~\cref{lem:U_i_j_reduction}, and thus $O(\alpha \myhs k \myhs |U^i_j| \log n)$ is upper bounded by $O(\alpha \myhs k \myhs  |U^i_i| \log n)$. As there are at most $O(\log \frac{n}{k})$ levels within the $i$-th algorithm~$\mathcal{BD}_i$ according to~\cref{lem:num_level_i}, the procedure \texttt{BuildLazily}$(U_i^i)$ requires $O(\alpha \myhs k \myhs |U^i_i| \log n \myhs \log \frac{n}{k})$ total time.
\end{proof}

\begin{lemma} \label{lem:BD_i_adv_update_req}
    Consider a fixed level $i \in [0, t]$. Then the $i$-th algorithm $\mathcal{BD}_i$ requires at most $\epsilon \myhs |U^i_i|$ adversarial updates with high probability to compute all its data structures.
\end{lemma}
\begin{proof}
    Based on~\cref{lem:BD_i_time}, the $i$-th algorithm $\mathcal{BD}_i$ requires $O(\alpha \myhs k \myhs |U^i_i| \log n \myhs \log\frac{n}{k})$ total time with high probability to compute all its data structures. By the construction of~\texttt{BuildLazily}$(U_i^i)$ in~\cref{alg:build_lazily}, the $i$-th algorithm $\mathcal{BD}_i$ spends $\Theta(\frac{\alpha \myhs k}{\epsilon} \log n \myhs \log \frac{n}{k})$ time per adversarial point update.
    Thus, the number of adversarial updates required for $\mathcal{BD}_i$ to compute all its data structures is upper bounded by:
    \[
        O\biggl(\frac{\alpha \myhs k \myhs |U^i_i| \log n \myhs \log \frac{n}{k}}{\frac{\alpha \myhs k}{\epsilon} \log n \myhs \log\frac{n}{k}}\biggr) \; = \; \epsilon \myhs |U^i_i|.
    \]
    The reason the big-O notation can be omitted is that the constant hidden in the big-Theta term of the worst-case update time and the constant hidden in the big-O term of the computation time of $\mathcal{BD}_i$ can be chosen to be the same constant, and therefore cancel out. 
\end{proof}

\begin{lemma} \label{lem:symm_diff_U_ij_compl}
    After an adversarial point update, consider two fixed levels $i, j \in [0, t]$. Then 
    there are at most~$\epsilon \myhs |\hat{U}_j^i|$ adversarial updates between the $\tau_j^i$-th and the $\tau_i$-th adversarial updates.
\end{lemma}
\begin{proof}
    By arguments analogous to~\cref{lem:BD_i_time}, we can deduce that the $i$-th algorithm $\mathcal{BD}_i$ requires $O(\alpha \myhs k \myhs |\hat{U}^i_j| \log n \myhs \log\frac{n}{k})$ total time to compute all its data structures within the fixed execution set $\hat{U}^i_j$ (which evolves into $\lazy(\hat{U}^i_j)$) defined at time~$\tau_j^i$. Then by arguments analogous to~\cref{lem:BD_i_adv_update_req}, we can infer that~$\mathcal{BD}_i$ requires at most $\epsilon \myhs |\hat{U}_j^i|$ adversarial updates to compute all its data structures within $\lazy(\hat{U}^i_j)$. Since the computation of all data structures within $\lazy(\hat{U}^i_j)$ finishes at time $\tau_i$, there must be at most $\epsilon \myhs |\hat{U}_j^i|$ adversarial updates between the times $\tau_j^i$ and~$\tau_i$.
\end{proof}

In the rest of the analysis, sometimes we omit the high-probability bounds inherited from previous statements for simplicity, without affecting the validity of our results.\footnote{The high-probability bound applies only to the size and approximation ratio of the solution $S$.}
We remind the reader that $\hat{U}_j^{i_j}$ (where $j \in [0, t]$ and $i_j \in [0, j]$) is defined at time $\tau^{i_j}_j \leq \tau_{i_j}$ and denotes the relevant internal $j$-th execution set of the $i_j$-th algorithm $\mathcal{BD}_{i_j}$. For more details, see~\cref{subsec:aux_cont_upd}. The next lemma (\cref{lem:reduction_global_U}) shows that the relevant execution set sizes decrease geometrically by a constant factor.

\begin{lemma}\label{lem:reduction_global_U}
    After an adversarial point update, consider a fixed level $j \in [1, t]$ and the corresponding~$i_j \in [0, j]$. Then it holds that $|\hat{U}_j^{i_j}| \leq (1 - \beta + 5\epsilon) \myhs |\hat{U}^{i_{j-1}}_{j-1}|$. 
\end{lemma}
\begin{proof} 
    For the fixed level $j \in [1, t]$, whenever $\mathcal{BD}_0$ reports completion (or immediately after the preprocessing phase), the claim holds by~\cref{lem:U_i_j_reductionbeta} since~$i_j = i_{j-1} = 0$.  Let~$\xi$ be the smallest level such that the $\xi$-th algorithm $\mathcal{BD}_\xi$ reports completion (there is always such a level). If~$\xi > j$ then the execution sets~$\hat{U}^{i_{j-1}}_{j-1}$ and~$\hat{U}^{i_j}_j$ (and their indices $i_{j-1}, i_j$) remain unchanged, and thus the claim follows by an induction argument on adversarial point updates. Otherwise if $\xi < j$, we have $i_j = i_{j-1} = \xi$ and the claim holds due to~\cref{lem:U_i_j_reductionbeta} and the construction of the last level $t$.
    
    Hence assume that $\xi$ is equal to $j$, namely $\xi = j$ is the smallest level such that the $\xi$-th algorithm~$\mathcal{BD}_\xi$ reports completion (see also Figure~\ref{fig:BD_global_U_decr}). To simplify the notation, let~$\zeta \coloneqq i_{j-1}^\old = i_{j-1}$ denote the level such that~$S^\old_{j-1} = \bigl(\hat{S}_{j-1}^\zeta\bigr)^\old$ before the adversarial update; similarly let $\tau^\old \coloneqq \tau_\zeta^\old$, $U^\old_{j-1} \coloneqq \bigl(\hat{U}_{j-1}^\zeta\bigr)^\old$, and $U^\old_j \coloneqq \bigl(\hat{U}_j^\zeta\bigr)^\old$ before the adversarial update.\footnote{Note that in fact the superscript $^\old$ on the right-hand sides is not needed, since $\mathcal{BD}_\xi$ (where $\zeta < \xi = j$) does not affect~$S_{j-1}$ and~$i_{j-1}$.} Since the $j$-th algorithm $\mathcal{BD}_j$ reports completion, the current moment is~$\tau_j$.
    
    \begin{claim} \label{clm:U_j_prev_U_j}
        There are at most $\epsilon \myhs \lvert \lazy(U_{j-1}^\old) \rvert \leq \epsilon \myhs (1+\epsilon) |U_{j-1}^\old|$ adversarial updates from the $\tau^\old$-th adversarial update until the current $\tau_j$-th adversarial update. 
    \end{claim}
    \begin{proof}
        Observe that the $(j-1)$-th algorithm $\mathcal{BD}_{j-1}$ restarted with $\lazy(U^\old_{j-1})$ as input, once~$\mathcal{BD}_\zeta$ reported completion at time $\tau^\old$. Based on~\cref{lem:BD_i_adv_update_req}, the $(j-1)$-th algorithm $\mathcal{BD}_{j-1}$ can report completion after at most~$\epsilon \myhs \lvert \lazy(U_{j-1}^\old) \rvert$ adversarial updates. Thus, there exists a level $i \in [0, j-1]$ such that the $i$-th algorithm~$\mathcal{BD}_i$ reports completion within $\epsilon \myhs \lvert \lazy(U_{j-1}^\old) \rvert$ adversarial updates after the $\tau^\old$-th adversarial update. Consequently, based on the definition of~$\tau^\old$ which would be affected due to~$\mathcal{BD}_i$, 
        there are at most~$\epsilon \myhs \lvert \lazy(U_{j-1}^\old) \rvert$ adversarial updates from time~$\tau^\old$ until time $\tau_j$, as needed. Based on~\cref{lem:symm_diff_U_ij_compl}, there are at most~$\epsilon \myhs |U_{j-1}^\old|$ adversarial updates from the definition of $U_{j-1}^\old$ until time $\tau^\old$, which implies that $\lvert \lazy(U_{j-1}^\old) \rvert \leq (1+\epsilon) \myhs |U_{j-1}^\old|$.
        Therefore, between time $\tau^\old$ and time $\tau_j$ there are at most $\epsilon \myhs \lvert \lazy(U_{j-1}^\old) \rvert \leq \epsilon \myhs (1+\epsilon) |U_{j-1}^\old|$ adversarial point updates, as required (see also Figure~\ref{fig:BD_global_U_decr}).
    \end{proof}
    \noindent
    Since the $j$-th algorithm $\mathcal{BD}_j$ reports completion, at the current moment~$\tau_j$ we have $i_j = j$. 
    The definition of~$\hat{U}_j^{i_j}$ occurs between the $\tau^\old$-th and the $\tau_j$-th adversarial updates. 
    Moreover, since $\mathcal{BD}_j$ does not affect~$i_{j-1}$, we have $\hat{U}_{j-1}^{i_{j-1}} = U_{j-1}^\old$. Using~\cref{clm:U_j_prev_U_j} and~\cref{lem:symm_diff_U_ij_compl}, the ratio between the set $\hat{U}^{i_j}_j$ and the set~$\hat{U}^{i_{j-1}}_{j-1}$ is analyzed as follows:
    \begin{align*}
        \frac{|\hat{U}^{i_j}_j|}{|\hat{U}^{i_{j-1}}_{j-1}|} \;&\leq\; \frac{|U^\old_j| + \epsilon \myhs |U_j^\old| + (1+\epsilon) \myhs \epsilon \myhs |U_{j-1}^\old|}{|U_{j-1}^\old|}.
    \end{align*}
    Based on~\cref{lem:U_i_j_reductionbeta}, it holds that~$|U_j^\old| \,\leq\, (1 - (\beta - \epsilon)) \myhs |U_{j-1}^\old|$, and thus it follows that:
    \begin{align*}
        \frac{|\hat{U}^{i_j}_j|}{|\hat{U}^{i_{j-1}}_{j-1}|} \;&\leq\; \frac{(1-(\beta - \epsilon)) \myhs |U_{j-1}^\old| + \epsilon \myhs (1-(\beta - \epsilon)) \myhs |U_{j-1}^\old| + (1+\epsilon) \myhs \epsilon \myhs |U_{j-1}^\old|}{|U_{j-1}^\old|}  \;\leq\; 1 - \beta + 5\epsilon.
    \end{align*}
\end{proof}

\vspace{1em}
\begin{figure}[H]
    \centering
    \includegraphics[width=0.7\linewidth]{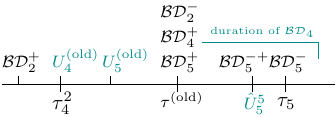}
    \vspace{1em}
    \caption{Illustration of successive moments in time. Superscript $^+$ indicates that the algorithm begins, $^-$~indicates that the algorithm reports completion, and $^{-+}$ indicates that the algorithm reports completion and then restarts. We just use $^-$ or $^+$ in place of $^{-+}$ whenever additional details are unnecessary for understanding the example provided. \\[0.1em]
    We have $\zeta = 2$ and $\xi = j = 5$. The execution set $\hat{U}^2_4 = U_4^\old$ is defined at time $\tau_4^2$. At time $\tau_2 = \tau^\old$, $\mathcal{BD}_2$ reports completion and $\mathcal{BD}_4$ is restarted with $\lazy(U_4^\old) = \lazy(\hat{U}^2_4)$ as input. $\mathcal{BD}_5$ reports completion at the current time $\tau_5$,  and $\hat{U}_5^5$ was constructed when $\mathcal{BD}_5$ was previously restarted. The ``duration of $\mathcal{BD}_4$'' is upper bounded by $\epsilon \myhs \lvert \lazy(U_4^\old) \rvert \leq \epsilon \myhs (1+\epsilon) |U_4^\old|$ according to~\cref{clm:U_j_prev_U_j}.}
    \label{fig:BD_global_U_decr}
\end{figure}

\begin{lemma} \label{lem:num_levels_t}
    After an adversarial point update, the value of the last level $t$ is at most $O(\log \frac{n}{k})$.    
\end{lemma}
\begin{proof}
    Based on~\cref{lem:reduction_global_U}, it holds that $|\hat{U}_{t-1}^{i_{t-1}}| \leq (1 - \beta + 5\epsilon)^{t-1} \cdot |U_0^{i_0}|$. Observe that $|U_0^{i_0}| \leq 2n$, $|\hat{U}_{t-1}^{i_{t-1}}| \geq k$, and we choose~$\beta$ and $\epsilon$ such that $\beta > 5\epsilon$ which implies that $0 < 1 - \beta + 5\epsilon < 1$. Hence it follows that $k \,\leq\, (1 - \beta + 5\epsilon)^{t-1} \cdot 2n \,\iff\, t \,\leq\, 1 + \log_{1 - \beta + 5\epsilon} \frac{k}{2n} \,\iff\, t \,\leq\, 1 + \frac{-\log \frac{2n}{k}}{\log (1 - \beta + 5\epsilon)}$. Therefore as $1 - \beta + 5\epsilon$ is a fixed constant in $(0, 1)$, the number of levels is at most $t + 1 \,\leq\,  2 + \frac{-\log \frac{2n}{k}}{\log (1 - \beta + 5\epsilon)} \,=\, O(\log \frac{n}{k})$.
\end{proof}

\begin{lemma} \label{lem:size_of_S_worst_upd}
    After an adversarial point update, the size of the bicriteria approximate solution $S$ is with high probability at most $O(k \log n \log \frac{n}{k})$.
\end{lemma}
\begin{proof}
    By the construction of the fully dynamic bicriteria approximation algorithm, we have $S = \bigcup_{j = 0}^t \hat{S}_j^{i_j}$. Each set $\hat{S}_j^{i_j}$ is constructed in~\linecref{line:sample_Si_lazy} of~\cref{alg:build_lazily}, by including each point of $\hat{U}_j^{i_j}$ independently with probability $\min\Big(\frac{\alpha \myhs k \log n}{|\hat{U}^{i_j}_j|}, 1\Big)$, where $\alpha$ is a sufficiently large constant. Hence by a straightforward application of Chernoff bound, it holds with high probability that $|\hat{S}_j^{i_j}| = O(\alpha \myhs k \log n)$. Also due to~\linescref{algline:while_rebuild_lazy}{line:last_level_lazy} in~\cref{alg:build_lazily} and by the construction of $\mathcal{BD}_t$, the size of $\hat{S}_t^{i_t}$ is at most $O(k \log n \myhs \log \frac{n}{k})$. Since $t = O(\log \frac{n}{k})$ according to~\cref{lem:num_levels_t}, by applying a union bound over the $t$ levels, it holds with high probability that $|S| = O(k \log n \log \frac{n}{k})$. 
\end{proof}

\subsubsection{Lower Bound on the Optimal Cost} \label{subsec:lower_bound_worst_case_upd}

\begin{lemma} \label{lem:symm_diff_Uii}
    After an adversarial point update, consider a fixed level $i \in [0, t]$.
    Then it holds that $|U_i^i \triangle \lazy(U_i^i)| \leq \epsilon \myhs |U^i_i|$.
\end{lemma}
\begin{proof}
    Based on~\cref{lem:BD_i_adv_update_req}, the $i$-th algorithm $\mathcal{BD}_i$ can
    compute all its data structures after at most~$\epsilon \myhs |U^i_i|$ adversarial updates.
    This implies that within~$\epsilon \myhs |U^i_i|$ adversarial updates after the restart of $\mathcal{BD}_i$, there must be a level~$\xi \in [0, i]$ such that the $\xi$-th algorithm $\mathcal{BD}_\xi$ reports completion. Since the fully dynamic algorithm restarts all $\mathcal{BD}_j$ for $j \in [\xi, t]$, it follows that
    $\mathcal{BD}_i$ is restarted and the fixed execution set $U^i_i$ is redefined; in turn, $\lazy(U^i_i)$ of~\cref{def:lazy_set} is reinitialized to $U_i^i$. Therefore, the fixed execution set $U^i_i$ remains \emph{lazy} for at most $\epsilon \myhs |U^i_i|$ adversarial point updates, which implies that $|U_i^i \triangle \lazy(U_i^i)| \leq \epsilon \myhs |U^i_i|$.
\end{proof}

The following statements are derived by reasoning analogous to that of~\cref{lem:sym_Ui_upper_bound,lem:symmetric_diff,obs:equal_lazy_sets,lem:rel_mu_old_curr,lem:rel_OPT_mu}. The goal is to demonstrate that for any level $j \in [0, t]$, some appropriately chosen values $\mu(\cdot, \cdot)$ serve as lower bounds on the current optimal $k$-center cost $R^*$.

\begin{observation} \label{obs:equal_lazy_sets_worst_case} 
    After an adversarial point update, consider a fixed level $j \in [0, t]$. Then~$\lazy(\hat{U}_j^{i_j})$ is the same set as $\lazy(U_j^j)$ (i.e., $\lazy(\hat{U}_j^{i_j}) = \lazy(U_j^j)$).
\end{observation}
\begin{proof}
    For the fixed level $j \in [0, t]$, the $i_j$-th algorithm~$\mathcal{BD}_{i_j}$ with~$i_j \in [0, j]$ has reported completion most recently. By the construction of the fully dynamic algorithm, at the $\tau_{i_j}$-th adversarial update when $\mathcal{BD}_{i_j}$ reported completion, the $j$-th algorithm $\mathcal{BD}_j$ restarted with $\lazy(\hat{U}_j^{i_j})$ as input. Hence at time $\tau_{i_j}$, we had $U_j^j = \lazy(\hat{U}_j^{i_j})$, and thus by construction and the definition of $\lazy(\cdot)$ (in~\cref{def:lazy_set}), it follows that $\lazy(U_j^j) = \lazy(\hat{U}_j^{i_j})$.
\end{proof}

\begin{lemma} \label{lem:symm_diff_U_ij_compl_curr}
    After an adversarial point update, for every level $j \in [0, t]$ it holds that $|\hat{U}_j^{i_j} \triangle\myhs U_j^j| \leq~\epsilon \myhs |\hat{U}_j^{i_j}|$.
\end{lemma}
\begin{proof}
    For a fixed level $j \in [0, t]$, the value of $i_j \in [0, j]$ is the level such that $S_j = \hat{S}_j^{i_j}$. The  $i_j$-th algorithm~$\mathcal{BD}_{i_j}$ has reported completion at the $\tau_{i_j}$-th adversarial update, and the fixed execution set $\hat{U}_j^{i_j}$ has been constructed at the $\tau_j^{i_j}$-th adversarial update. Notice that there is no level $\xi \in [0, j]$ such that~$\tau_{i_j} < \tau_\xi$. Hence, the execution set $U_j^j$ is defined as $\lazy(\hat{U}_j^{i_j})$ at time $\tau_{i_j}$ when $\mathcal{BD}_j$ is restarted as a result of the completion of~$\mathcal{BD}_{i_j}$. By~\cref{lem:symm_diff_U_ij_compl},
    there are at most $\epsilon \myhs |\hat{U}_j^{i_j}|$ adversarial point updates between the $\tau_j^{i_j}$-th and the $\tau_{i_j}$-th adversarial updates.
    Therefore, it follows that $|\hat{U}_j^{i_j} \triangle\myhs U_j^j| \leq \epsilon \myhs |\hat{U}_j^{i_j}|$, as needed.
\end{proof}

\begin{lemma} \label{lem:symmetric_diff_worst_update}
    After an adversarial point update, for every level $j \in [0, t]$ it holds that:
    \[
        |\hat{U}_j^{i_j} \triangle \lazy(\hat{U}_j^{i_j})| \leq 3\myhs \epsilon \myhs |\hat{U}_j^{i_j}|.
    \]
\end{lemma}
\begin{proof}
    By~\cref{lem:symm_diff_U_ij_compl_curr} we have
    $|\hat{U}_j^{i_j} \triangle\myhs U_j^j| \leq \epsilon \myhs |\hat{U}_j^{i_j}|$, and by~\cref{lem:symm_diff_Uii} it holds that $|U_j^j \triangle \lazy(U_j^j)| \leq \epsilon \myhs |U^j_j|$.
    Thus based on~\cref{obs:sym_diff_prop}, we obtain that:
    \begin{align*}
        |\hat{U}_j^{i_j} \triangle \lazy(U_j^j)| \;&\leq\; |\hat{U}_j^{i_j} \triangle\myhs U_j^j| \,+\, |U_j^j \triangle \lazy(U_j^j)| \\
        \;&\leq\; \epsilon \myhs |\hat{U}_j^{i_j}| \,+\, \epsilon \myhs |U_j^j|.
    \end{align*}
    According to~\cref{obs:equal_lazy_sets_worst_case}, we have $\lazy(\hat{U}_j^{i_j}) = \lazy(U_j^j)$, and so it follows that:
    \begin{align*}
        |\hat{U}_j^{i_j} \triangle \lazy(\hat{U}_j^{i_j})| 
        \;&\leq\; \epsilon \myhs |\hat{U}_j^{i_j}| \,+\, \epsilon \myhs |U_j^j|.
    \end{align*}
    By the construction of the fully dynamic algorithm, the fixed execution set $U_j^j$ is initialized to $\lazy(\hat{U}_j^{i_j})$ at time $\tau_{i_j}$. Hence using~\cref{lem:symm_diff_U_ij_compl_curr}, we deduce $|U_j^j| \leq (1+\epsilon) |\hat{U}_j^{i_j}|$. Therefore, it holds that
    $|\hat{U}_j^{i_j} \triangle \lazy(\hat{U}_j^{i_j})| \leq 
    \epsilon \myhs |\hat{U}_j^{i_j}| \,+\, \epsilon \myhs (1+\epsilon) |\hat{U}_j^{i_j}|$, and the claim follows because $\epsilon < 1$.
\end{proof}

\antonis{add picture for lemma 5.24.}

\begin{lemma} \label{lem:rel_mu_old_curr_worst_case_upd}
    After an adversarial point update, consider a fixed level $j \in [0, t]$. Then for $\tilde{\gamma} \geq \frac{\gamma + 3 \myhs \epsilon}{1 - 3 \myhs \epsilon}$, $0 < \epsilon < \frac{1}{6}$, and $0 < \gamma < 1 - 6\epsilon$, it holds that:
    \[
        \mu(\hat{U}_j^{i_j}, \gamma) \leq 2 \cdot \mu(\lazy(\hat{U}_j^{i_j}), \tilde{\gamma}).
    \]
\end{lemma}
\begin{proof}
    Let $X_j \subseteq P$ be the set that satisfies the requirements for $\mu(\lazy(\hat{U}_j^{i_j}), \tilde{\gamma})$ in~\cref{def:mu_i}. To simplify the notation, let $\hat{\mu}_j \coloneqq \mu(\lazy(\hat{U}_j^{i_j}), \tilde{\gamma})$. In turn, it holds that:
    \[
        \lvert\ball[\lazy(\hat{U}_j^{i_j}), X_j, \hat{\mu}_j]\rvert \,\geq\, \tilde{\gamma}\myhs\lvert\lazy(\hat{U}_j^{i_j})\rvert.
    \]
    Since by~\cref{lem:symmetric_diff_worst_update} we have $|\hat{U}_j^{i_j} \triangle \lazy(\hat{U}_j^{i_j})| \leq 3\myhs \epsilon \myhs |\hat{U}_j^{i_j}|$, the size of $\ball[\hat{U}_j^{i_j}, X_j, \hat{\mu}_j]$ is evaluated as follows:
    \begin{align*}
        \lvert\ball[\hat{U}_j^{i_j}, X_j, \hat{\mu}_j]\rvert &\;=\; \lvert\ball[\hat{U}_j^{i_j} \cup \lazy(\hat{U}_j^{i_j}), X_j, \hat{\mu}_j] \setminus \ball[\lazy(\hat{U}_j^{i_j}) \setminus \hat{U}_j^{i_j}, X_j, \hat{\mu}_j]\rvert \\
        &\;\geq \lvert\ball[\lazy(\hat{U}_j^{i_j}), X_j, \hat{\mu}_j]\rvert - \lvert\ball[\lazy(\hat{U}_j^{i_j}) \setminus \hat{U}_j^{i_j}, X_j, \hat{\mu}_j]\rvert \\
        &\;\geq\; \tilde{\gamma}\myhs\lvert\lazy(\hat{U}_j^{i_j})\rvert - \lvert\lazy(\hat{U}_j^{i_j}) \setminus \hat{U}_j^{i_j}\rvert 
        \;\geq\; \frac{\gamma + 3 \myhs \epsilon}{1 - 3 \myhs \epsilon} \myhs \lvert\lazy(\hat{U}_j^{i_j})\rvert - 3 \myhs \epsilon \myhs \lvert\hat{U}_j^{i_j}\rvert \\
        &\;\geq\; \frac{\gamma + 3 \myhs \epsilon}{1 - 3 \myhs \epsilon} \myhs \big(|\hat{U}_j^{i_j}| - 3 \myhs \epsilon \myhs |\hat{U}_j^{i_j}|\big) - 3 \myhs \epsilon \myhs |\hat{U}_j^{i_j}| \\
        &\;=\; \gamma \myhs |\hat{U}_j^{i_j}|.
    \end{align*}
    Notice that $X_j$ may not be a subset of the point set $P^{(\tau_j^{i_j})}$ at the moment $\hat{U}_j^{i_j}$ is defined (i.e., at the $\tau_j^{i_j}$-th adversarial update). However using the triangle inequality, we can find a subset $Y_j$ of~$P^{(\tau_j^{i_j})}$ such that $\ball[\hat{U}_j^{i_j}, X_j, \hat{\mu}_j] \subseteq \ball[\hat{U}_j^{i_j}, Y_j, 2\myhs\hat{\mu}_j]$. In other words, there exists a subset of points $Y_j \subseteq~P^{(\tau_j^{i_j})}$ with $|Y_j| \leq k$ such that $\lvert\ball[\hat{U}_j^{i_j}, Y_j, 2\myhs\hat{\mu}_j]\rvert \geq \gamma\myhs|\hat{U}_j^{i_j}|$. Therefore based on~\cref{def:mu_i} for $\mu(\hat{U}_j^{i_j}, \gamma)$, this implies that $\mu(\hat{U}_j^{i_j}, \gamma) \leq 2\myhs\hat{\mu}_j = 2 \cdot \mu(\lazy(\hat{U}_j^{i_j}), \tilde{\gamma})$, as required.
\end{proof}

\begin{lemma} \label{lem:rel_OPT_mu_worst_update}
    After an adversarial point update, consider a fixed level $j \in [0, t]$. Then for $\tilde{\gamma} = \frac{\gamma + 3 \myhs \epsilon}{1 - 3 \myhs \epsilon}$, $0 < \epsilon < \frac{1}{6}$, and $0 < \gamma < 1 - 6 \epsilon$, it holds that:
    \[
        R^* \;\geq\; \mu(\lazy(\hat{U}_j^{i_j}), \tilde{\gamma}).
    \]
\end{lemma}
\begin{proof}
    Let $C^* \subseteq P$ be a current optimal $k$-center solution. Since $R^* = \max_{p \in P} \dist(p, C^*)$ and $\lazy(\hat{U}_j^{i_j}) \subseteq P$, it follows that:
    \begin{enumerate}
        \item $\lvert\ball[\lazy(\hat{U}_j^{i_j}), C^*, R^*]\rvert \geq \lvert\lazy(\hat{U}_j^{i_j})\rvert$.
        \item $\lvert\lazy(\hat{U}_j^{i_j}) \setminus \ball(\lazy(\hat{U}_j^{i_j}), C^*, R^*)\rvert \geq 0$.
    \end{enumerate}
    Note that as $0 < \tilde{\gamma} < 1$, we have $\tilde{\gamma} \myhs \lvert\lazy(\hat{U}_j^{i_j})\rvert < \lvert\lazy(\hat{U}_j^{i_j})\rvert$ and $(1 - \tilde{\gamma}) \myhs \lvert\lazy(\hat{U}_j^{i_j})\rvert > 0$. Hence, there is a value $R \leq R^*$ such that: 
    \begin{enumerate}
        \item $\lvert\ball[\lazy(\hat{U}_j^{i_j}), C^*, R]\rvert \,\geq\, \tilde{\gamma} \myhs \lvert\lazy(\hat{U}_j^{i_j})\rvert$.
        \item $\lvert\lazy(\hat{U}_j^{i_j}) \setminus \ball(\lazy(\hat{U}_j^{i_j}), C^*, R)\rvert \,\geq\, (1-\tilde{\gamma}) \myhs \lvert\lazy(\hat{U}_j^{i_j})\rvert$.
    \end{enumerate}
    Since $\mu(\lazy(\hat{U}_j^{i_j}), \tilde{\gamma})$ is the minimum value that satisfies the requirements in~\cref{def:mu_i} for $\lazy(\hat{U}_j^{i_j})$ and $\tilde{\gamma}$, we conclude that $\mu(\lazy(\hat{U}_j^{i_j}), \tilde{\gamma}) \,\leq\, R \,\leq\, R^*$.
\end{proof}

\subsubsection{Finishing the Proof of~\cref{thm:bicr_alg_worst_case_update_time}} \label{sec:finish_proof_update_time}

By combining the previous observations and lemmas, and by setting $\tilde{\gamma} = \frac{\gamma + 3 \myhs \epsilon}{1 - 3 \myhs \epsilon}$, $0 < \epsilon < \frac{1}{6}$, and $0 < \gamma < 1 - 6 \epsilon$, we establish the approximation ratio of our bicriteria approximate solution $S$ for the $k$-center clustering problem. Finally, we finish the proof of~\cref{thm:bicr_alg_worst_case_update_time} which we restate for convenience.

\begin{lemma} \label{lem:approx_ratio_update_time}
    After an adversarial point update, it holds with high probability that
    $\cost(S) \leq 8 \myhs R^*$, where $R^*$ is the current optimal $k$-center cost.
\end{lemma}
\begin{proof}
    Based on~\cref{lem:cost_S_worst_case_update_time}, it holds that $\cost(S) \;\leq\; 2 \myhs \max_{j \in [0, t]} \hat{\nu}_j^{i_j}$. According to~\cref{lem:nu_2mu_worst_case}, for every level $j \in [0, t]$ we have $\hat{\nu}_j^{i_j}  \;\leq\; 2 \, \mu(\hat{U}_j^{i_j}, \gamma)$. In addition, by~\cref{lem:rel_mu_old_curr_worst_case_upd} for every level $j \in [0, t]$ we obtain~$\mu(\hat{U}_j^{i_j}, \gamma) \leq 2 \cdot \mu(\lazy(\hat{U}_j^{i_j}), \tilde{\gamma})$.
    Therefore, the $k$-center cost of the bicriteria approximate solution $S$ can be upper bounded by:
    \[
        \cost(S) \;\leq\; 8 \cdot \max_{j \in [0, t]} \mu\big(\lazy(\hat{U}_j^{i_j}), \tilde{\gamma}\big).
    \]
    Finally by~\cref{lem:rel_OPT_mu_worst_update}, it follows that:
    \[
        \cost(S) \;\leq\; 8 \cdot \max_{j \in [0, t]} R^* \;=\; 8 \cdot R^*.
    \]
\end{proof}

\bicralgworstcaseupdatetime*
\begin{proof}
    The approximation ratio holds by~\cref{lem:approx_ratio_update_time} and the size of the solution $S$ by~\cref{lem:size_of_S_worst_upd}.
    The fully dynamic algorithm maintains an $i$-th algorithm $\mathcal{BD}_i$ for each level $i \in [0, t]$.\footnote{Note that outdated $\mathcal{BD}_j$ with $j > t$ must be terminated (which is safe for the correctness analysis).} Since $t = O(\log\frac{n}{k})$ according to~\cref{lem:num_levels_t}, the guarantee on the worst-case update time follows. 
\end{proof}

\subsection{Fully Dynamic Bicriteria Approximation Algorithm} \label{sec:merged_bicr_alg}
Equipped with~\cref{thm:bicr_alg_worst_case,thm:bicr_alg_worst_case_update_time}, we now have all the necessary tools to merge the previously presented algorithms into a single fully dynamic bicriteria approximation algorithm, as stated in~\cref{thm:bicr_alg_merged}.

\bicralgmerged*

Let $\mathcal{REC}$ be the algorithm of~\cref{thm:bicr_alg_worst_case} and $\mathcal{UPD}$ be the algorithm of~\cref{thm:bicr_alg_worst_case_update_time}.
Roughly speaking, $\mathcal{REC}$ employs~\texttt{RebuildFromLayer}$(\hspace{0.5pt})$ (in~\cref{alg:aux_rebuild_proced}) to produce a set $S$ and~\texttt{LazySync}$(\hspace{0.5pt})$ (in~\cref{alg:lazy_sync}) to synchronize the solution $\hat{S}$ with the set $S$ during the transition phases. The key observation is that~\texttt{RebuildFromLayer}$(\hspace{0.5pt})$ can be replaced by~$\mathcal{UPD}$, since~$\mathcal{UPD}$ essentially lazily rebuilds the set~$S$ to which the bicriteria approximate solution $\hat{S}$ transitions. The merged algorithm results in even `lazier' updates, since the set $S$ from $\mathcal{UPD}$ may be more outdated when used within $\mathcal{REC}$. We remark that the merged algorithm outputs $\hat{S}$, and the final fully dynamic algorithm of~\cref{thm:bicr_alg_merged} is described in the following paragraph.

\paragraph{The final fully dynamic bicriteria approximation algorithm.} 
During the $\ell$-th transition phase, the bicriteria approximate solution~$\hat{S}$ gradually transitions to the target set $S^{(\ell)}$, a snapshot of the set $S$ taken at the beginning of the $\ell$-th transition phase. In this final version of the algorithm, the set $S$ is maintained by the dynamic algorithm $\mathcal{UPD}$, and thus the execution sets originate from the dynamic algorithm $\mathcal{UPD}$.

When the adversary sends a point update, it is passed to $\mathcal{UPD}$ to maintain the set~$S$ and then to~$\mathcal{REC}$. Here, $\mathcal{REC}$ is simulated without using the procedure~\texttt{CheckThresholdRebuild}$(\hspace{0.5pt})$, since $\mathcal{UPD}$ now handles the set $S$. Mainly
the procedure~\texttt{LazySync}$(\hspace{0.5pt})$ is invoked, which gradually converts the solution $\hat{S}$ to $S^{\ellexp} \cup I^\ellexp$. At the end of the $\ell$-th transition phase, the maintained bicriteria approximate solution $\hat{S}$ is equal to $S^{\ellexp} \cup I^\ellexp$, and the $(\ell+1)$-th transition phase begins. 

\subsubsection{Analysis of the Fully Dynamic Bicriteria Approximation Algorithm}
Since the set $S$ is now maintained by $\mathcal{UPD}$, compared to the original version of $\mathcal{REC}$, the final dynamic algorithm results in even `lazier' updates; it suffices to appropriately adjust the parameters $\epsilon$ and $\lambda$. First we discuss the size of the solution $\hat{S}$, then its approximation ratio, and finally we conclude with the proof of~\cref{thm:bicr_alg_merged}.

To bound the size of the maintained solution $\hat{S}$, we need to upper bound the length of each transition phase. Since the target set $S^{(\ell)}$ is the output of~$\mathcal{UPD}$ at the beginning of the $\ell$-th transition phase, we have $|S^{(\ell)}| \leq \rho \myhs \alpha \myhs k \log n \log \frac{n}{k}$ for  sufficiently large constants $\alpha$ and $\rho$ (see~\cref{lem:size_of_S_worst_upd}). For simplicity, we make another adjustment to the dynamic algorithm $\mathcal{REC}$; in the procedure \texttt{LazySync}$(\hspace{0.5pt})$ the solution $\hat{S}$ is modified by $\frac{32 \rho}{\lambda^2}$ points at each adversarial update. Hence, we obtain the following analog of~\cref{lem:num_updates_phase}. 

\begin{lemma}\label{lem:num_updates_phase_merge}
      For every transition phase index $\ell \geq 0$, there are at most $ \lambda \myhs \alpha \myhs k \log n \myhs \log \frac{n}{k}$ adversarial updates during the $\ell$-th transition phase.
\end{lemma}

\begin{proof}
The arguments are similar to those in~\cref{lem:len_trans} and~\cref{lem:num_updates_phase}. Specifically, the length of the $\ell$-th transition phase is upper bounded as follows:  
\begin{equation*}
    \left\lceil \big(|S^{(\ell -1)}| + |I^{(\ell -1)}| + |S^\ellexp|\big)\cdot\frac{\lambda^2}{32 \rho} \right\rceil \;\leq\; \big(|S^{(\ell -1)}| + |I^{(\ell -1)}| + |S^\ellexp|\big) \cdot \frac{\lambda^2}{16 \rho}.
\end{equation*}
Using the size of the target sets and an induction argument, it follows that:
\begin{align*}
    \big(|S^{(\ell-1)}| + |I^{(\ell-1)}| + |S^\ellexp|\big) \cdot \frac{\lambda^2}{16 \rho} &\;\leq\; \big(2 \rho \myhs \alpha \myhs k \log n \myhs \log \frac{n}{k} + \lambda \myhs \alpha \myhs k \log n \myhs \log \frac{n}{k}\big) \cdot \frac{\lambda^2}{16 \rho} \\
    &\;\leq\; \big((2\rho + \lambda) \myhs \alpha \myhs k \log n \myhs \log \frac{n}{k}\big) \cdot \frac{\lambda^2}{16 \rho} \\
    &\;\leq\; \lambda \myhs \alpha \myhs k \log n \myhs \log \frac{n}{k}.
\end{align*}
\noindent
The first inequality follows from the induction hypothesis for $|I^{(\ell-1)}|$ and the upper bound on $|S^{(\ell - 1)}|$ and~$|S^\ellexp|$, and the last inequality holds since $0 < \lambda < 1$ and $\rho > 1$.
\end{proof}

In turn, we get that $|\hat{S}^{(\ell)}| = |S^\ellexp| + |I^\ellexp| \leq (\rho + \lambda) \myhs \alpha \myhs k \log n \myhs \log \frac{n}{k}$. The following analog of~\cref{lem:size_of_S_hat} bounds the size of the maintained bicriteria approximate solution.

\begin{lemma}\label{lem:size_of_S_hat_merge}
      After an adversarial point update, it holds that $|\hat{S}| \leq 2\myhs (\rho + \lambda) \myhs \alpha \myhs k \log n \myhs \log \frac{n}{k}$. 
\end{lemma} 

We continue with the analysis of the approximation ratio. For the notation and definitions of the dynamic algorithm~$\mathcal{REC}$, see \cref{sec:bicr_recourse} and specifically~\cref{sec:upper_bound,sec:lower_bound}.
For the notation and definitions of the dynamic algorithm~$\mathcal{UPD}$, see \cref{sec:bicr_update_time} and specifically~\cref{subsec:aux_cont_upd}. Similar to~\cref{sec:upper_bound,sec:lower_bound} for the dynamic algorithm~$\mathcal{REC}$, we set~$\hat{U}_j^{(\ell)}$ equal to $\hat{U}_j^{i_j}$ at the beginning of the $\ell$-th transition phase, for some level $j \in [0, t]$. The main difference in this final dynamic algorithm is the size of the symmetric difference between the \emph{relevant} execution sets and their \emph{lazy} counterparts (see Figure~\ref{fig:U_i_merge}).

\vspace{1em}
\begin{figure}[H]
    \centering
    \includegraphics[width=0.8\linewidth]{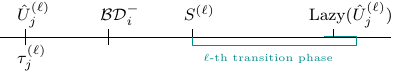}
    \vspace{1em}
    \caption{Illustration of successive moments in time. Superscript $^{-}$ indicates that the ``subalgorithm'' of~$\mathcal{UPD}$ reports completion. The value of $i$ is the level such that the $i$-th algorithm $\mathcal{BD}_i$ has reported completion most recently with $i \in [0, j]$. At time $\tau_j^{(\ell)}$, we have $\hat{U}^\ellexp_j = \hat{U}^i_j$.
    }
    \label{fig:U_i_merge}
\end{figure}

\begin{lemma}
For every transition phase index $\ell \geq 0$, consider a fixed level $j \in [0, t)$ during the $\ell$-th transition phase. Then both of the following inequalities hold:
    \begin{itemize}
        \item $|\hat{U}^{(\ell-1)}_j \triangle \lazy(\hat{U}_j^{(\ell-1)})| \leq 3\myhs ( \lambda+ \epsilon) \myhs |\hat{U}^{(\ell-1)}_j|.$

        \item $|\hat{U}^\ellexp_j \triangle \lazy(\hat{U}_j^\ellexp)| \leq (\lambda+3\epsilon) \myhs \lambda \myhs |\hat{U}^\ellexp_j|.$
    \end{itemize}
\end{lemma}

\begin{proof}
The arguments are similar to those in~\cref{lem:symmetric_diff}. Let $U^{(\ell)}_j$ be the set $\lazy(\hat{U}_j^{(\ell)})$ at the beginning of the $\ell$-th transition phase for any transition phase index $\ell \geq 0$. Based on~\cref{lem:symmetric_diff_worst_update}, we have $|\hat{U}^{(\ell-1)}_j \triangle \myhs U^{(\ell-1)}_j| \le 3\epsilon \myhs |\hat{U}^{(\ell-1)}_j|$, and thus it follows that:
\begin{align*}
    |\hat{U}^{(\ell-1)}_j \triangle \lazy(\hat{U}_j^{(\ell-1)})| & \;\le\; |\hat{U}^{(\ell-1)}_j \triangle \myhs U^{(\ell-1)}_j| + |U^{(\ell - 1)}_j \triangle \lazy(U^{(\ell-1)}_j)|\\
    &\;\le\; 3\epsilon \myhs |\hat{U}^{(\ell-1)}_j| \myhs + \myhs 2 \lambda \myhs \alpha \myhs k \log n  \log \frac{n}{k} \\
    &\;\le\; 3\epsilon \myhs |\hat{U}^{(\ell-1)}_j| \myhs + \myhs 2 \lambda \myhs |\hat{U}^{(\ell-1)}_j|\\
    &\;\le\; 3 \myhs (\lambda + \epsilon) \myhs |\hat{U}^{(\ell-1)}_j|. 
\end{align*}
The first inequality follows from~\cref{obs:sym_diff_prop}, the second from~\cref{lem:num_updates_phase_merge}, and the third from the fact that $|\hat{U}^{(\ell-1)}_j| \geq \alpha \myhs k \log n \log \frac{n}{k}$. The upper bound for $|\hat{U}^\ellexp_j \triangle \lazy(\hat{U}_j^\ellexp)|$ can be derived similarly.
\end{proof}

Furthermore, we obtain the analogs of~\cref{lem:rel_mu_old_curr} and \Cref{lem:rel_OPT_mu}, by setting $\tilde{\gamma} \coloneqq \frac{\gamma + 3(\lambda + \epsilon)}{1 - 3 (\lambda + \epsilon)}$ with $0 < \gamma < 1- 6(\lambda + \epsilon)$ and $0 < \lambda + \epsilon < \frac{1}{6}$. Putting everything together, \cref{thm:bicr_alg_merged} follows.

\begin{proof}[Proof of \cref{thm:bicr_alg_merged}]
    The size and the approximation ratio of the maintained bicriteria approximate solution~$\hat{S}$ follow from the earlier discussion.
    By the construction of the final dynamic algorithm, and since~\texttt{LazySync}$(\hspace{0.5pt})$ runs for $\frac{32\rho}{\lambda^2}$ steps per adversarial update (where $\lambda$ and $\rho$ are fixed constants), the claims on worst-case recourse and update time follow. 
\end{proof}

\antonis{Mention we can use the better $4$ approx and size of bicr.}
\section{Fully Dynamic $k$-Center Clustering with Near-Linear Update Time} \label{sec:near_linear_alg}
\antonis{What about the aspect ratio?}

Consider a point set $P$ with $|P| = n$ in an arbitrary metric space, undergoing adversarial point insertions and point deletions. The fully dynamic $k$-center algorithm by Forster and Skarlatos~\cite{forster_skarlatos2025} has a worst-case update time of $O(nk)$. In this section, we build upon their algorithm by introducing a sophisticated modification to its most challenging component, thereby achieving a worst-case update time of $O(n \log n)$, as stated in~\cref{thm:k_center_near_linear}.

\kcenternearlinear*

We denote by $\mathcal{A}$ the fully dynamic algorithm in~\cite{forster_skarlatos2025}, and by $\mathcal{F}$ our fully dynamic algorithm
with~$O(n \log n)$ worst-case update time. The most challenging component of $\mathcal{A}$ arises when the adaptive adversary deletes a center and $\mathcal{A}$ proceeds with its case C2 (see~Algorithm 1 in~\cite{forster_skarlatos2025}). To justify our claims, we provide descriptions of all components of the fully dynamic algorithm $\mathcal{A}$ in the respective subsections, but we refer the reader to~\cite{forster_skarlatos2025} for the full algorithm and further details.
To avoid repetition, while presenting the components of $\mathcal{A}$, we
also explain how our fully dynamic algorithm~$\mathcal{F}$ speeds up the corresponding steps of $\mathcal{A}$. For the most challenging component, we describe in detail the more intricate modifications incorporated into our fully dynamic $k$-center algorithm~$\mathcal{F}$.

\paragraph{State of the fully dynamic algorithm $\mathcal{F}$.}
Our fully dynamic algorithm $\mathcal{F}$ maintains a $k$-center solution~$S$ along with the corresponding clusters---which form a partition of the point set $P$---and also an estimate on the cluster radius $r^{(\delta)}$. In accordance with $\mathcal{A}$, the centers $c_i$ and clusters $C_i$ are labeled as either regular or extended (i.e., non-zombie) or zombie. Regular centers are responsible for points within distance $r^{(\delta)}$ while extended centers are also responsible for points at a greater distance. Similarly,  zombie centers can be responsible for points at a distance greater than $r^{(\delta)}$, as they may have been deleted and replaced. For every point $p \in P$, the algorithm $\mathcal{F}$ maintains two min-heaps $\mathcal{H}_p$ and $\mathcal{H}^{(nz)}_p$.
For a fixed point $p \in P$, the min-heap $\mathcal{H}_p$ stores the distances $\dist(p, c)$ for all centers $c \in S$, while the min-heap $\mathcal{H}^{(nz)}_p$ stores the distances $\dist(p, c)$ only for non-zombie centers $c \in S$.

\vspace{1em}
To establish the correctness of our fully dynamic $k$-center algorithm~$\mathcal{F}$, we only need the definition of a \emph{regular} center and cluster. The parameter $r^{(\delta)}$ introduced in this section is an internal radius used within the algorithm~$\mathcal{A}$. Since $r^{(\delta)}$ is used in the same way inside~$\mathcal{F}$, we treat it as a given without further elaboration.

\begin{definition}[regular centers and clusters] \label{def:reg}
    A cluster $C$ is  regular if and only if all points in $C$ lie within  distance $r^{(\delta)}$ from its unique center $c$. In turn, the center $c$ is also regular.
\end{definition}

We begin by describing how to speed up the most straightforward operations in~$\mathcal{A}$.
Afterwards, we provide a detailed explanation of the most challenging component and our key modification. Throughout the section, when we mention \emph{case C1}, \emph{case C2}, \emph{subcase C2a}, or \emph{subcase C2b}, we refer to the cases as presented in~\cite{forster_skarlatos2025}.

\subsection{Fully Dynamic Algorithms $\mathcal{A}$ and $\mathcal{F}$}
Our fully dynamic algorithm~$\mathcal{F}$ uses min-heaps to access distances efficiently. After each adversarial point update, the algorithm~$\mathcal{F}$ updates both $\mathcal{H}_p$ and $\mathcal{H}^{(nz)}_p$ appropriately for every point $p \in P$. In turn, these updates take $O(n \log n)$ time for all points.

\paragraph{Decreasing Operation.}
The Decreasing Operation is called multiple times within~$\mathcal{A}$ and is invoked after every adversarial point update. During the Decreasing Operation, three tasks are performed:
\begin{enumerate}
    \item First, if the maximum distance between any point and the solution $S$ is at most $r^{(\delta)}$ then all clusters become regular and the maintained radius is decreased appropriately (if possible).

    \item Second, any cluster containing only points within distance~$r^{(\delta)}$ is labeled as regular.

    \item Third, points farther than $r^{(\delta)}$ from their center but within distance $r^{(\delta)}$ of another non-zombie center are moved to the corresponding non-zombie cluster.
\end{enumerate}
\noindent The total time charged to this procedure is~$O(n \log n)$, since the clusters form a partition of the point set~$P$, and the distances $\dist(p, S)$ for every point $p \in P$ can be retrieved from $\mathcal{H}_p$ and $\mathcal{H}^{(nz)}_p$ in $O(\log n)$ time. The Decreasing Operation is handled the same way in our algorithm~$\mathcal{F}$ and requires $O(n \log n)$ time.

\paragraph{Point insertions.}
If the adversarially inserted point $p^+$ is within distance $r^{(\delta)}$ from a non-zombie center, it is simply  added to that cluster. Otherwise, there are two possible cases: 
\begin{itemize}
    \item If there are two centers $c_i, c_j$ with $i < j$ such that $\dist(c_i, c_j) \leq r^{(\delta)}$, then $c_i$ becomes responsible for the points of $c_j$, and $c_j$ is replaced as follows:
    \begin{itemize}
    \item If $p^+$ is within distance $r^{(\delta)}$ from a zombie center $c_z$, the point $c_z$ becomes the $j$-th center which is responsible for $p^+$ and all points from its old cluster $C_z$ within distance $r^{(\delta)}$; the updated~$c_j$ then becomes a regular center. For the remaining points in the updated $C_z$, the algorithm assigns the zombie $z$-th center based on the case C1 and case C2 (see~Algorithm 1 in~\cite{forster_skarlatos2025}), which we describe in~\cref{subsec:point_del}. 
    \item Otherwise, $p^+$ itself becomes a regular center, and specifically the updated $c_j$ is~located at~$p^+$.
    \end{itemize}
    \item Otherwise, the Increasing Operation is called: while there is a point $p \in P \setminus S$ such that pairwise distances in $S \cup \{p\}$ are greater than $r^{(\delta)}$, the estimate $r^{(\delta)}$ is increased, making all clusters regular. After the Increasing Operation, if the farthest point $p$ from $S$ is at a distance greater than $r^{(\delta)}$, then there must be two centers $c_i, c_j$ with $i < j$ such that $\dist(c_i, c_j) \leq r^{(\delta)}$. The center $c_i$ then becomes responsible for the points of $c_j$, and $c_j$ is replaced as a center by $p$. Otherwise, all points are close to some center, and~$p^+$ is added to the cluster of its nearest center.
\end{itemize}
The Increasing Operation is used in the same way in both fully dynamic algorithms $\mathcal{A}$ and $\mathcal{F}$. The adversarial point insertions in both algorithms are handled identically except for when a zombie center needs to be replaced as described above. The time required for such an adversarial point insertion relies on the same component as that of an adversarial point deletion and is analyzed in~\cref{subsec:point_del}. The time needed for any other adversarial point insertion is upper bounded by $O(n \log n)$. This is because the distances $\dist(p, S)$ (or $\dist(p, S^{(nz)})$, where $S^{(nz)}$ is the set of non-zombie centers) for every point $p \in P$ can be retrieved from $\mathcal{H}_p$ (or $\mathcal{H}^{(nz)}_p$) in $O(\log n)$ and there are $n$ points in total.

\subsubsection{Point Deletions} \label{subsec:point_del}
We begin with a short description of how algorithm $\mathcal{A}$ handles adversarial point deletions, and provide detailed explanations for the relevant cases of $\mathcal{F}$ later in this section.
If the adversarially deleted point~$p^-$ is not a center, the algorithm only removes it from its cluster and calls the Decreasing Operation. Otherwise,~$p^-$ is also removed from the set of centers $S$ and the corresponding cluster $C$ becomes (if it is not already) a zombie cluster. The Decreasing Operation is then called; next, there are two possible cases:
\begin{itemize}
    \item[] \textbf{Case C1.} If there is a point $p$ in $C$ whose distance from the remaining centers is greater than $r^{(\delta)}$, then~$p$ becomes the new center of cluster $C$.
    \item[] \textbf{Case C2.} Otherwise if there is no such point, the algorithm tries to detect a sequence of $2l-1$ points:
    \begin{equation} \label{seq_del}
        p_1, c_2, p_2, \dots, c_l, p_l, \tag{seq} 
    \end{equation}
     where $p_1$ is in the cluster of the deleted center $p^-$, such that for every index $i \in [2, l]:$
    \begin{itemize}
        \item $c_i$ is a zombie center with $\dist(p_{i-1}, c_i) \leq r^{(\delta)}$, and
        \item $p_i \in C_i$ with $\dist(p_i, c_i) > r^{(\delta)}$,
    \end{itemize}
    and also $\dist(p_l, S) > r^{(\delta)}$, where $p_l$ is added to the set of centers $S$.
    Subsequently, the indices of the centers are shifted with respect to the sequence (subcase C2a). If such a sequence does not exist, the points of the relevant zombie clusters are reassigned and then the relevant clusters become regular (subcase C2b).\footnote{We describe the Reassigning Operation below.}
\end{itemize}
When the deleted point is not a center, our algorithm $\mathcal{F}$ proceeds in the same way as $\mathcal{A}$. The size of each cluster is at most $n$, and the distance $\dist(p, S)$ for every point $p \in P$ can be retrieved from $\mathcal{H}_p$ in $O(\log n)$ time. As a result, case C1 in both~$\mathcal{A}$ and $\mathcal{F}$ requires $O(n \log n)$ time. The most expensive operations occur in case C2 of~$\mathcal{A}$, which according to Lemma 4.14 in~\cite{forster_skarlatos2025} takes $\Theta(nk)$ time. In the following paragraph, we examine this case in more detail and describe the changes within $\mathcal{F}$.

\paragraph{Modified Case C2.}
Let Seq$_\mathcal{A}$ be the set containing all possible sequences satisfying~(\ref{seq_del}). The authors in~\cite{forster_skarlatos2025} provide the time complexity for detecting such a sequence in their Lemma 4.14, which is based on converting the point set into an underlying directed graph $G = (V, E)$ constructed as follows. The vertex set is the set of points, namely $V = P$.
The edge set $E \coloneqq E_1 \cup E_2$ is the union of:
\begin{itemize}
    \item $E_1 \coloneqq \{(p, c) \in (P \setminus S) \times S \mid \dist(p, c) \leq r^{(\delta)}\}$, and
    \item $E_2 \coloneqq \{(c, p) \in S \times (P \setminus S) \mid p \text{ belongs to the cluster of $c$ } \textbf{and} \dist(p, c) > r^{(\delta)}\}$.
\end{itemize}
The attempt of finding a sequence thus can be viewed as a graph search problem on the constructed graph. 
We observe that the problematic edge subset is $E_1$, as each point $p \in P$ can potentially have $k$ edges connecting it to the $k$ centers, implying that $|E_1| = \Theta(nk)$. On the other hand, since the clusters form a partition of the point set $P$, it holds that $|E_2| = O(n)$. Essentially, this is the part of the algorithm~$\mathcal{A}$ that leads to the $\Theta(nk)$ update time bound.

\subparagraph{Our sequence and underlying graph.}
    The crucial structural observation we make is that it suffices for a point to select an arbitrary (zombie) center. Namely in our fully dynamic algorithm~$\mathcal{F}$, the edge subset~$E_1$ is \emph{implicitly} constructed so that any point $p$ appears in at most one tuple in $E_1$. In particular, our fully dynamic algorithm~$\mathcal{F}$ constructs \emph{implicitly} an underlying directed graph $G = (V, E = E_1 \cup E_2)$ such that:
    \begin{itemize}
        \item for each point $p$, at most one center $c$ within distance $r^{(\delta)}$ is selected,
        and the edge $(p, c)$ is implicitly added to~$E_1$,
        \item for every center $c$ and every point $p$ inside its cluster with $\dist(p, c) > r^{(\delta)}$, the edge $(c, p)$ is implicitly added to $E_2$.
    \end{itemize}
    As a consequence, our underlying graph consists of $O(n)$ edges instead of $\Theta(nk)$. Let Seq$_\mathcal{F}$ be the subset of~Seq$_\mathcal{A}$ containing all such sequences induced by our underlying graph. Whenever case C2 inside~$\mathcal{F}$ attempts to detect a sequence in Seq$_\mathcal{F}$, our underlying graph is constructed efficiently as described in~\cref{alg:find_sequence}. Specifically, the procedure~\texttt{FindSequence}$(C, p^-, \emptyset)$ either finds a sequence in~Seq$_\mathcal{F}$ or determines that there is no such sequence in $O(n \log n)$ time.
    
    \begin{algorithm}[H]\footnotesize
\algnewcommand{\LineComment}[1]{\State \(\triangleright\) #1}
\algrenewcommand\algorithmiccomment[1]{\hspace{1em} \(\triangleright\) #1}
\caption{\textsc{Find Sequence and Case C2a}{}}\label{alg:find_sequence}

\begin{algorithmic}[1]

\Function{FindSequence}{$C, c, \textit{Seq}_\text{act}$} 
    \State $\text{visited}[c] \gets 1$
    \State $\textit{Seq}_\text{act} \gets \textit{Seq}_\text{act} \cup \{c\}$ \Comment{Maintains points on the currently explored path}
    \vspace{0.3em}
    
    \For{$p \in C$ with $\dist(p, c) > r^{(\delta)}$} \label{line:for_far_point}
        \State Add $p$ to $\text{explore}[c]$ \Comment{The implicit edge subset $E_2$}
        \vspace{0.2em}
        \State $\hat{c} \gets \mathcal{H}_p.\textit{getMin}()$ \Comment{$\hat{c}$ is a zombie center because the Decreasing Operation has been called}
        \State Let $\hat{C}$ be the cluster of $\hat{c}$
        \vspace{0.2em}
        
        \If{$\dist(p, \hat{c})$ > $r^{(\delta)}$} \Comment{Found free point to serve as $p_l$}
            \State \Call{SchiftSequence}{$\textit{Seq}_\text{act} \cup \{p\}, p$}
            \State \textbf{exit} from \Call{FindSequence}{$\hspace{0.5pt}$}
        \Else
        \State $\text{blocked}[p] \gets \hat{c}$ \Comment{The implicit edge subset $E_1$} \label{line:blocked_p}
        \vspace{0.3em}
        \If{$\text{visited}[\hat{c}] \neq 1$} \label{line:stop_exploration}
            \State \Call{FindSequence}{$\hat{C}, \hat{c}, \textit{Seq}_\text{act} \cup \{p\}$}
        \EndIf
        \EndIf
    \EndFor
\EndFunction

\vspace{1em}

\Function{ShiftSequence}{$\textit{Seq}, p$} \Comment{Case C2a}
    \State Let $p_1, c_2, p_2, \dots, c_l, p_l$ be the points in $\textit{Seq}$, where $p^- = c_1$ is the adversarially  deleted center and $p = p_l$
    \vspace{0.3em}
    \For{$j \in \{1, \dots, l-1\}$}
        \State $c_j \gets c_{j+1}$ 
    \EndFor
    \vspace{0.2em}
    \State $c_l \gets p$
\EndFunction

\end{algorithmic}
\end{algorithm}

    \subparagraph{Subcases C2a and C2b.}
        If our fully dynamic algorithm~$\mathcal{F}$ identifies a sequence in Seq$_\mathcal{F}$ then 
        continues with subcase C2a, and the
        running time is $O(n \log n)$.
        Otherwise, our fully dynamic algorithm~$\mathcal{F}$ continues with subcase C2b, where
        the Reassigning Operation is executed.
        In this scenario, observe that for every point in the cluster of the deleted center $p^-$, there is another (zombie) center within distance $r^{(\delta)}$ that can be regular. 
        During case C2b, the fully dynamic algorithm~$\mathcal{A}$ needs access to the furthest point from the set of centers $S$. This information can be maintained in $O(n \log n)$ time, using the heaps and by scanning all points.

\subparagraph{Reassigning Operation.} 
    During the Reassigning Operation, the fully dynamic algorithm~$\mathcal{A}$ uses a queue~$Q$ that contains the points scanned during the exploration of the sequence. The algorithm~$\mathcal{A}$ then adds each point $p \in Q$ to the cluster whose center lies within distance $r^{(\delta)}$ from $p$. Finally, the algorithm~$\mathcal{A}$ converts each cluster involved in exploring Seq$_\mathcal{A}$ into a regular cluster.
    As a result, all the points of the cluster of the deleted center $p^-$, move to a different regular cluster. 
    
    Our fully dynamic algorithm~$\mathcal{F}$ behaves in the same way, by reassigning all the points scanned during the exploration of sequences in Seq$_\mathcal{F}$. The Reassigning Operation in~$\mathcal{F}$ utilizes the underlying graph constructed in \texttt{FindSequence}$(\hspace{0.5pt})$, and it is described in \Cref{alg:reassign}. Notice that if no sequence is found by \texttt{FindSequence}$(\hspace{0.5pt})$ (i.e., Seq$_\mathcal{F} = \emptyset$), then all centers visited during its execution are assigned points within distance $r^{(\delta)}$, and thus qualify as regular (see~\cref{def:reg}). The Reassigning Operation converts these centers to regular, and it runs in~$O(n \log n)$ time.  
    
\begin{algorithm}[H]\footnotesize
\algnewcommand{\LineComment}[1]{\State \(\triangleright\) #1}
\algrenewcommand\algorithmiccomment[1]{\hspace{1em} \(\triangleright\) #1}
\caption{\textsc{Reassigning Operation}{}}\label{alg:reassign}

\begin{algorithmic}[1]

\Function{ReassignPoints}{$c$}
    \State $\text{visited}[c] \gets 1$

    \For{$p \in \text{explore}[c]$}
        \State $\hat c \gets \text{blocked}[p]$ \label{algline:blocked_p_c}
        \State Assign $p$ to the cluster of center $\hat c$
        \If{$\text{visited}[\hat{c}] \neq 1$}
            \State \Call{ReassignPoints}{$\hat c$} 
        \EndIf
    \EndFor
\EndFunction

\end{algorithmic}
\end{algorithm}

\subsection{Analysis of Our Fully Dynamic Algorithm~$\mathcal{F}$} \label{sec:analysis_fully_alg_F}

We remind the reader that during case C2, if no sequence is found by $\mathcal{F}$ (i.e., Seq$_\mathcal{F} = \emptyset$) then $\mathcal{F}$ proceeds with the Reassigning Operation. Since Seq$_\mathcal{F} \neq$ Seq$_\mathcal{A}$, we must provide a correctness argument for the Reassigning Operation, analogous to Claim 4.13 in~\cite{forster_skarlatos2025}. Furthermore, we prove that the centers visited during \texttt{FindSequence}$(\hspace{0.5pt})$ are indeed qualified as regular centers.

\paragraph{Comparison of the sequences in~$\mathcal{A}$ and~$\mathcal{F}$.}
By the construction of both types of sequences, it follows that
Seq$_\mathcal{F} \subseteq$ Seq$_\mathcal{A}$. In other words, whenever $\mathcal{F}$ finds a sequence, then $\mathcal{A}$ finds a sequence as well. We remark that based on the correctness of $\mathcal{A}$ in~\cite{forster_skarlatos2025}, any sequence in Seq$_\mathcal{A}$ can be employed during subcase~C2a. Hence whenever $\mathcal{F}$ finds a sequence in Seq$_\mathcal{F}$, both algorithms $\mathcal{A}$ and~$\mathcal{F}$ operate in the same way during subcase C2a.

However, it is possible that~\texttt{FindSequence}$(\hspace{0.5pt})$ in $\mathcal{F}$ cannot find some sequences in Seq$_\mathcal{A}$ that $\mathcal{A}$ can detect. Consequently, we must argue the correctness of $\mathcal{F}$ in the scenario where $\mathcal{A}$ proceeds with subcase~C2a, but $\mathcal{F}$ continues with subcase C2b. Namely, it is likely that $\mathcal{F}$ executes the Reassigning Operation during subcase C2b, while $\mathcal{A}$ continued with subcase C2a.

This is a subtle situation for two reasons: First, we must prove that the Reassigning Operation performs correctly in our scenario. Second, we must show that the correctness analysis in~\cite{forster_skarlatos2025} still applies to the resulting state of all centers and clusters. The intuitive reason is that if Seq$_\mathcal{F} = \emptyset$ then all points scanned by~\texttt{FindSequence}$(\hspace{0.5pt})$ are within distance $r^{(\delta)}$ from a center that can become regular. One conceptual difference between the fully dynamic algorithm $\mathcal{A}$ from~\cite{forster_skarlatos2025} and our fully dynamic algorithm~$\mathcal{F}$ is the following observation.
\begin{observation}
    It suffices for a point to be assigned to a potentially regular cluster whose points, in turn, can be assigned to another potentially regular cluster.
\end{observation}

\begin{lemma}\label{lem:close_center}
    During the Reassigning Operation in~$\mathcal{F}$, for every point which is encountered during \texttt{FindSequence}$(\hspace{0.5pt})$, there is a center within distance $r^{(\delta)}$ from it.
\end{lemma}
\begin{proof}
The procedure~\texttt{FindSequence}$(\hspace{0.5pt})$ terminates because all scanned centers have been visited.
Since \texttt{ReassignPoints}$(\hspace{0.5pt})$ has been called,
no sequence has been found during~\texttt{FindSequence}$(\hspace{0.5pt})$. Hence, whenever a center $c$ is reached during \texttt{FindSequence}$(\hspace{0.5pt})$, the for loop in~\linecref{line:for_far_point} of~\cref{alg:find_sequence} ensures that all points $p$ in the cluster of $c$ with $\dist(c,p) > r^{(\delta)}$ are explored.

Suppose to the contrary that there is a point $p$ encountered during~\texttt{FindSequence}$(\hspace{0.5pt})$ (in~\cref{alg:find_sequence}) such that $\dist(p, S) > r^{(\delta)}$ from its closest (zombie) center. In this case, the point $p$ could serve as $p_l$ in~(\ref{seq_del}) of~Seq$_\mathcal{F}$ (or case C1 can be used), contradicting the fact that~\texttt{ReassignPoints}$(\hspace{0.5pt})$ has been invoked. Therefore, there is a center within distance~$r^{(\delta)}$ from every point during~\texttt{ReassignPoints}$(\hspace{0.5pt})$ (i.e.,~\linecref{algline:blocked_p_c} in~\cref{alg:reassign} is executed properly).
\end{proof}

\begin{lemma} \label{lem:close_points_to_center}
    Once the Reassigning Operation in~$\mathcal{F}$ terminates, every center encountered during \texttt{FindSequence}$(\hspace{0.5pt})$ is assigned only points within distance $r^{(\delta)}$ from it.
\end{lemma}
\begin{proof}
By~\cref{lem:close_center}, for every point $p$ explored during~\texttt{FindSequence}$(\hspace{0.5pt})$ there is a center~$\hat{c}$ such that $\dist(p,\hat{c}) \leq r^{(\delta)}$, and $\text{blocked}[p]$ is set to $\hat{c}$  in~\linecref{line:blocked_p} of~\cref{alg:find_sequence}. Based on the information stored in~\text{explore}$[\cdot]$ and $\text{blocked}[\cdot]$,
each point~$p$ in the cluster of~$c$ with $\dist(c,p) > r^{(\delta)}$ is assigned to the cluster of~$\hat{c}$ with $\dist(p,\hat{c}) \leq r^{(\delta)}$ (see~\linecref{algline:blocked_p_c} in~\cref{alg:reassign}). Thus, for every center $c$ visited during \texttt{FindSequence}$(\hspace{0.5pt})$, all points in its updated cluster lie within distance~$r^{(\delta)}$ of $c$. 
\end{proof}

Based on~\cref{lem:close_center}, for every point in the cluster of the deleted center, there is another center within distance $r^{(\delta)}$. These centers can be labeled as regular based on~\cref{lem:close_points_to_center}. Therefore, the correctness of the fully dynamic algorithm $\mathcal{F}$ holds, and since $\mathcal{F}$ requires $O(n \log n)$ update time per adversarial point update,~\cref{thm:k_center_near_linear} follows.

\subsection{Putting Everything Together} \label{sec:put_together_fully_kcenter}
The following lemma describes how we can combine our two subroutines from~\cref{thm:bicr_alg_merged,thm:k_center_near_linear}.

\begin{lemma}[Theorem 3.1 in \cite{bhattacharya2025alm_opt_kcenter}] \label{lem:comb_two_subrout}
For the $k$-center clustering problem, consider two subroutines:
\begin{itemize}
    \item a fully dynamic  ($\alpha$, $\beta$)-bicriteria approximation algorithm with a worst-case update time of $T_S(n)$ and a worst-case recourse of $R_S(n)$, and
    \item a fully dynamic $\rho$-approximation algorithm with a worst-case update time of~$T_A(n)$ and a worst-case recourse of~$R_A(n)$, where $n$ is the number of points. 
\end{itemize}
Then there is a fully dynamic $(\alpha + 2\rho)$-approximation algorithm for the $k$-center clustering problem with a worst-case update time of $O(T_S(n) + R_S(n) \cdot T_A(\beta \myhs k))$ and a worst-case recourse of $O(R_S(n) \cdot R_A(\beta \myhs k))$.
\end{lemma}

\noindent
Therefore, using~\cref{lem:comb_two_subrout,thm:bicr_alg_merged,thm:k_center_near_linear}, we obtain our main result.

\mainresult*

\newpage
\printbibliography[heading=bibintoc] 


\end{document}